%% file: main.tex
\def\llncs{0}
\def\fullpage{1}
\def\anonymous{0}
\def\authnote{0}
\def\notxfont{0}
\def\submission{0}%when submit 30page version to conference, it is 1. For full version, it is 0.
\def\cameraready{0}

\ifnum\submission=1
\def\llncs{1}
\fi

\ifnum\llncs=1
	\documentclass[envcountsect,a4paper,runningheads,10pt]{llncs}%default is 10pt
\else
	\documentclass[letterpaper,hmargin=1.05in,vmargin=1.05in,11pt]{article}
			\ifnum\fullpage=1
		\usepackage{fullpage}
		\fi
\fi

\input{preamble_usepackages.tex}
\input{macros}

\usepackage{breakcites}

\title{Quantum Advantage from One-Way Functions}

\ifnum\anonymous=1
\ifnum\llncs=1
\author{\empty}\institute{\empty}
\else
\author{}
\fi
\else
\ifnum\llncs=1
\author{
	Tomoyuki Morimae\inst{1} \and Takashi Yamakawa\inst{2,3,1}
}
\institute{
	Yukawa Institute for Theoretical Physics, Kyoto University, Kyoto, Japan \and NTT Social Informatics Laboratories, Tokyo, Japan
  \and NTT Research Center for Theoretical Quantum Information, Atsugi, Japan
}
\else
\author[1]{Tomoyuki Morimae}
\author[2,3,1]{Takashi Yamakawa}
\affil[1]{{\small Yukawa Institute for Theoretical Physics, Kyoto University, Kyoto, Japan}\authorcr{\small tomoyuki.morimae@yukawa.kyoto-u.ac.jp}}
\affil[2]{{\small NTT Social Informatics Laboratories, Tokyo, Japan}\authorcr{\small takashi.yamakawa.ga@hco.ntt.co.jp}}
\affil[3]{{\small NTT Research Center for Theoretical Quantum Information, Atsugi, Japan}}

\fi %%%%% END OF LNCS branch
\fi

\date{}

\begin{document}

\maketitle

\begin{abstract}
Is quantum computing truly faster than classical computing?
Demonstrating unconditional quantum computational advantage lies beyond the reach of the current complexity theory, and therefore
we have to rely on some complexity assumptions.
While various results on quantum advantage have been obtained, all 
necessitate
relatively stronger or less standard assumptions
in complexity theory or classical cryptography.
In this paper, we show quantum advantage based on several fundamental assumptions, specifically relying solely on the existence of classically-secure one-way functions. 
Given the fact that one-way functions are necessary for almost all classical cryptographic primitives,
our findings yield a surprising implication:
{\it if there is no quantum advantage, then there is no classical cryptography!}
More precisely, we introduce {\it inefficient-verifier proofs of quantumness} (IV-PoQ), and construct it from statistically-hiding and computationally-binding classical bit commitments.
IV-PoQ is an interactive protocol between a verifier and a quantum polynomial-time prover consisting of two phases.
In the first phase, the verifier is classical probabilistic polynomial-time, and it interacts 
with the quantum polynomial-time prover over a classical channel.
In the second phase, the verifier becomes inefficient, and makes its 
decision based on the transcript of the first phase. If the quantum prover is honest, the inefficient verifier accepts with high probability, but any classical probabilistic polynomial-time malicious prover only has a small probability of being accepted by the inefficient verifier.
In our construction, the inefficient verifier can be a classical deterministic polynomial-time algorithm that queries an 
$\mathbf{NP}$ oracle.
Our construction demonstrates the following results based on the known constructions of statistically-hiding and computationally-binding commitments from one-way functions or distributional collision-resistant hash functions:
\begin{itemize}
\item
If one-way functions exist, then IV-PoQ exist.
\item
If distributional collision-resistant hash functions exist (which exist if hard-on-average problems in $\mathbf{SZK}$ exist),
then constant-round IV-PoQ exist. 
\end{itemize}
We also demonstrate quantum advantage based on worst-case-hard assumptions.
We define {\it auxiliary-input IV-PoQ} (AI-IV-PoQ) that only require 
that for any malicious prover, there exist infinitely many auxiliary inputs under which the prover cannot cheat.   
We construct AI-IV-PoQ from an auxiliary-input version of commitments in a similar way, showing that
\begin{itemize}
\item
If auxiliary-input one-way functions exist (which exist if $\mathbf{CZK}\not\subseteq\mathbf{BPP}$), then AI-IV-PoQ exist.
\item 
If auxiliary-input collision-resistant hash functions exist (which is equivalent to $\mathbf{PWPP}\nsubseteq \mathbf{FBPP}$) 
or $\mathbf{SZK}\nsubseteq \mathbf{BPP}$,
then constant-round AI-IV-PoQ exist.
\end{itemize}
Finally, we also show that some variants of PoQ can be constructed from quantum-evaluation one-way functions (QE-OWFs), which are similar to classically-secure classical one-way functions except that the evaluation algorithm is not classical but quantum.
QE-OWFs appear to be weaker than classically-secure classical one-way functions, and therefore it demonstrates quantum advantage based on assumptions even weaker than one-way functions.
\end{abstract}

\ifnum\submission=1
\else
\clearpage
\newpage
\setcounter{tocdepth}{2}
\tableofcontents
\fi

\newpage

\input{introduction}
\input{Preliminaries}

\input{interactive_supermacy}

\input{coherent}
\input{construction}

\input{completeness}

\input{soundness}

\input{two_round}
\input{QOWFs}

\ifnum\anonymous=1
\else
{\bf Acknowledgements.}
We thank 
\ifnum\cameraready=1
\else
Mark Zhandry for pointing out \cref{rem:why_not_quantum} and 
\fi
Kai-Min Chung for pointing out the relationship between  \Cref{lem:hash2} and Valiant-Vazirani theorem. 
TM is supported by
JST CREST JPMJCR23I3,
JST Moonshot R\verb|&|D JPMJMS2061-5-1-1, 
JST FOREST, 
MEXT QLEAP, 
the Grant-in-Aid for Scientific Research (B) No.JP19H04066, 
the Grant-in Aid for Transformative Research Areas (A) 21H05183,
and 
the Grant-in-Aid for Scientific Research (A) No.22H00522.
\fi

\ifnum\submission=1 
\else
\appendix
\input{PPneqBPP}
\input{omitted_preliminaries}
\input{app_completeness}
\input{app_dOWFs}
\fi

\ifnum\submission=0
\bibliographystyle{alpha} 
\else
\bibliographystyle{splncs04}
\fi
\bibliography{abbrev3,crypto,reference}

\ifnum\submission=1
\ifnum\cameraready=1
\else
\appendix

\input{two_round}
\input{QOWFs}
\input{PPneqBPP}
\input{omitted_preliminaries}
\input{IO_and_AI_PoQ}
\input{gap_amplification}

\input{app_completeness}
\input{V2}
\fi
\fi

\end{document}

%% file: preamble_usepackages.tex
\usepackage[%
  colorlinks=true,
  citecolor=blue,
  pagebackref=true
]{hyperref}

\usepackage{amsmath, amsfonts, amssymb, mathtools,amscd}

\usepackage{amsthm}

\usepackage{lmodern}
\usepackage[T1]{fontenc}
\usepackage[utf8]{inputenc}

\usepackage{arydshln} % In order to use \hdashline
\usepackage{url}
\usepackage{ifthen}
\usepackage{bm}
\usepackage{multirow}
\usepackage[dvips]{graphicx}
\usepackage[usenames]{color}
\usepackage{xcolor,colortbl} % Leave out in case the usepackage cannot be found. Not critical
\usepackage{threeparttable}
\usepackage{comment}
\usepackage{paralist,verbatim}
\usepackage{cases}
\usepackage{booktabs}
\usepackage{braket}
\usepackage{cancel} 
\usepackage{ascmac} 
\usepackage{framed}
\usepackage{authblk}
\usepackage{pifont}
\usepackage{qcircuit}
\definecolor{darkblue}{rgb}{0,0,0.6}
\definecolor{darkgreen}{rgb}{0,0.5,0}
\definecolor{maroon}{rgb}{0.5,0.1,0.1}
\definecolor{dpurple}{rgb}{0.2,0,0.65}

\allowdisplaybreaks
\usepackage[capitalise,noabbrev]{cleveref}
\usepackage[absolute]{textpos}
\usepackage[final]{microtype}
\usepackage[absolute]{textpos}
\usepackage{everypage}
\DeclareMathAlphabet{\mathpzc}{OT1}{pzc}{m}{it}

\usepackage{algorithmic}
\usepackage{algorithm}
\usepackage{here}

\newtheoremstyle{thicktheorem}%
{\topsep}
{\topsep}
{\itshape}{}%
{\bfseries}%
{.}
{ }%
{\thmname{#1}\thmnumber{ #2}%
		\thmnote{ (#3)}%
}

\newtheoremstyle{remark}%name
{\topsep}
{\topsep}
	{}%body font
	{}%indent amount
	{}%theorem head font
	{.}%punctuation after theorem head
	{ }%space after theorem head
	{\textit{\thmname{#1}}\thmnumber{ #2}%theorem head specs
			\thmnote{ (#3)}%
	}

\ifnum\llncs=0
	\theoremstyle{thicktheorem}
	\newtheorem{theorem}{Theorem}[section]
	\newtheorem{lemma}[theorem]{Lemma}

	\newtheorem{definition}[theorem]{Definition}

	\theoremstyle{remark}
	
	\newtheorem{remark}[theorem]{Remark}

\else
\fi

\Crefname{MyClaim}{Claim}{Claims}
	\crefname{theorem}{Theorem}{Theorems}
	\crefname{assumption}{Assumption}{Assumptions}
	\crefname{construction}{Construction}{Constructions}
	\crefname{corollary}{Corollary}{Corollaries}
	\crefname{conjecture}{Conjecture}{Conjectures}
	\crefname{definition}{Definition}{Definitions}
	\crefname{exmaple}{Example}{Examples}
	\crefname{experiment}{Experiment}{Experiments}
	\crefname{counterexample}{Counterexample}{Counterexamples}
	\crefname{lemma}{Lemma}{Lemmata}
	\crefname{observation}{Observation}{Observations}
	\crefname{proposition}{Proposition}{Propositions}
	\crefname{remark}{Remark}{Remarks}
	\crefname{claim}{Claim}{Claims}
	\crefname{fact}{Fact}{Facts}
	\crefname{note}{Note}{Notes}

\ifnum\llncs=1
 \crefname{appendix}{App.}{Appendices}
 \crefname{section}{Sec.}{Sections}
\else
\fi

\ifnum\llncs=1
\renewcommand*{\backref}[1]{}
\else
	\renewcommand*{\backref}[1]{(Cited on page~#1.)}
	\ifnum\notxfont=1
	\else
		\usepackage{newtxtext}
	\fi
\fi

%% file: macros.tex
\ifnum\authnote=0  %%%%% Remove comments %%%%%
\newcommand{\mor}[1]{}
\newcommand{\minki}[1]{}
\newcommand{\takashi}[1]{}

\else
\newcommand{\mor}[1]{$\ll$\textsf{\color{red} Tomoyuki: { #1}}$\gg$}
\newcommand{\takashi}[1]{$\ll$\textsf{\color{orange} Takashi: { #1}}$\gg$}
\newcommand{\minki}[1]{$\ll$\textsf{\color{darkgreen} Minki: { #1}}$\gg$}

\fi

\newcommand{\Nseq}{N\text{-}\mathsf{seq}}
\newcommand{\primitive}{\mathtt{P}}

\newcommand{\good}{\mathsf{Good}}
\newcommand{\verygood}{\mathsf{VeryGood}}

\newcommand{\BQP}{\mathbf{BQP}}

\newcommand{\NP}{\mathbf{NP}}

\newcommand{\seteq}{\coloneqq}

\newcommand{\cA}{\mathcal{A}}
\newcommand{\cB}{\mathcal{B}}
\newcommand{\cC}{\mathcal{C}}

\newcommand{\cP}{\mathcal{P}}

\newcommand{\cR}{\mathcal{R}}
\newcommand{\cS}{\mathcal{S}}

\newcommand{\cV}{\mathcal{V}}

\newcommand{\cX}{\mathcal{X}}
\newcommand{\cY}{\mathcal{Y}}

\def\makeuppercase#1{
\expandafter\newcommand\csname tl#1\endcsname{\widetilde{#1}}
}

\def\makelowercase#1{
\expandafter\newcommand\csname tl#1\endcsname{\widetilde{#1}}
}

\newcommand{\regB}{\mathbf{B}}
\newcommand{\regA}{\mathbf{A}}

\newcommand{\secp}{\lambda}

\newcommand{\A}{\entity{A}}
\newcommand{\B}{\entity{B}}

\newcommand*{\entity}[1]{\mathcal{#1}}
\newenvironment{boxfig}[2]{\begin{figure}[#1]\fbox{\begin{minipage}{0.97\linewidth}
                        \vspace{0.2em}
                        \makebox[0.025\linewidth]{}
                        \begin{minipage}{0.95\linewidth}
            {{
                        #2 }}
                        \end{minipage}
                        \vspace{0.2em}
                        \end{minipage}}}{\end{figure}}

\newcommand{\bit}{\{0,1\}}

\newcommand{\negl}{{\mathsf{negl}}}

\newcommand{\poly}{{\mathrm{poly}}}

\makeatletter
\DeclareRobustCommand
  \myvdots{\vbox{\baselineskip4\p@ \lineskiplimit\z@
    \hbox{.}\hbox{.}\hbox{.}}}
\makeatother

%% file: introduction.tex
\section{Introduction}
\label{sec:introduction}
Is quantum computing truly faster than classical computing? Demonstrating unconditional quantum computational advantage lies beyond the reach of the current complexity theory,
and therefore we have to rely on some complexity assumptions. While various results on quantum
advantage have been obtained, all necessitate relatively stronger or less standard assumptions in
complexity theory or classical cryptography.\footnote{Another important goal of quantum computational advantage is experimental implementation. 
In such a scenario,
priority would be given to the experimental feasibility (such as simplicity of quantum circuits, etc.) rather than the reliability of underlying complexity assumptions. 
Our interest in this paper is purely theoretical one, and our objective is to identify the weakest complexity assumption for quantum advantage assuming any polynomial-time quantum computing is possible.}
Let us first summarize previous approaches. (See also \cref{table:QA}.)

\paragraph{\bf Approach 1: Sampling.}
One approach to demonstrate quantum advantage is
the sampling-based one.
In the sampling-based quantum advantage, quantum polynomial-time (QPT) algorithms can
sample from certain probability distributions but no classical probabilistic polynomial-time (PPT) algorithm can.
A great merit of the approach is that
relatively simple quantum computing models are enough, such as the Boson Sampling model~\cite{STOC:AarArk11}, 
the IQP model~\cite{BreJozShe10}, 
the random circuit model~\cite{NatPhys:BFNV19},
and the one-clean-qubit model~\cite{FKMNTT18}.\footnote{The Boson Sampling model is a quantum computing model that uses non-interacting
bosons, such as photons. The IQP (Instantaneous Quantum Polytime) model is a quantum computing model
where only commuting quantum gates are used.
The random circuit model is a quantum computing model where each gate is randomly chosen.
The one-clean-qubit model is a quantum computing model where the input is
$|0\rangle\langle0|\otimes\frac{I^{\otimes m}}{2^m}$.}
Output probability distributions of these restricted quantum computing models
cannot be sampled by any PPT algorithm within a constant multiplicative error\footnote{We say that the output probability distribution of a quantum algorithm is sampled by a classical
algorithm within a constant multiplicative error $\epsilon$ if
$|q_z-p_z|\le\epsilon p_z$ is satisfied for all $z$, where $q_z$ is the probability that
the quantum algorithm outputs the bit string $z$, and $p_z$ is the probability that
the classical algorithm outputs the bit string $z$.}
unless the polynomial-time
hierarchy collapses to the third~\cite{STOC:AarArk11,BreJozShe10} or the second level~\cite{FKMNTT18}.\footnote{\cite{TD04} previously showed that output probability distributions
of constant-depth quantum circuits
cannot be sampled classically unless $\mathbf{BQP}\subseteq\mathbf{AM}$. Their assumption can be easily improved to the assumption that the polynomial-time hierarchy
does not collapse to the second level.}
The assumption that the polynomial-time hierarchy does not collapse is 
a widely-believed assumption in classical complexity theory, but one disadvantage of these results is that
the multiplicative-error sampling is unrealistic.
The requirement of the multiplicative-error sampling can be relaxed to that of the constant additive-error 
sampling~\cite{STOC:AarArk11,BreMonShe16,Morimae17,NatPhys:BFNV19},\footnote{We say that the output probability distribution of a quantum algorithm is sampled by a classical
algorithm within a constant additive error $\epsilon$ if
$\sum_z|q_z-p_z|\le\epsilon$ is satisfied, where $q_z$ is the probability that
the quantum algorithm outputs the bit string $z$, and $p_z$ is the probability that
the classical algorithm outputs the bit string $z$.}
but the trade-off is that
the underlying classical complexity assumptions become less standard:
some new assumptions about average-case $\verb|#|\mathbf{P}$-hardness of some problems,
which were not studied before, 
have to be introduced.

Another disadvantage of the sampling-based approach is that it is not known to be verifiable.
For the multiplicative-error case, we do not know how to verify quantum advantage even with
a computationally-unbounded verifier.
Also for the additive-error case, we do not know how to verify the quantum advantage efficiently.
(For example, there is a negative result that suggets that exponentially-many samples are necessary to verify
the correctness of the sampling~\cite{HKEG19}.)
At least, we can say that
if there exists a sampling-based quantum advantage in the additive-error case,
there exists an inefficiently-verifiable
quantum advantage for a certain search problem~\cite{Aar14}.\footnote{\cite[Theorem 21]{Aar14} 
showed that if 
there exists an additive-error sampling problem that is quantumly easy but classically hard,
then there exists a search problem that is quantumly easy but classically hard.
The relation of the search problem is verified inefficiently.
Note that the search problem depends on the time-complexity of the classical adversary, and therefore 
it is incomparable to our (AI-)IV-PoQ.}
\if0
footnote{Note that this is true only for the additive-error sampling. Quantum advantage in the multiplicative-error setting does not imply $\mathbf{SampBQP}\neq \mathbf{SampBPP}$ (because these classes are defined w.r.t. additive errors). In particular, it is not known that $\mathbf{SampBQP}\neq \mathbf{SampBPP}$ follows from non-collapsing of PH (and there is a negative evidence for that \cite{CCC:AarChe17}: there is an oracle relative to which $\mathbf{SampBQP}=\mathbf{SampBPP}$ and PH is infinite).}
\fi

\paragraph{\bf Approach 2: Search problems.}
Some inefficiently-verifiable search problems that exhibit quantum advantage have been introduced.
For example, 
for the random circuit model,
\cite{CCC:AarChe17,AarGun19} introduced so-called Heavy Output Generation (HOG)
and Linear Cross-Entropy Heavy Output Generation (XHOG)
where given a quantum circuit $C$
it is required to output bit strings 
that satisfy certain relations about $C$.
The relations can be verified inefficiently.  
The classical hardnesses of these problems
are, however, based on new assumptions introduced by the authors.
\cite{STOC:Aaronson10} constructed an inefficiently-verifiable search problem (Fourier Fishing),
but its quantum advantage is relative to random oracles. 
\cite{ACCGSW22} constructed another inefficiently-verifiable
search problem (Collision Hashing), but its quantum advantage is also relative to random oracles.

\paragraph{\bf Approach 3: Proofs of quantumness.}
There is another approach of demonstrating quantum advantage where the verification is efficient, namely, proofs of quantumness (PoQ) \cite{JACM:BCMVV21}.
In PoQ, we have a QPT prover and a PPT verifier.
They interact over a classical channel, and the verifier finally makes the decision.
If the QPT prover behaves honestly, the verifier accepts with high probability,
but for any malicious PPT prover, the verifier accepts with only small probability.
The simplest way of realizing PoQ is to let the prover solve an $\mathbf{NP}$ problem
that is quantumly easy but classically hard, such as 
factoring~\cite{FOCS:Shor94}. 
Such a simplest way is, however, based on specific assumptions that certain specific problems are hard for PPT algorithms.

The first construction of PoQ based on a general assumption was given in \cite{JACM:BCMVV21}
where (noisy) trapdoor claw-free functions with the adaptive-hardcore-bit property\footnote{The adaptive-hardcore-bit property very roughly means that
it is hard to find $x_b$ ($b\in\bit$) and $d\neq \mathbf{0}$ such that $f_0(x_0)=f_1(x_1)$
and $d\cdot(x_0\oplus x_1)=0$, given a claw-free pair $(f_0,f_1)$.}
is assumed.
Such functions can be instantiated with the LWE assumption, for example \cite{JACM:BCMVV21}. 
The adaptive-hardcore-bit property was removed in \cite{NatPhys:KMCVY22}, where
only trapdoor 2-to-1 collision-resistant hash functions are assumed.
In \cite{ITCS:MorYam23},
PoQ was constructed from (full-domain) trapdoor permutations. 
PoQ can also be constructed from
quantum homomorphic encryptions (QHE) \cite{cryptoeprint:2022/400} for a
certain class of quantum operations (such as controlled-Hadamard gates), which can be instantiated with the LWE assumption~\cite{FOCS:Mahadev18b}. 
These contructions are interactive, i.e., the verifier and the prover have to exchange
many rounds of messages.
Recently, a 
non-interactive PoQ has been realized with only random oracles \cite{FOCS:YamZha22}.
This result demonstrates efficiently-verifiable quantum advantage with an ``unstructured'' problem for the first time. 
However, it is known that hardness relative to a random oracle does not necessarily imply hardness in the unrelativized world where the random oracle is replaced with a real-world hash function~\cite{CGH04}.  Thus,  \cite{FOCS:YamZha22} does not give quantum advantage under a standard assumption in the unrelativized world. 
%but the random oracle model is not always sound~\cite{CGH04} and thus \cite{FOCS:YamZha22} does not give quantum advantage under a standard assumption in the unrelativized world. 

\paragraph{\bf Our question.}
In summary, several types of quantum advantage have been shown based on various assumptions in classical complexity theory or classical cryptography.
However, all previous results are based on assumptions that are specific, relatively stronger, less standard, or newly
introduced ones. Is it possible to show quantum advantage based on general, weak, and standard assumptions?
It is widely recognized that the existence of classical cryptography is nearly equivalent to the existence of classically-secure
one-way functions (OWFs), because almost all classical cryptographic primitives 
imply the existence of OWFs, 
and many useful cryptographic applications can be constructed from OWFs,
such as pseudorandom generators, pseudorandom functions, commitments, secret-key encryption,
digital signatures, zero-knowledge, and more~\cite{DBLP:books/cu/Goldreich2001,DBLP:books/cu/Goldreich2004}.
Consequently, we face the following open problem.
\begin{center}
{\it Can we show quantum advantage based on classically-secure OWFs?}    
\end{center}
Ideally, such quantum advantage should be efficiently verifiable one, such as PoQ, but
unfortunately, we do not know how to construct PoQ from OWFs.
We emphasize that this open problem is highly non-trivial even if quantum advantage is inefficiently-verifiable one, given the fact that
all previous constructions of inefficiently-verifiable (or even non-verifiable) quantum advantage used assumptions much stronger than or less standard than OWFs.
In this paper, we therefore focus on the following open problem.
\begin{center}
{\it Can we show inefficiently-verifiable quantum advantage based on classically-secure OWFs?}    
\end{center}

\subsection{Our Results}
In this paper, we answer the above open problem affirmatively.
We demonstrate inefficiently-verifiable quantum advantage based on several fundamental assumptions,
specifically relying solely on the existence of OWFs. 

Because OWFs are necessary for almost all classical cryptographic primitives,
our findings yield the following surprising implication.
\begin{center}
{\it If there is no quantum advantage, then there is no classical cryptography!}
\end{center}

Let us explain our results more precisely. We construct what we call {\it inefficient-verifier proofs of quantumness} (IV-PoQ) from 
statistically-hiding and computationally-binding classical bit commitments.
IV-PoQ is an interactive protocol between
a verifier and a QPT prover,
which is divided into two phases.
In the first phase, the verifier is PPT, and it interacts with the QPT prover over a classical channel.
In the second phase, the verifier becomes inefficient, and makes the decision based on the transcript of the first phase.\footnote{The inefficient
verifier could also take the efficient verifier's secret information as input in addition to the transcript. However, without loss of generality,
we can assume that the inefficient verifier takes only the transcript as input, because we can always modify the protocol of the first phase
in such a way that the efficient verifier sends its secret information to the prover at the end of the first phase.}
If the QPT prover is honest, the inefficient verifier accepts with high probability,
but for any PPT malicious prover, the inefficient verifier accepts with only small probability.
The new notion of IV-PoQ captures both the standard PoQ and inefficiently-verifiable quantum advantage (including inefficiently-verifiable search problems).

Our main result is the following:
\begin{theorem}\label{thm:PoQ_from_Com}
$(k+6)$-round IV-PoQ exist if statistically-hiding and computationally-binding classical bit commitments with $k$-round commit phase exist. 
\end{theorem}
A proof of \cref{thm:PoQ_from_Com} is given in \cref{sec:construction}.
Note that we actually need the statistical hiding property only for the honest receiver,
because the receiver corresponds to the verifier.
Moreover, note that in our construction, the inefficient verifier
in the second phase is enough to be a classical deterministic polynomial-time algorithm that queries the $\mathbf{NP}$ oracle.
(See \cref{sec:powerofV2}.)

Because statistically-hiding and computationally-binding classical bit commitments can be constructed from
OWFs~\cite{SIAM:HNORV09},
we have the following result.
\begin{theorem}\label{thm:PoQ_from_OWFs}
IV-PoQ exist if OWFs exist. 
\end{theorem}
%To our knowledge, this is the first time that quantum advantage is shown based only on OWFs. \takashi{I don't think we need to repeat this twice.}

Moreover, it is known that constant-round statistically-hiding and computationally-binding bit commitments can be 
constructed from distributional collision resistant hash functions~\cite{EC:BHKY19}\footnote{A distributional collision-resistant hash function~\cite{STOC:DubIsh06} is a weaker variant of a collision-resistant hash function that requires the hardness of sampling a collision $(x,y)$ where $x$ is uniformly random and $y$ is uniformly random conditioned on colliding with $x$.},
which exist if 
there is an hard-on-average problem in $\mathbf{SZK}$~\cite{C:KomYog18}.  
Therefore we also have the following result.\footnote{It is also known that constant-round statistically-hiding and computationally-binding
commitments can be constructed from multi-collision resistant hash functions~\cite{EC:BDRV18,EC:KomNaoYog18}, and therefore we have constant-round
IV-PoQ from multi-collision resistant hash functions as well.}
\begin{theorem}\label{thm:PoQ_from_SZK}
Constant-round IV-PoQ exist if there exist distributional collision-resistant hash functions,
which exist if there is an hard-on-average problem in $\mathbf{SZK}$.  
\end{theorem}
%\mor{Note that the completeness-soundness gap of IV-PoQ in \cref{thm:PoQ_from_SZK} is only an inverse-polynomial one, and we do not know how to amplify
%it with the parallel repetition. For details, see the remark below.}

The assumptions in \cref{thm:PoQ_from_OWFs,thm:PoQ_from_SZK} are average-case-hard assumptions.
We can further weaken the assumptions to worst-case-hard ones if we require only worst-case soundness for IV-PoQ. 
Namely, we define {\it auxiliary-input IV-PoQ} (AI-IV-PoQ) that only requires 
that for any malicious prover, there exist infinitely many auxiliary inputs under which the prover cannot cheat.   
We can show the following:
\begin{theorem}\label{thm:AIPoQ_from_Com}
$(k+6)$-round AI-IV-PoQ exist if auxiliary-input statistically-hiding and computationally-binding classical bit commitments with $k$-round commit phase exist. 
\end{theorem}
Its proof is omitted because it is similar to that of \cref{thm:PoQ_from_Com}.
Although AI-IV-PoQ is weaker than IV-PoQ, we believe that it still demonstrates a meaningful notion of quantum advantage, because
it shows ``worst-case quantum advantage'' in the sense that no PPT algorithm can simulate the QPT honest prover on all auxiliary inputs.   

Auxiliary-input OWFs\footnote{Roughly speaking, auxiliary-input OWFs are keyed functions 
such that for each adversary there exist infinitely many keys on which the adversary fails to invert the function.} 
exist if
$\mathbf{CZK}\not\subseteq\mathbf{BPP}$~\cite{OW93}.\footnote{$\mathbf{CZK}$ is the class of promise problems that have computational zero-knowledge proofs.
By abuse of notation, we write $\mathbf{BPP}$ to mean the class of promise problems (instead of languages) that are decidable in PPT.
}
Moreover, the construction of statistically-hiding and computationally-binding commitments from OWFs in \cite{SIAM:HNORV09} can be
modified for the auxiliary-input setting.
We therefore have the following result.
\begin{theorem}\label{thm:AXPoQ_from_OWFs}
AI-IV-PoQ exist if there exist auxiliary-input OWFs, which exist if $\mathbf{CZK}\not\subseteq\mathbf{BPP}$.
\end{theorem}

Furthermore, relying on the known constructions of constant-round (auxiliary-input) statistically-hiding commitments~\cite{C:HalMic96,TCC:OngVad08}, 
we obtain the following result.
\begin{theorem}\label{thm:AXPoQ_from_SZK}
Constant-round AI-IV-PoQ exist if auxiliary-input collision-resistant hash functions 
%against PPT adversaries \takashi{We can omit "against PPT adversaries" because this is not stated in the other theorems.} 
exist (which is equivalent to $\mathbf{PWPP}\nsubseteq \mathbf{FBPP}$)\footnote{
\ifnum\cameraready=1
See the full version for the definitions of $\mathbf{PWPP}$ and $\mathbf{FBPP}$.
\else
See \cref{sec:AXCRH} for the definitions of $\mathbf{PWPP}$ and $\mathbf{FBPP}$.
\fi
} or $\mathbf{SZK}\nsubseteq \mathbf{BPP}$.  
\end{theorem}

Finally, 
we can also define another variant of IV-PoQ that we call
{\it infinitely-often IV-PoQ} (IO-IV-PoQ) where the soundness is satisfied for infinitely many
values of the security parameter. We note that IO-IV-PoQ lie between IV-PoQ and AI-IV-PoQ. 
It is known that 
infinitely-often OWFs exist if
$\mathbf{SRE}\not\subseteq\mathbf{BPP}$~\cite{C:AppRay16}.\footnote{$\mathbf{SRE}$ is the class of problems that admit statistically-private randomized encoding with polynomial-time client and computationally-unbounded server.} 
Therefore we also have the following result.
\begin{theorem}\label{thm:ioPoQ_from_SRE}
IO-IV-PoQ exist if 
infinitely-often OWFs exist, which exist if
$\mathbf{SRE}\not\subseteq\mathbf{BPP}$. 
\end{theorem}

A comparison table among existing and our results on quantum advantage can be found in 
\cref{table:QA}. 

\ifnum\llncs=0
\input{table.tex} %In the full version, the table is placed here
\fi

\paragraph{\bf Remarks on completeness-soundness gap.}  
We remark that the above theorems consider (AI-/IO-) IV-PoQ that only have an inverse-polynomial completeness-soundness gap, i.e., 
the honest QPT prover passes verification with probability at least $c$
and any PPT cheating prover passes verification with probability at most $s$ where $c-s\ge 1/\poly(\secp)$ for the security parameter $\secp$.  
Due to the inefficiency of verification, it is unclear if we can generically amplify the gap {\it even by sequential repetition}.\footnote{If we assume soundness against \emph{non-uniform} PPT adversaries, then it is easy to show that sequential repetition generically amplifies the gap. However, we consider the uniform model of adversaries in this paper since otherwise we would need non-uniform assumptions like one-way functions against non-uniform adversaries which is stronger than the mere existence of one-way functions against uniform adversaries.}
Fortunately, we find a stronger definition of soundness called strong soundness which our constructions satisfy and enables us to amplify the gap by sequential repetition. Roughly speaking, strong soundness requires that soundness holds for almost all fixed cheating prover's randomness rather than on average. See \cref{def:strong_soundness} %and \cref{def:AXstrong_soundness} 
for the formal definition. 
This enables us to amplify the completeness-soundness gap to be optimal for any of our constructions. 
However, we remark that this increases the round complexity and in particular, the schemes of \cref{thm:PoQ_from_SZK,thm:AXPoQ_from_SZK} are no longer constant-round if we amplify the completeness-soundness gap. 
This issue could be resolved if we could prove that parallel repetition amplifies the gap, but we do not know how to prove this. Remark that we cannot use existing parallel repetition theorems for interactive arguments because verification is inefficient. 
Indeed, it is observed in \cite{TCC:CanHalSte05} that parallel repetition may not amplify the gap when verification is inefficient even for two-round arguments. Thus, we believe that it is very challenging or even impossible to prove a general parallel repetition theorem for (AI-/IO-)IV-PoQ. Nonetheless, it may be still possible to prove a parallel repetition theorem for our particular constructions, which we leave as an interesting open problem. 
%At least, we show that gap amplification by sequential repetition works if we assume {\it non-uniform soundness}. Thus, if we want to obtain (auxiliary-input) PoQIV with completeness-soundness gap almost equal to $1$, we need non-uniform variants of the assumptions used in the above theorems (e.g., one-way functions against non-uniform PPT adversaries for \cref{thm:PoQ_from_OWFs} or $\mathbf{CZK}\not\subseteq\mathbf{P}/\mathbf{poly}$ for \cref{thm:AXPoQ_from_OWFs} etc.). 
%Unfortunately,  it is unclear if parallel repetition works even if we assume non-uniform soundness. It is an interesting open problem to remove the necessity of non-uniform soundness or to prove that parallel repetition amplifies the completeness-soundness gap. 

\ifnum\llncs=1
\input{table.tex}  %In the LNCS version, the table is placed here
\fi

\paragraph{\bf Implausibility of two-round AI-IV-PoQ.} 
It is natural to ask how many rounds of interaction are needed. 
As already mentioned, it is trivial to construct two-round PoQ if we assume the existence of classically-hard and quantumly-easy problems such as factoring. 
We show evidence that it is inevitable to rely on such an assumption for constructing two-round (AI-/IO-)IV-PoQ. 
In the following, 
we state theorems for AI-IV-PoQ, but they immediately imply similar results for IV-PoQ and IO-IV-PoQ because they are stronger than AI-IV-PoQ.

First, we prove that there is no classical black-box reduction from security of two-round AI-IV-PoQ to standard cryptographic assumptions unless the assumptions do not hold against QPT adversaries. 
\begin{theorem}[Informal]\label{thm:impossible_reduction}
For a two-round AI-IV-PoQ, if its soundness can be reduced to a game-based assumption by a classical black-box reduction, then the assumption does not hold against QPT adversaries.  
\end{theorem}
\ifnum\cameraready=1
The formal version of the theorem is given in the full version. 
\else
The formal version of the theorem is given in \cref{thm:impossible_reduction_formal}. 
\fi
Here, game-based assumptions are those formalized as a game between the adversary and challenger that include
(but not limited to) general assumptions such as security of OWFs, public key encryption, digital signatures, oblivious transfers, indistinguishability obfuscation, succinct arguments etc. as well as concrete assumptions such as the hardness of factoring, discrete-logarithm, LWE etc.\footnote{This is similar to falsifiable assumptions \cite{C:Naor03,STOC:GenWic11} but there is an important difference that we do not restrict the challenger to be efficient.} 
\ifnum\cameraready=1
See the full version for a formal definition.
\else
See \cref{def:game_based_assumption} for a formal definition.
\fi
In particular, since we believe that quantumly-secure 
OWFs exist, the above theorem can be interpreted as a negative result on constructing two-round AI-IV-PoQ from general 
OWFs.  

The proof idea is quite simple: Suppose that there is a classical black-box reduction algorithm $R$ that is given a malicious prover as an oracle and breaks an assumption. 
Intuitively, the reduction should still work even if it is given the honest quantum prover $\mathcal{P}$ as an oracle. By considering the combination of $R$ and  $\mathcal{P}$ as a single quantum adversary, the assumption is broken. 
We remark this can be seen as an extension of an informal argument in \cite{TQC:BKVV20} where they argue that it is unlikely that a two-round PoQ can be constructed from the hardness of the LWE problem.\footnote{They use one-round PoQ to mean what we call two-round PoQ by counting interaction from the verifier to prover and from the prover to verifier as a single round.} 

Note that \cref{thm:impossible_reduction} only rules out classical reductions. 
One may think that the above argument extends to rule out quantum reductions, but there is some technical difficulty.
Roughly speaking, the problem is that a coherent execution of the honest quantum prover may generate an entanglement between its message register and internal register unlike a coherent execution of a classical cheating prover 
\ifnum\cameraready=1
(see the full version for more explanations).
\else
(see \cref{rem:why_not_quantum} for more explanations).
\fi
\footnote{This observation is due to Mark Zhandry.}
To complement this, we prove another negative result that also captures some class of quantum reductions. 
\begin{theorem}[Informal]\label{thm:oracle_separation_two-round_OWF}
If a cryptographic primitive $\primitive$ has a quantumly-secure construction (possibly relative to a classical oracle), 
then there is a randomized classical oracle relative to which two-round AI-IV-PoQ do not exist but a quantumly-secure construction of $\primitive$ exists.   
\end{theorem} 
%\takashi{It's pointed out in https://arxiv.org/pdf/0806.0450.pdf (Section 5) that an oracle separation relative to randomized oracle is a kind of cheating. Then I'm not sure if this theorem is meaningful. I'll think about if we can make the oracle classical.If that doesn't work, this result will be removed from the paper.}
\ifnum\cameraready=1
The formal version of the theorem is given in the full version. 
\else
The formal version of the theorem is given in \cref{thm:separation_formal}. 
\fi
The above theorem can be interpreted as a negative evidence on constructing two-round IV-PoQ from a cryptographic primitive for which we believe that quantumly-secure constructions exist (e.g., OWFs, public key encryption, indistinguishability obfuscation etc.)
In particular, the above theorem rules out any constructions that work relative to randomized classical oracles.\footnote{Note that reductions that work relative to \emph{deterministic} classical oracles do not necessarily work relative to \emph{randomized} classical oracles \cite[Section 5]{Aaronson08}.}
\cref{thm:oracle_separation_two-round_OWF} is incomparable to \cref{thm:impossible_reduction} since \cref{thm:oracle_separation_two-round_OWF} does not require the reduction to be classical unlike \cref{thm:impossible_reduction}, but requires that the construction and reduction work relative to randomized classical oracles. 
%\cref{thm:impossible_reduction} as long as the reduction is classically black-box.\footnote{When we construct a cryptographic primitive $\primitive$ from $\mathtt{Q}$, we give a \emph{construction} of $\primitive$ based on that of $\mathtt{Q}$ and a \emph{reduction} that converts an adversary of $\primitive$ into that of $\mathtt{Q}$. Thus we can argue black-boxness of constructions and reductions separately. See e.g., \cite{TCC:ReiTreVad04,AC:BaeBrzFis13} for more details.}

Again, the proof idea is simple. Suppose that a quantumly-secure construction $f$ of a primitive $\primitive$ exists relative to an oracle $O$. Then we introduce an additional oracle $Q^O$ that takes a description of a quantum circuit $C^O$ with $O$-gates and its input $x$ as input and outputs a classical string according to the distribution of $C^O(x)$. Relative to oracles $(O,Q^O)$, there do not exist two-round AI-IV-PoQ since a classical malicious prover can query the description of the honest quantum prover to $Q^O$ to get a response that passes the verification with high probability. On the other hand, $f$ is quantumly-secure relative to $(O,Q^O)$ since we assume that it is quantumly-secure relative to $O$ and the additional oracle $Q^O$ is useless for quantum adversaries since they can simulate it by themselves.   

We remark that the above theorems do not completely rule out black-box constructions of two-round AI-IV-PoQ from quantumly-hard assumptions. 
For example, consider a quantum black-box reduction that queries a cheating prover with a fixed randomness multiple times. 
Such a reduction is not captured by \cref{thm:impossible_reduction} because it is quantum. Moreover, it is not captured by \cref{thm:oracle_separation_two-round_OWF} because it does not work relative to randomized classical oracles since we cannot fix the randomness of the randomized classical oracle. It is a very interesting open problem to study if such a reduction is possible.

\if0
Proofs of quantumness with inefficient verification relative to random oracles were
constructed implicitly in \cite{STOC:Aaronson10} (FourierFishing) and explicitly in \cite{ACCGSW22} (CollisionHashing). 
In particular, the idea of \cite{ACCGSW22}, namely, constructing a superposition of two computational-basis states
by querying random oracles is somewhat similar to ours.
\fi

%\mor{Our third result is that there exists an oracle relative to which SampBPP=SampBQP but OWFs exist (?).}
%\takashi{We can also add the necessity of $\mathbf{PP}\nsubseteq \mathbf{BPP}$ ($\mathbf{PSPACE}\nsubseteq \mathbf{BPP}$ if we can prove it.)}  

\paragraph{\bf Quantum advantage based on quantum primitives weaker than OWFs.}
The existence of OWFs is the most fundamental assumption in classical cryptography.
Interestingly, it has been realized recently that it is not necessarily 
the case in quantum  cryptography~\cite{C:JiLiuSon18,Kre21,C:MorYam22,C:AnaQiaYue22,ITCS:BraCanQia23,TCC:Luowen,EPRINT:CaoXue22b,EPRINT:MorYam22c,KreQiaSinTal22}.
Many quantum cryptographic tasks can be realized with new quantum primitives, which seem to be weaker than OWFs,
such as pseudorandom states generators~\cite{C:JiLiuSon18}, one-way states generators~\cite{C:MorYam22},
and EFI~\cite{ITCS:BraCanQia23}.
Can we construct PoQ (or its variants) from quantum primitives that seem to be weaker than OWFs?
We show that variants of PoQ can be constructed from (classically-secure) quantum-evaluation OWFs (QE-OWFs).
QE-OWFs is the same as the standard classically-secure classical OWFs except that the function evaluation algorithm
is not deterministic classical polynomial-time but quantum polynomial-time. 
\ifnum\cameraready=1
(For its precise definition, see the full version.)
\else
(For its precise definition, see \cref{sec:QEOWFs}.)
\fi
QE-OWFs seem to be weaker than classically-secure classical OWFs. (For example, consider the function $f$ that on input $(x,y)$ outputs $\Pi_L(x)\|g(y)$,
where $L$ is any language in $\mathbf{BQP}\setminus \mathbf{BPP}$,
$\Pi_L$ is a function such that $\Pi_L(x)=1$ if $x\in L$ and $\Pi_L(x)=0$ if $x\notin L$,
and $g$ is any classically-secure classical OWF. 
$f$ is a QE-OWF, and 
$f$ cannot be evaluated in classical polynomial-time if $\mathbf{BQP}\neq\mathbf{BPP}$.
\ifnum\cameraready=1
For details, see the full version.) 
\else
For details, see \cref{sec:QEOWFs}.) 
\fi
We show the following result.
\begin{theorem}\label{thm:QE_OWFs}
If QE-OWFs exist, then quantum-verifier PoQ (QV-PoQ) exist or infinitely-often classically-secure classical OWFs exist.   
\end{theorem}
\ifnum\cameraready=1
A proof of the theorem is given in the full version.
\else
A proof of the theorem is given in \cref{sec:QEOWFs}.
\fi
QV-PoQ is the same as PoQ except that the verifier is a QPT algorithm.
Such a new notion of PoQ will be useful, for example, when many local quantum computers are connected
over classical internet: A quantum local machine may want to check whether it is interacting with a quantum computer or not over a classical channel.

The proof idea of \cref{thm:QE_OWFs} is as follows.
Let $f$ be a QE-OWF. We construct QV-PoQ as follows:
The verifier first chooses $x\gets\bit^n$ and sends it to the prover.
The prover then returns $y$. The verifier finally evaluates $f(x)$ by himself, and accepts if it is equal to $y$.
If the soundness holds, we have QV-PoQ.
On the other hand, if the soundness does not hold, then it means that 
$f$ can be evaluated in PPT, which means that $f$ is a classical OWF.
It is an interesting open problem whether PoQ or its variants can be constructed from pseudorandom quantum states generators,
one-way states generators, or EFI.

Infinitely-often classically-secure OWFs imply IO-IV-PoQ (\cref{thm:ioPoQ_from_SRE}), and therefore \cref{thm:QE_OWFs}
shows that the existence of QE-OWFs anyway implies quantum advantage (i.e., QV-PoQ or IO-IV-PoQ). 
Moreover, QV-PoQ in \cref{thm:QE_OWFs} implies IV-PoQ (and therefore IO-IV-PoQ).
(In general, QV-PoQ does not necessarily imply IV-PoQ, but in our case, it does because
our construction of QV-PoQ is a two-round protocol with 
the verifier's first message being a uniformly-randomly-chosen classical bit string.)
Hence we have the result that QE-OWFs implies IO-IV-PoQ in either case.

\if0
\paragraph{\bf NISQ implementations.}
One important motivation underlying many previous results on quantum advantage for sampling~\cite{STOC:AarArk11,BreJozShe10,NatPhys:BFNV19,FKMNTT18,BreMonShe16,Morimae17} 
and inefficiently-verifiable searching~\cite{CCC:AarChe17,AarGun19,STOC:Aaronson10} is experimental implementation on near-term quantum computers, such as random circuits.
%Realizations on NISQ devices are important, but we also believe that exploring the weakest assumption for quantum advantage is 
%a non-trivial and very important fundamental research subject even if full-fledged quantum computers are assumed to be available.
It is an interesting open problem whether the prover's operation in our PoQ can be simplified so that it can be demonstrated on a NISQ device.
\fi

\subsection{Technical Overview}\label{sec:overview}
In this subsection, we provide technical overview of our main result, \cref{thm:PoQ_from_Com}, namely,
the construction of IV-PoQ from statistically-hiding commitments.
(The construction of AI-IV-PoQ is similar.)
Our construction is based on
PoQ of \cite{NatPhys:KMCVY22}.
Let us first review their protocol.
Their protocol can be divided into two phases.
In the first phase,
the verifier first generates a pair of a trapdoor and a
trapdoor 2-to-1 collision resistant hash function $F$.
The verifier sends $F$ to the prover.
The prover generates the quantum state $\sum_{x\in\bit^\ell}\ket{x}\ket{F(x)}$, and measures the second register
in the computational basis
to obtain the measurement result $y$. The post-measurement state is
$\ket{x_0}+\ket{x_1}$, where $F(x_0)=F(x_1)=y$.
This is the end of the first phase.

In the second phase,
the verifier chooses a challenge bit $c\in\bit$ uniformly at random.
If $c=0$, the verifier asks the prover to measure the state in the computational basis.
The verifier accepts and halts if the prover's measurement result is $x_0$ or $x_1$. (The verifier can compute
$x_0$ and $x_1$ from $y$, because it has the trapdoor.)
The verifier rejects and halts if the prover's measurement result is not correct.
If $c=1$, the verifier sends the prover a bit string $\xi\in\bit^\ell$ which is chosen uniformly at random.
The prover changes the state $\ket{x_0}+\ket{x_1}$ into the state
$\ket{\xi\cdot x_0}\ket{x_0}+\ket{\xi\cdot x_1}\ket{x_1}$, and
measures the second register in the Hadamard basis.
If the measurement result is $d\in\bit^\ell$, the post-measurement state is
$\ket{\xi\cdot x_0}+(-1)^{d\cdot(x_0\oplus x_1)}\ket{\xi\cdot x_1}$, 
which is one of the BB84 states $\{\ket{0},\ket{1},\ket{+},\ket{-}\}$.
The verifier then asks the prover to measure this single-qubit state in a certain basis,
and accepts if the measurement result is appropriate. This is the end of the second phase.
Intuitively, the soundness comes from the collision resistance of $F$:
If a malicious PPT prover is accepted by the verifier with some high probability for both challenges, $c=0$ and $c=1$,
we can construct a PPT adversary that can find both $x_0$ and $x_1$ with non-negligible probability, which contradicts the collision resistance.

Therefore, once we can construct an interactive protocol where a verifier can let a prover generate $\ket{x_0}+\ket{x_1}$ in such a way that
no malicious PPT prover can learn both $x_0$ and $x_1$, we can construct PoQ
by running the second phase of \cite{NatPhys:KMCVY22} on it.
Can we do that with only OWFs? 
Our key idea is to {\it coherently} execute statistically-hiding classical bit commitments, which can be constructed
from OWFs~\cite{SIAM:HNORV09}. 
(A similar idea was also used in \cite{ITCS:MorYam23}.)
The prover plays the role of the sender of the commitment scheme,
and the verifier plays the role of the receiver of the commitment scheme.
The prover first generates the state
$\sum_{b\in\bit}\sum_{x\in\bit^\ell}\ket{b}\ket{x}$,
which is the superposition of the bit $b\in\bit$ to commit
and sender's random seed $x\in\bit^\ell$.
The prover and the verifier then run the interactive commitment phase.
When the prover computes its message, it coherently computes the message on its state, and measures a register to obtain
the measurement result.\footnote{For example, in the prover's $j$th round, if the prover possesses a state $\sum_{b\in\bit}\sum_{x\in X_b}\ket{b}\ket{x}$, where $X_b$ is a certain set,
it changes the state into
$\sum_{b\in\bit}\sum_{x\in X_b}\ket{b}\ket{x}\ket{f_j(b,x,t_j)}$, and measures the third register to obtain the measurement result $\alpha_j$,
where $f_j$ is the function that computes sender's $j$th message, and $t_j$ is the transcript obtained before the $j$th round.
The prover sends $\alpha_j$ to the verifier as the sender's $j$th message.}
The prover sends the measurement result as the sender's message to the verifier.
The verifier runs classical receiver's algorithm, and sends classical message to the prover.
At the end of the commit phase,
the honest prover possesses
the state
\begin{equation}
   \ket{0}\sum_{x\in X_{0,t}}\ket{x}+\ket{1}\sum_{x\in X_{1,t}}\ket{x}, 
   \label{Pstate}
\end{equation}
where $X_{b,t}$ is the set of sender's random seeds that are consistent with the committed bit $b$
and the transcript $t$,
which is the sequence of all classical messages exchanged between the prover and the verifier.

If $|X_{0,t}|=|X_{1,t}|=1$, \cref{Pstate}
is $\ket{0}\ket{x_0}+\ket{1}\ket{x_1}$, where $x_b$ is the unique element of $X_{b,t}$ for
each $b\in\bit$.
In that case, we can run the second phase of \cite{NatPhys:KMCVY22} on it.\footnote{Strictly speaking, 
$\ket{0}\ket{x_0}+\ket{1}\ket{x_1}$ is not equal to
$\ket{x_0}+\ket{x_1}$, but the protocol can be easily modified.
Given $\xi$, the prover has only to change
$\ket{0}\ket{x_0}+\ket{1}\ket{x_1}$ 
to
$\ket{\xi\cdot x_0}\ket{x_0}+\ket{1\oplus (\xi\cdot x_1)}\ket{x_1}$.
}
However, in general,
$|X_{0,t}|=|X_{1,t}|=1$ is not always satisfied, 
and if it is not satisfied, we do not know how to realize PoQ
from the state of \cref{Pstate}. This is our first problem.
Moreover, even if
$|X_{0,t}|=|X_{1,t}|=1$ is satisfied, we have the second problem:
The efficient verifier cannot compute $(x_0,x_1)$, because there is no trapdoor. 
The efficient verifier therefore cannot check whether the prover passes the tests or not.\footnote{
In \cite{ITCS:MorYam23}, they resolve the first problem by using a specific commitment scheme of \cite{C:NOVY92} and resolve the second problem by simply assuming the existence of a trapdoor. However, since the commitment scheme of   \cite{C:NOVY92} relies on one-way {\it permutations}, their idea does not work based on OWFs even if we give up efficient verification. 
}

Unfortunately, we do not know how to solve the second problem, and therefore we have to give up the efficient verification.
On the other hand,
we can solve the first problem by introducing a new hashing technique, which is similar to \cite{VV86}.
%which may have further applications. 
First, we notice that
$|X_{0,t}|\simeq|X_{1,t}|$ with overwhelming probability,
because otherwise the statistical-hiding of the classical bit commitment
scheme is broken.
Next, let $\mathcal{H}\coloneqq\{h:\cX\to\cY\}$
be a pairwise-independent hash family with $\cX=\bit^\ell$.
The verifier chooses $h_0,h_1\in\mathcal{H}$ uniformly at random, and sends $(h_0,h_1)$ to the prover.
The prover changes the state of \cref{Pstate} into
   $
   \ket{0}\sum_{x\in X_{0,t}}\ket{x}\ket{h_0(x)}+\ket{1}\sum_{x\in X_{1,t}}\ket{x}\ket{h_1(x)}, 
   $
and measures the third register in the computational-basis to obtain the measurement result $y$. 
We show that if $|\cY|$ is chosen so that $|\cY|\simeq 2|X_{b,t}|$,
the state collapses by the measurement
to
$\ket{0}\ket{x_0}+\ket{1}\ket{x_1}$
with constant probability,
where $x_b\in X_{b,t}\cap h_b^{-1}(y)$ for $b\in\bit$.
The remaining problem is that the efficient verifier cannot compute $|X_{b,t}|$,
and therefore it cannot find the appropriate $|\cY|$.
This problem is solved by noticing that even if the verifier chooses $|\cY|$ randomly,  
it is $\simeq2|X_{b,t}|$ with non-negligible probability. 
More precisely, let $m$ be an integer such that $(1+\epsilon)^m\ge 2^{\ell+1}$,
where $0<\epsilon<1$ is a small constant (which we take $\epsilon=1/100$, for example).
Then, we show that there exists a $j^*\in\{0,1,...,m-1\}$ such that $\lceil (1+\epsilon)^{j^*}\rceil\simeq 2|X_{b,t}|$.
Therefore, if the efficient verifier chooses $j\in\{0,1,...,m-1\}$ uniformly at random, 
and sets $\cY\coloneqq [\lceil (1+\epsilon)^{j}\rceil]$,
then $|\cY|\simeq 2|X_{b,t}|$ is satisfied with probability $1/m=1/\poly(\secp)$.

In summary, the efficient verifier can let the honest prover generate
$\ket{0}\ket{x_0}+\ket{1}\ket{x_1}$ with non-negligible probability.
Fortunately, the second phase of \cite{NatPhys:KMCVY22} is a public coin one, which means that
all messages from the verifier are uniformly-chosen random bit strings,
and therefore our efficient verifier can send all its messages without
doing any inefficient computation (such as finding an element of $X_{b,t}\cap h_b^{-1}(y)$, etc.).
All verifications are later done by the inefficient verifier.

The soundness of our construction is shown from
the computational-binding of the classical bit commitment scheme.
In the soundness proof of \cite{NatPhys:KMCVY22}, they use the fact that no PPT malicious prover
can find both $x_0$ and $x_1$, which comes from the collision resistance.
In our case, we have that property from the computational-binding of the classical bit commitment scheme.
In a similar way as the soundness proof of \cite{NatPhys:KMCVY22},
we can construct
a PPT adversary $\cA$ that 
can find both $x_0$ and $x_1$ from a PPT malicious prover
that passes both challenges with some high probability.
We can then construct a PPT adversary $\cB$ that
breaks computational-binding of the classical bit commitment scheme
from $\cA$.

There is, however, a large difference in our case from that of \cite{NatPhys:KMCVY22}. 
In the protocol of \cite{NatPhys:KMCVY22}, the honest prover's state is always $\ket{x_0}+\ket{x_1}$, but in our case
$|X_{0,t}\cap h_0^{-1}(y)|=|X_{1,t}\cap h_1^{-1}(y)|=1$ is not always satisfied.
In order to keep the $1/\poly$ completeness-soundness gap in our protocol, we need a trick
for the algorithm of the inefficient verifier.
The inefficient verifier first checks whether
$|X_{0,t}\cap h_0^{-1}(y)|=|X_{1,t}\cap h_1^{-1}(y)|=1$ is satisfied or not.
If it is satisfied, the inefficient verifier computes the unique element
$x_b\in X_{b,t}\cap h_b^{-1}(y)$ for each $b\in\bit$, and checks whether the transcript passes the second phase of the protocol of \cite{NatPhys:KMCVY22} or not.
On the other hand, if $|X_{0,t}\cap h_0^{-1}(y)|=|X_{1,t}\cap h_1^{-1}(y)|=1$ is not satisfied, we need some trick.
A naive attempt would be to always accept in such a case.
Intuitively, this would give a $1/\poly$ completeness-soundness gap because we have a constant completeness-soundness gap conditioned on $|X_{0,t}\cap h_0^{-1}(y)|=|X_{1,t}\cap h_1^{-1}(y)|=1$ by \cite{NatPhys:KMCVY22}  
and such an event occurs with probability $1/\poly$ as explained above.  
However, there is a flaw in the argument because a malicious prover may change the probability that $|X_{0,t}\cap h_0^{-1}(y)|=|X_{1,t}\cap h_1^{-1}(y)|=1$ holds. 
For example, if it can control the probability to be $1$, then it passes the verification with probability $1$, which is even higher than the honest quantum prover's success probability!
Due to a similar reason, an attempt to let the inefficient verifier always reject when $|X_{0,t}\cap h_0^{-1}(y)|=|X_{1,t}\cap h_1^{-1}(y)|=1$ is not satisfied also does not work.
Our idea is to take the middle of the two attempts:
If 
$|X_{0,t}\cap h_0^{-1}(y)|=|X_{1,t}\cap h_1^{-1}(y)|=1$ is not satisfied,
the inefficient verifier accepts with probability $s$ and rejects with probability $1-s$,
where $s$ is the soundness parameter of the PoQ protocol of \cite{NatPhys:KMCVY22},
i.e., for any malicious prover, the verifier accepts with probability at most $s+\negl(\secp)$.
Let $p_{\mathbf{good}}$ be the probability that
$|X_{0,t}\cap h_0^{-1}(y)|=|X_{1,t}\cap h_1^{-1}(y)|=1$ is satisfied in the interaction between the honest prover and the verifier.
Then, the probability that the inefficient verifier accepts the honest prover is at least
$p_{\mathbf{good}}c+(1-p_{\mathbf{good}})s$,
where $c$ is the completeness parameter of the PoQ protocol of \cite{NatPhys:KMCVY22},
i.e., the verifier accepts the honest prover with probability at least $c$.
On the other hand, we show that the soundness parameter of our protocol is also $s$. 
(Intuitively, this is because if $|X_{0,t}\cap h_0^{-1}(y)|=|X_{1,t}\cap h_1^{-1}(y)|=1$ is satisfied, then a malicious prover can pass the verification 
with probability at most $s+\negl(\secp)$ by the soundness of the PoQ protocol of \cite{NatPhys:KMCVY22}, 
and 
if $|X_{0,t}\cap h_0^{-1}(y)|=|X_{1,t}\cap h_1^{-1}(y)|=1$ is not satisfied, 
the verifier accepts with probability $s$ regardless of the prover's behavior.)
Therefore, we have
$p_{\mathbf{good}}c+(1-p_{\mathbf{good}})s-s=p_{\mathbf{good}}(c-s)\ge1/\poly$,
because $p_{\mathbf{good}}\ge1/\poly$ as we have explained.
In this way, we can achieve the $1/\poly$ completeness-soundness gap.

Finally, in our construction, the inefficient verifier
is enough to be a classical deterministic polynomial-time algorithm
that queries the $\mathbf{NP}$ oracle, because as we have explained above, inefficient computations that
the inefficient verifier has to do are verifying
$|X_{0,t}\cap h_0^{-1}(y)|=|X_{1,t}\cap h_1^{-1}(y)|=1$
and finding the single element $x_b\in X_{b,t}\cap h_b^{-1}(y)$ for each $b\in\bit$.

\subsection{Related Works}
IV-PoQ from random oracles was
constructed in \cite{ACCGSW22}, which they call Collision Hashing. 
Their construction is based on the observation that if the state $\sum_x\ket{x}\ket{g(x)}$ is generated, where $g$ is a random oracle,
and the second register is measured in the computational basis, the post-measurement state
$\sum_{x\in g^{-1}(y)}\ket{x}$ corresponding to the measurement result $y$ is a superposition 
of two computational-basis states with some probability on which the second phase of \cite{NatPhys:KMCVY22} can be run.
(Actually, because they assume random oracles, the non-interactive protocol of \cite{TQC:BKVV20} can be run instead of \cite{NatPhys:KMCVY22}.)
This idea seems to be somehow related to our idea. 

\cite{STOC:AarAmb15} studied a sampling problem, Fourier Sampling, where
given an oracle $f:\bit^n\to\{+1,-1\}$, it is required to sample from
the distribution $\{p_y\}_y$, where 
$p_y\coloneqq 2^{-n}\hat{f}(y)^2=\left(\frac{1}{2^n}\sum_{x\in\bit^n}f(x)(-1)^{x\cdot y}\right)^2$,
within an additive error.  It needs exponentially-many queries to classically solve it
relative to a random oracle.
\cite{STOC:Aaronson10} also introduced a search problem, Fourier Fishing, 
where given an oracle
$f:\bit^n\to\{+1,-1\}$,
find $z\in\bit^n$ such that $|\hat{f}(z)|\ge1$. 
It needs exponentially-many queries to classically solve it relative to a random oracle.
The verification of Fourier Fishing can be done inefficiently.
\cite{STOC:Aaronson10} also introduced a decision problem, Fourier Checking, and show that
it requires exponentially-many queries to solve it classically
relative to a certain oracle.
Whether $\mathbf{BQP}\neq\mathbf{BPP}$ relative to a random oracle is an open problem,
and given the Aaronson-Ambainis conjecture~\cite{TC:AarAmb14}, showing it seems to be difficult.

\cite{CCC:AarChe17}
showed that if 
OWFs exist, then there are
oracles $A\in\mathbf{P}/\mathbf{poly}$ such that $\mathbf{BPP}^A \neq \mathbf{BQP}^A$ 
(and even $\mathbf{BQP}^A\not\subset \mathbf{SZK}^A$).
The paper also showed that
if there exist subexponentially-secure OWFs,
then Fourier Sampling and Fourier Fishing are classically hard relative to oracles in $\mathbf{P}/\mathbf{poly}$.
Regarding
the possibility of removing the oracles, the authors
say that
``{\it ... in the unrelativized world, there seems to be no hope at present of proving $\mathbf{BPP} \neq \BQP$ under any hypothesis nearly as weak as the existence of one-way functions}'',
which suggests the difficulty of demonstrating quantum advantage based only on one-way functions.
We bypass the difficulty by considering interactive protocols.

It was pointed out in \cite{ITCS:LiuLiuQian22} that
the complexity assumption of 
$\mathbf{PP}\neq \mathbf{BPP}$ is necessary for the existence of PoQ.
A similar idea can be applied to show that
$\mathbf{PP}\neq \mathbf{BPP}$ is necessary for the existence of
(AI-/IO-)IV-PoQ. 
\ifnum\cameraready=1
(For the convenience of readers, we provide a proof in the full version.)
\else
(For the convenience of readers, we provide a proof in \cref{sec:necessity_assumption}.)
\fi
We remark that the proof holds even if we allow the honest prover to perform post-selection.
Moreover, it holds even if the verifier in the first phase is unbounded-time.  
%One may think that this directly follows from the fact that $\mathbf{PostBQP}=\mathbf{PP}$~\cite{Aaronson04}. However, we believe that we cannot use it in a black-box manner since they are classes of non-interactive decision problems whereas we have to simulate an interactive QPT algorithm that outputs multi-bit strings. Fortunately, we observe that the proof of $\mathbf{PostBQP}=\mathbf{PP}$ can be adapted to our setting.

Unconditional quantum advantage over restricted classical computing
was also studied~\cite{BraGosKon18,BraGosKonTom20,STOC:WKST19,STOC:GriSch20}.
Unconditional separations between quantum and classical computing are appealing, 
but in this paper we do not focus on the setups of restricting classical computing.
Note that showing unconditional quantum advantage without restricting classical computing is at least as hard as proving 
$\mathbf{PP}\neq \mathbf{BPP}$ 
\ifnum\cameraready=1
(\cite{ITCS:LiuLiuQian22}), 
\else
(\cite{ITCS:LiuLiuQian22} and \cref{sec:necessity_assumption}), 
\fi
which is a major open problem in complexity theory.

The idea of coherently running statistically-hiding commitments was first introduced in \cite{ITCS:MorYam23}.
However, they could apply the idea only to the specific commitment scheme of \cite{C:NOVY92} whereas we can apply it to \emph{any} statistically-hiding commitments. This is made possible by introducing the hashing technique as explained in \cref{sec:overview}.

\if0
\mor{Aaronson Chen showed that $\mathbf{NP}\not\subseteq\mathbf{BPP}$ is necessary for 
$SampBQP^A\neq SampBPP^A$ for an oracle A in P/poly.}
\fi

\if0
\begin{table}[h]
 \caption{Some results on quantum advantage}
 \label{table:QA}
 \centering
  \begin{tabular}{clll}
   \hline
   Type & Verification & Non-interactive? & Assumption \\
   \hline \hline
   Multiplicative-error sampling & Maybe impossible & Yes & PH does not collapse \\
   Additive-error sampling & Maybe impossible & Yes & Ad hoc~/~random oracle \\
   HOG and XHOG & Inefficient (public) & Yes & Ad hoc \\
   FOURIER FISHING & Inefficient (public) & Yes & random oracle \\
   Specific problems in $\mathbf{NP}$ (such as factoring) & Efficient (public) & Yes & Specific assumptions \\
   Proofs of quantumness & Efficient (private) & No & TDCRHF~/~TDP~/~QHE \\
   Yamakawa-Zhandry & Efficient (public) & Yes & Random oracle \\
   $\mathbf{BQP}^A\neq\mathbf{BPP}^A$~\cite{CCC:AarChe17}& Inefficient (QPT) & Yes & OWFs+ oracles $A\in\mathbf{P}/\mathbf{poly}$\\ 
   FOURIER SAMPLING~\cite{CCC:AarChe17}& May be impossible & Yes & subexp-OWFs+ oracles $A\in\mathbf{P}/\mathbf{poly}$\\ 
   FOURIER FISHING~\cite{CCC:AarChe17}& Inefficient(public) & Yes & subexp-OWFs+ oracles $A\in\mathbf{P}/\mathbf{poly}$\\ 
   Our results & Inefficient (public) & No & OWFs~/~$\mathbf{CZK}\not\subseteq\mathbf{BPP}$ \\ 
   \hline
  \end{tabular}
\end{table}
\fi

%% file: table.tex
%\ifnum\llncs=0
%\begin{table}[h]
%\else
\begin{table}[H]
%\fi
 \caption{Comparison among results on quantum advantage. 
In column ``Verification'', ``No'' means that the verification is not known to be possible. (Actually, it seems to be impossible.)
In column ``Assumption'', 
PH stands for the polynomial-time hierarchy, 
seOWFs stands for subexponentially secure one-way functions,
2-1 TDCRHFs stands for 2-to-1 trapdoor collision-resistant hash functions, 
QHE stands for quantum homomorphic encryption, 
fdTDPs stands for full-domain trapdoor permutations, 
OWFs stands for one-way functions, 
dCRHFs stands for distributional collision-resistant hash functions, 
and CRHFs stands for collision-resistant hash functions.
  In column ``Misc'', 
 Mult.err. and Add.err. stand for multiplicative and additive errors, respectively.  
 In the row of \cite{FOCS:Shor94}, the number of rounds is two, because the verifier sends a composite number to the prover, and the prover returns 
its factorization. 
It can be considered as a non-interactive if the composite number is given as an auxiliary input.\\
 }
 \label{table:QA}
 \centering
  \begin{tabular}{lcccc}
   \hline
   Ref.& Verification& \#Rounds &Assumption&Misc\\
   \hline \hline
\cite{TD04,STOC:AarArk11,BreJozShe10,FKMNTT18}  & No & $1$ & PH does not collapse& Mult.err. sampling\\
\hline
\cite{STOC:AarArk11,BreMonShe16,NatPhys:BFNV19,Morimae17}& No & $1$ & Ad hoc& Add.err. sampling\\
\hline
\cite{STOC:AarAmb15}& No & $1$ & Random oracle& Fourier Sampling\\
\hline
\cite{CCC:AarChe17}& No & $1$ & seOWFs+$\mathbf{P}/\mathbf{poly}$-oracle& Fourier Sampling\\
\hline
   \cite{CCC:AarChe17,AarGun19} & Inefficient & $1$ & Ad hoc &HOG, XHOG\\
   \hline
   \cite{STOC:Aaronson10}& Inefficient  & $1$ &Random oracle& Fourier Fishing\\
   \hline
   \cite{CCC:AarChe17}& Inefficient  & $1$ &seOWFs+$\mathbf{P}/\mathbf{poly}$-oracle& Fourier Fishing\\
   \hline
   \cite{ACCGSW22} & Inefficient  & $1$ & Random oracle &Collision Hashing\\
   \hline
   \cite{FOCS:Shor94}& Efficient  &$2$ &Factoring/Discrete-log&\\
   \hline
       \cite{FOCS:YamZha22} &Efficient & $1$ & Random oracle &\\ 
       \hline
       \cite{JACM:BCMVV21,NatPhys:KMCVY22}& Efficient  & $O(1)$ &  (Noisy) 2-1 TDCRHFs&\\
       \hline
       \cite{cryptoeprint:2022/400}& Efficient  & $O(1)$ &  QHE &\\
       \hline
\cite{ITCS:MorYam23}& Efficient  & $\poly(\secp)$ & fdTDPs &\\
\hline
   \cref{thm:PoQ_from_OWFs}  & Inefficient& $\poly(\secp)$ &OWFs &\\ 
       \hline
      \cref{thm:PoQ_from_SZK}  & Inefficient& $O(1)$ &dCRHFs&\\ 
 
   \hline
    \cref{thm:AXPoQ_from_OWFs}  & Inefficient &$\poly(\secp)$ &   \begin{tabular}{c} Auxiliary-input OWFs~/~\\
    $\mathbf{CZK}\not\subseteq\mathbf{BPP}$
    \end{tabular}&AI-IV-PoQ\\ 
     \hline
    \cref{thm:AXPoQ_from_SZK}  & Inefficient& $O(1)$ &
   \begin{tabular}{c}
    Auxiliary-input CRHFs~/~\\
    $\mathbf{SZK}\not\subseteq\mathbf{BPP}$
    \end{tabular}&AI-IV-PoQ
    \\
   \hline
\cref{thm:ioPoQ_from_SRE}
     & Inefficient& $\poly(\secp)$ &
   \begin{tabular}{c}
    Infinitely-often OWFs~/~\\
    $\mathbf{SRE}\not\subseteq\mathbf{BPP}$
    \end{tabular}&IO-IV-PoQ
    \\
   \hline
  \end{tabular}
\end{table}

%% file: preliminaries.tex
\section{Preliminaries}\label{sec:preliminaries}

\subsection{Basic Notations}
\label{sec:basic_notations}
%check done 2022/09/28

We use the standard notations of quantum computing and cryptography.
We use $\secp$ as the security parameter.
$[n]$ means the set $\{1,2,...,n\}$.
For any set $S$, $x\gets S$ means that an element $x$ is sampled uniformly at random from the set $S$.
For a set $S$, $|S|$ means the cardinality of $S$.
We write $\negl$ to mean a negligible function and $\poly$ to mean a polynomial. 
PPT stands for (classical) probabilistic polynomial-time and QPT stands for quantum polynomial-time.
%If we say that an adversary is PPT, it implicitly means non-uniform PPT. \takashi{I think we should remove this sentence because we consider uniform adversaries by default.}
%A QPT unitary is a unitary operator that can be implemented in a QPT quantum circuit.
For an algorithm $A$, $y\gets A(x)$ means that the algorithm $A$ outputs $y$ on input $x$.
%In particular, if $x$ and $y$ are quantum states and $A$ is a quantum algorithm, $y\gets A(x)$ means the following:
%a unitary $U$ is applied on $x\otimes|0...0\rangle\langle0...0|$, and some qubits are traced out.
%Then, the state of remaining qubits is $y$.
%This, importantly, means that the state $y$ is {\it uniquely decided} by the state $x$.
%If $A$ is a QPT algorithm, the unitary $U$ is QPT and the number of ancilla qubits 
%$|0...0\rangle$ is $\poly(\secp)$.
%If $x$ is a classical bit string, $y$ is a quantum state, and $A$ is a quantum algorithm,
%$y\gets A(x)$ sometimes means the following:
%a unitary $U_x$ that depends on $x$ is applied on $|0...0\rangle$, and some qubits are traced out.
%The state of the remaining qubits is $y$.
%We choose this picture if it is more convenient.
For two bit strings $x$ and $y$, $x\|y$ means the concatenation of them.
For simplicity, we sometimes omit the normalization factor of a quantum state.
(For example, we write $\frac{1}{\sqrt{2}}(|x_0\rangle+|x_1\rangle)$ just as
$|x_0\rangle+|x_1\rangle$.)
$I\coloneqq|0\rangle\langle0|+|1\rangle\langle 1|$ is the two-dimensional identity operator.
For the notational simplicity, we sometimes write $I^{\otimes n}$ just as $I$ when
the dimension is clear from the context.

\if0
$\|X\|_1\coloneqq\mbox{Tr}\sqrt{X^\dagger X}$ is the trace norm.
$\mbox{Tr}_\regA(\rho_{\regA,\regB})$ means that the subsystem (register) $\regA$ of the state $\rho_{\regA,\regB}$ on
two subsystems (registers) $\regA$ and $\regB$ is traced out.
For simplicity, we sometimes write $\mbox{Tr}_{\regA,\regB}(|\psi\rangle_{\regA,\regB})$ to mean
$\mbox{Tr}_{\regA,\regB}(|\psi\rangle\langle\psi|_{\regA,\regB})$.
$I$ is the two-dimensional identity operator. For simplicity, we sometimes write $I^{\otimes n}$ as $I$ 
if the dimension is clear from the context.
For the notational simplicity, we sometimes write $|0...0\rangle$ just as $|0\rangle$,
when the number of zeros is clear from the context.
For two pure states $|\psi\rangle$ and $|\phi\rangle$,
we sometimes write $\||\psi\rangle\langle\psi|-|\phi\rangle\langle\phi|\|_1$
as
$\||\psi\rangle-|\phi\rangle\|_1$
to simplify the notation.
$F(\rho,\sigma)\coloneqq\|\sqrt{\rho}\sqrt{\sigma}\|_1^2$
is the fidelity between $\rho$ and $\sigma$.
We often use the well-known relation between the trace distance and the fidelity:
$1-\sqrt{F(\rho,\sigma)}\le\frac{1}{2}\|\rho-\sigma\|_1\le\sqrt{1-F(\rho,\sigma)}$.
\fi

\ifnum\llncs=0
\input{definitions.tex}
\else

\ifnum\cameraready=1
Definitions of basic cryptographic tools including pair-wise independent hashes, one-way functions, commitments, and the infinitely-often and auxiliary-input variants of one-way functions and commitments are given in the full version. 
\else
Definitions of basic cryptographic tools including pair-wise independent hashes, one-way functions, commitments, and the infinitely-often and auxiliary-input variants of one-way functions and commitments are given in \Cref{sec:omitted_pre}. 
\fi

We use the following lemma about pairwise independent hashes.

\begin{lemma}
\label{lem:hash2}
Let $\mathcal{H}\coloneqq\{h:\cX\to\cY\}$
be a pairwise-independent hash family such that $|\cX|\ge2$.
Let $S\subseteq\cX$ be a subset of $\cX$.
For any $y\in\cY$,
\begin{align}
\Pr_{h\gets\mathcal{H}}[|S\cap h^{-1}(y)|=1]\ge
\frac{|S|}{|\cY|}
-\frac{|S|^2}{|\cY|^2}.
\label{ge2}
\end{align}
\end{lemma} 
A similar lemma appeared as an intermediate step for proving Valiant-Vazirani theorem~\cite{VV86} (see also \cite[Lemma 17.19]{AroraBarak09}).  
\ifnum\cameraready=1
We give a proof of \Cref{lem:hash2} in the full version.
\else
We give a proof of \Cref{lem:hash2} in \Cref{sec:proof_hashing} for completeness.
\fi
\fi

%% file: definitions.tex
%\ifnum\llncs=1
%\section{Omitted Definitions of Basic %Tools}\label{sec:omitted_definition}
%\fi
\subsection{Pairwise-Independent Hash Family}
\begin{definition}
A family of hash functions $\mathcal{H}\coloneqq \{h:\cX\to\cY\}$ 
is pairwise-independent
if 
for any two $x\neq x'\in \cX$ and any two $y,y'\in\cY$,
$
\Pr_{h\gets\mathcal{H}}[h(x)=y\wedge h(x')=y']=\frac{1}{|\cY|^2}.
$
\end{definition}

\ifnum\llncs=0
We use the following lemma about pairwise independent hashes.

\begin{lemma}
\label{lem:hash2}
Let $\mathcal{H}\coloneqq\{h:\cX\to\cY\}$
be a pairwise-independent hash family such that $|\cX|\ge2$.
Let $S\subseteq\cX$ be a subset of $\cX$.
For any $y\in\cY$,
\begin{align}
\Pr_{h\gets\mathcal{H}}[|S\cap h^{-1}(y)|=1]\ge
\frac{|S|}{|\cY|}
-\frac{|S|^2}{|\cY|^2}.
\label{ge2}
\end{align}
\end{lemma} 
A similar lemma appeared as an intermediate step for proving Valiant-Vazirani theorem~\cite{VV86} (see also \cite[Lemma 17.19]{AroraBarak09}).  
We give a proof of \Cref{lem:hash2} in \Cref{sec:proof_hashing} for completeness.
\fi

\subsection{OWFs}
\begin{definition}[OWFs]
\label{def:OWFs}
A function $f:\bit^*\to\bit^*$ is a (classically-secure) OWF if
it is computable in classical deterministic polynomial-time, and
for any PPT adversary $\cA$, there exists a negligible function $\negl$ such that
for any $\secp$,
\begin{equation}
\Pr[f(x')=f(x):x'\gets\cA(1^\secp,f(x)),x\gets\bit^\secp]
\le\negl(\secp).
\end{equation}
\end{definition}

\if0
\begin{remark}
In this paper, what we need is {\it classical-secure} one-way functions, namely, one-way functions secure against {\it PPT} adversaries.
No quantum-secure one-way functions (that is a one-way function secure against QPT adversaries) is necessary.
\end{remark}
\fi

\begin{definition}[Infinitely-often OWFs]
\label{def:IOOWFs}
A function $f:\bit^*\to\bit^*$ is a (classically-secure) infinitely-often OWF if
it is computable in classical deterministic polynomial-time, and there exists an infinite set $\Lambda\subseteq \mathbb{N}$ such that 
for any PPT adversary $\cA$, 
\begin{equation}
\Pr[f(x')=f(x):x'\gets\cA(1^\secp,f(x)),x\gets\bit^\secp]
\le\negl(\secp)
\end{equation}
for all $\secp\in \Lambda$. 
\end{definition}

\if0
\begin{definition}[Infinitely-often classically-secure OWFs]
\label{def:IOOWFs_alt}
A function $f:\bit^*\to\bit^*$ is an infinitely-often classically-secure OWF if
it is computable in classical deterministic polynomial-time, and
for any PPT adversary $\cA$ and polynomial $\poly$, 
\begin{equation}
\Pr[f(x')=f(x):x'\gets\cA(1^\secp,f(x)),x\gets\bit^\secp]
\le\frac{1}{\poly(\secp)}
\end{equation}
for infinitely many $\secp$. 
\takashi{I added this definition.}
\end{definition}
\fi

\begin{definition}[Auxiliary-input function ensemble]
An auxiliary-input function ensemble is a collection of functions
$\mathcal{F}\coloneqq \{f_\sigma:\bit^{p(|\sigma|)}\to \bit^{q(|\sigma|)}\}_{\sigma\in\bit^*}$, 
where $p$ and $q$ are polynomials. We
call $\mathcal{F}$ polynomial-time computable if there is a classical deterministic
polynomial-time algorithm $F$ such that for every $\sigma\in\bit^*$ and $x\in\bit^{p(|\sigma|)}$, we
have $F(\sigma,x)=f_\sigma(x)$.
\end{definition}

\begin{definition}[Auxiliary-input OWFs]
A (classically-secure) auxiliary-input OWF is a polynomial-time computable
auxiliary-input function ensemble $\mathcal{F}\coloneqq \{f_\sigma:\bit^{p(|\sigma|)}\to \bit^{q(|\sigma|)}\}_{\sigma\in\bit^*}$
such that for every
uniform PPT adversary $\cA$ and a polynomial $\poly(\secp)$, there exists an infinite set $\Lambda\subseteq\bit^*$ such that,
\begin{equation}
\Pr[f_\sigma(x')=f_\sigma(x): x'\gets\cA(\sigma, f_\sigma(x)),x\gets\bit^{p(|\sigma|)}]
\le\frac{1}{\poly(|\sigma|)}
\end{equation}
for all $\sigma\in \Lambda$. 
%\takashi{slightly modified the definition}
\end{definition}
\begin{remark}
It is easy to see that OWFs imply infinitely-often OWFs, and infinitely-often OWFs imply auxiliary-input OWFs.
\end{remark}

\begin{theorem}[\cite{OW93}]
\label{thm:AOWFs}
Auxiliary-input OWFs exist if $\mathbf{CZK}\not\subseteq \mathbf{BPP}$.    
\end{theorem}

\begin{remark}
As is pointed out in \cite{SIAM:Vadhan06},
auxiliary-input OWFs secure against non-uniform
PPT adversaries exist if
$\mathbf{CZK}\not\subseteq \mathbf{P}/\mathbf{poly}$.    
\end{remark}

\ifnum\llncs=1
\begin{definition}[Distributionally OWFs~\cite{FOCS:ImpLub89}]\label{def:dOWFs}
A polynomial-time-computable function $f:\bit^*\to\bit^*$ is distributionally one-way
if there exists a polynomial $p$ such that for any PPT
algorithm $\cA$ and all sufficiently large $\secp$, the statistical difference between 
$\{(x,f(x))\}_{x\gets\bit^\secp}$ and $\{(\cA(1^\secp,f(x)),f(x))\}_{x\gets\bit^\secp}$ is greater than $1/p(\secp)$.
\end{definition}

The following result is known.
\begin{lemma}[\cite{FOCS:ImpLub89}]
\label{lem:dOWFs}
If distributionally OWFs exist then OWFs exist.
\end{lemma}
\fi

\subsection{Commitments}
\begin{definition}[Statistically-hiding and computationally-binding classical bit commitments]
\label{def:commitments}
A statistically-hiding and computationally-binding classical bit commitment scheme is an interactive protocol $\langle \cS,\cR\rangle$
between two PPT algorithms $\cS$ (the sender) and $\cR$ (the receiver)
such that
\begin{itemize}
    \item 
    In the commit phase, 
    $\cS$ takes $b\in\bit$ and $1^\secp$ as input and $\cR$ takes $1^\secp$ as input. 
    $\cS$ and $\cR$ exchange classical messages.
    The transcript $t$, i.e., the sequence of all classical messages exchanged between $\cS$ and $\cR$, is called a commitment. At the end of the commit phase, $\cS$ privately outputs a decommitment $\mathsf{decom}$.  
    %$\cS$ takes $b\in\bit$ and $1^\secp$ as input .
    %$\cR$ takes $1^\secp$ as input.
    %$\cS$ and $\cR$ exchange classical messages.
    %$\cS$ outputs a bit string $\mathsf{com}$ (the commitment) and a bit string $\mathsf{decom}$ (the decommitment).
    %$\cR$ outputs $\mathsf{com}$.
   \item
   In the open phase, $\cS$ sends $(b,\mathsf{decom})$ to $\cR$. 
   $\cR$ on input $(t,b,\mathsf{decom})$ outputs $\top$ or $\bot$.
\end{itemize}
We require the following three properties.

\paragraph{\bf Perfect Correctness:}
For all $\secp\in \mathbb{N}$ and $b\in \bit$, if $\cS(1^\secp,b)$ and $\cR(1^\secp)$ behave honestly, $\Pr[\top\gets \cR]=1$. 

\paragraph{\bf Statistical Hiding:}
Let us consider the following security game between the honest sender $\cS$ and
a malicious receiver $\cR^*$:
\begin{enumerate}
    \item 
$\cS(b,1^\secp)$ and $\cR^*(1^\secp)$ run the commit phase.
    \item 
$\cR^*$ outputs $b'\in\bit$.
\end{enumerate}
We say that the scheme is statistically hiding if for any computationally unbounded adversary $\cR^*$,
$
|\Pr[0\gets\cR^*|b=0] 
-\Pr[0\gets\cR^*|b=1] 
|\le\negl(\secp).
$

\paragraph{\bf Computational Binding:}
Let us consider the following security game between a malicious sender $\cS^*$
and the honest receiver $\cR$:
\begin{enumerate}
    \item 
    $\cS^*(1^\secp)$ and $\cR(1^\secp)$ run the commit phase to generate a commitment $t$.
    \item
   $\cS^*$ sends $(0,\mathsf{decom}_0)$ and $(1,\mathsf{decom}_1)$ to $\cR$. 
\end{enumerate}
We say that the scheme is computationally binding if for any PPT malicious $\cS^*$,
$
    \Pr[\top\gets\cR(0,\mathsf{decom}_0)\wedge\top\gets\cR(1,\mathsf{decom}_1)]\le\negl(\secp).
    $
\end{definition}

%\begin{remark}
%Without loss of generality, we can take $c$ as the transcript (i.e., the set of all messages exchanged
%between $\cS$ and $\cR$) and
%$d$ as $\cS$'s random seed.
%\end{remark}

\if0
\begin{remark}
As the correctness, we adopt the perfect correctness,
because it is satisfied by OWFs-based constructions, and
the proofs of our results become simpler.
It could be possible to extend our results to non-perfect correct cases. \takashi{I don't think this remark is needed since perfect correctness is "standard" for commitments.}
\end{remark}
\fi

Statistically-hiding and computationally-binding bit commitments
can be constructed from OWFs.
\begin{theorem}[\cite{SIAM:HNORV09}]
\label{thm:commitment_from_OWFs}
If OWFs exist, then
statistically-hiding and computationally-binding bit commitments exist.
\end{theorem}

Moreover, constant-round schemes are known from 
collision-resistant hash functions~\cite{C:HalMic96}. 
The assumption is further weakened to the existence of distributional collision-resistant hash functions, which exist if there is an hard-on-average problem in $\mathbf{SZK}$. 
\begin{theorem}[\cite{C:KomYog18,EC:BHKY19}]\label{thm:const_commitment_from_dCRH}
If distributional collision-resistant hash functions exist, which exist if there is an hard-on-average problem in $\mathbf{SZK}$, then
 constant-round statistically-hiding and computationally-binding bit commitments exist.
\end{theorem}

We define an infinitely-often variant of statistically-hiding and computationally-binding commitments as follows.
\begin{definition}[Infinitely-often statistically-hiding and computationally-binding commitments]
\label{def:IOCom}
Infinitely-often statistically-hiding and computationally-binding commitments are defined similarly to \cref{def:commitments} except that we require the existence of an infinite set $\Lambda\subseteq \mathbb{N}$ such that
statistical hiding and computational binding hold for all $\secp\in \Lambda$ instead of for all $\secp\in \mathbb{N}$.  
\end{definition}
By using infinitely-often OWFs instead of OWFs in the commitment scheme of \cite{SIAM:HNORV09}, we obtain the following theorem. Since the construction and proof are almost identical to those of \cite{SIAM:HNORV09}, we omit the details.
\begin{theorem}[Infinitely-often variant of \cite{SIAM:HNORV09}]
\label{thm:IOcommitment_from_IOOWFs}
If infinitely-often OWFs exist, then
infinitely-often statistically-hiding and computationally-binding bit commitments exist.
\end{theorem}

We also define an
auxiliary-input variant of statistically-hiding and computationally-binding commitments. Intuitively, it is a family of commitment schemes indexed by an auxiliary input where correctness and statistical hiding hold for all auxiliary inputs and an ``auxiliary-input'' version of computational binding holds, i.e., 
for any PPT cheating sender $\cS^*$, there is an infinite set of auxiliary inputs under which computational binding holds. 
\begin{definition}[Auxiliary-input statistically-hiding and computationally-binding classical bit commitments]
\label{def:AXcommitments}
An auxiliary-input statistically-hiding and computationally-binding classical bit commitment scheme is an interactive protocol $\langle \cS,\cR\rangle$
between two PPT algorithms $\cS$ (the sender) and $\cR$ (the receiver) associated with an infinite subset $\Sigma\subseteq \bit^*$such that
\begin{itemize}
   \item 
    In the commit phase, 
    $\cS$ takes $b\in\bit$ and the auxiliary input $\sigma\in \Sigma$ as input and $\cR$ takes the auxiliary input $\sigma$ as input. 
    $\cS$ and $\cR$ exchange classical messages.
    The transcript $t$, i.e., the sequence of all classical messages exchanged between $\cS$ and $\cR$, is called a commitment. At the end of the commit phase, $\cS$ privately outputs a decommitment $\mathsf{decom}$.  
   \item
   In the open phase, $\cS$ sends $(b,\mathsf{decom})$ to $\cR$. 
   $\cR$ on input $(t,b,\mathsf{decom})$ outputs $\top$ or $\bot$.
\end{itemize}
We require the following properties:

\paragraph{\bf Perfect Correctness:}
For all $\sigma\in \Sigma$ and $b\in \bit$, 
if $\cS(b,\sigma)$ and $\cR(\sigma)$ behave honestly, $\Pr[\top\gets\cR]=1$.\\

\paragraph{\bf Statistical Hiding:}
Let us consider the following security game between the honest sender $\cS$ and
a malicious receiver $\cR^*$:
\begin{enumerate}
    \item 
$\cS(b,\sigma)$ and $\cR^*(\sigma)$ run the commit phase.
    \item 
$\cR^*$ outputs $b'\in\bit$.
\end{enumerate}
We say that the scheme is statistically hiding if for all $\sigma \in \Sigma$ and
any computationally unbounded adversary $\cR^*$,
$
|\Pr[0\gets\cR^*|b=0] 
-\Pr[0\gets\cR^*|b=1] 
|\le\negl(|\sigma|).
$

\paragraph{\bf Computational Binding:}
Let us consider the following security game between a malicious sender $\cS^*$
and the honest receiver $\cR$:
\begin{enumerate}
    \item 
    $\cS^*(\sigma)$ and $\cR(\sigma)$ run the commit phase to generate a commitment $t$.
    \item
   $\cS^*$ sends $(0,\mathsf{decom}_0)$ and $(1,\mathsf{decom}_1)$ to $\cR$. 
\end{enumerate}
We say that the scheme is computationally binding if for any PPT malicious sender $\cS^*$ and a polynomial $\poly$, there exists an infinite subset $\Lambda\subseteq \Sigma$ such that 
for any $\sigma \in \Lambda$, 
$
    \Pr[\top\gets\cR(t,0,\mathsf{decom}_0)\wedge\top\gets\cR(t,1,\mathsf{decom}_1)]\le\frac{1}{\poly(|\sigma|)}.
    $
%\takashi{slightly modified the definition}
\end{definition}

By using auxiliary-input OWFs instead of OWFs in the commitment scheme of \cite{SIAM:HNORV09}, we obtain the following theorem. Since the construction and proof are almost identical to those of \cite{SIAM:HNORV09}, we omit the details.
\begin{theorem}[Auxiliary-input variant of \cite{SIAM:HNORV09}]
\label{thm:AXcommitment_from_AXOWFs}
If auxiliary-input OWFs exist, then
auxiliary-input statistically-hiding and computationally-binding bit commitments exist.
\end{theorem}

Similarly, by using auxiliary-input collision-resistant hash functions instead of collision-resistant hash functions in the commitment scheme of \cite{C:HalMic96}, we obtain 2-round auxiliary-input statistically-hiding and computationally-binding bit commitments. 
As shown in \cref{sec:AXCRH}, auxiliary-input collision-resistant hash functions exist if and only if $\mathbf{PWPP}\nsubseteq \mathbf{FBPP}$. Thus, we obtain the following theorem.
\begin{theorem}[Auxiliary-input variant of \cite{C:HalMic96}]
\label{thm:const_AXcommitment_from_AICRH}
If auxiliary-input collision-resistant hash functions exist, which exist if and only if $\mathbf{PWPP}\nsubseteq \mathbf{FBPP}$, then
2-round auxiliary-input statistically-hiding and computationally-binding bit commitments exist.
\end{theorem}
In addition, we observe in  \cref{sec:commitment_from_worst_case_SZK} that the instance-dependent commitments for $\mathbf{SZK}$ of \cite{TCC:OngVad08} directly gives constant-round auxiliary-input statistically-hiding and computationally-binding bit commitments under the assumption that $\mathbf{SZK}\nsubseteq \mathbf{BPP}$.
\begin{theorem}[Auxiliary-input variant of \cite{TCC:OngVad08}]
\label{thm:const_AXcommitment_from_SZK}
If $\mathbf{SZK}\nsubseteq \mathbf{BPP}$, then
constant-round auxiliary-input statistically-hiding and computationally-binding bit commitments exist.
\end{theorem}
\begin{remark}\label{rem:Sigma_AI_com}
In the constructions for \cref{thm:AXcommitment_from_AXOWFs,thm:const_AXcommitment_from_AICRH}, we can set $\Sigma\seteq \bit^*$.  %i.e., correctness and statistical hiding  hold for all auxiliary-inputs $z\in \bit^*$. 
However, we do not know if this is possible for the construction for \cref{thm:const_AXcommitment_from_SZK} given in \cref{sec:commitment_from_worst_case_SZK}. This is why we introduce the subset $\Sigma$ in \cref{def:AXcommitments}. 
\end{remark}

%% file: interactive_supermacy.tex
\section{Inefficient-Verifier Proofs of Quantumness}
\label{sec:interactive_supremacy}
In this section, we define inefficient-verifier proofs of quantumness (IV-PoQ) and its variants.
Then we show that sequential repetition amplifies the completeness-soundness gap assuming a special
property of soundness, which we call a strong soundness, for the base scheme. 

\subsection{Definitions}
We define IV-PoQ. It is identical to the definition of PoQ, which are implicitly defined in \cite{JACM:BCMVV21}, except that we allow the verifier to be unbounded-time after completing interaction with the prover. 
\begin{definition}[Inefficient-verifier proofs of quantumness (IV-PoQ)]
\label{def:PoQinefficient}
An inefficient-verifier proof of quantumness (IV-PoQ)
is an
%interactive protocol $(\cV,\cP)$ between
%an algorithm $\cV$ (the verifier) and 
%a QPT algorithm $\cP$ (the prover)
interactive protocol $(\cP,\cV)$ between a QPT algorithm $\cP$ (the prover)
and an algorithm $\cV=(\cV_1,\cV_2)$ (the verifier) where $\cV_1$ is PPT and $\cV_2$ is unbounded-time.
The protocol is divided into two phases.
In the first phase, 
$\cP$ and $\cV_1$ take the security parameter $1^\secp$ as input and interact with each other over a classical channel.
%$\cV$ is PPT, and interacts with $\cP$ over a classical channel.
%(We write $\cV$ in the first phase as $\cV_1$.)
%$\cV_1$ takes the security parameter $1^\secp$ as input and outputs nothing.
%$\cR$ takes $1^\secp$ as input and outputs nothing.
Let $I$ be the transcript, i.e., the sequence of all classical messages exchanged between $\cP$ and $\cV_1$.
In the second phase, $\cV_2$ takes $I$ as input and outputs $\top$ or $\bot$.
%$\cV$ becomes inefficient.
%(We write $\cV$ in the second phase as $\cV_2$.)
%$\cV_2$ runs on input $I$,
%and outputs $\top$ or $\bot$.
We require the following two properties for some functions $c$ and $s$ such that $c(\secp)-s(\secp)\ge1/\poly(\secp)$.

\paragraph{\bf $c$-completeness:}
\begin{align}
\Pr[\top\gets\cV_2(I):I\gets\langle \cP(1^\secp),\cV_1(1^\secp)\rangle]\ge c(\secp)-\negl(\secp).
\end{align}

\paragraph{\bf $s$-soundness:}
For any PPT malicious prover $\cP^*$,
\begin{align}
\Pr[\top\gets\cV_2(I):I\gets\langle \cP^*(1^\secp),\cV_1(1^\secp)\rangle]\le s(\secp)+\negl(\secp).
\end{align}
\end{definition}

\begin{remark}
$\cV_2$ could take $\cV_1$'s secret information as input in addition to $I$, but without loss of generality,
we can assume that $\cV_2$ takes only $I$, because we can always modify the protocol of the first phase
in such a way that $\cV_1$ sends its secret information to $\cP$ at the end of the first phase.
\end{remark}

\begin{remark}
In our constructions, $\cV_2$ is actually enough to be a classical deterministic polynomial-time
algorithm that queries the $\mathbf{NP}$ oracle.    
\ifnum\cameraready=1
(See the full version.)
\else
(See \cref{sec:powerofV2}.)
\fi
\end{remark}

\begin{remark}
In the definition of soundness, we treat the malicious prover as a \emph{uniform} PPT machine. 
However, most of our results can be easily extended to the \emph{non-uniform} adversarial setting if we analogously strengthen the assumption (e.g., OWFs) to be secure against non-uniform PPT adversaries. 
\ifnum\cameraready=1
\else
The only place where the difference between uniform and non-uniform adversaries matters is \Cref{sec:QEOWFs} where we construct variants of IV-PoQ from quantum-evaluation OWFs. (See \Cref{rem:uniform_non-uniform} for details.)
\fi
\end{remark}

\ifnum\llncs=0
\input{IO_and_AI_PoQ}
\else
\ifnum\cameraready=1
We define infinitely-often and auxiliary-input variants of IV-PoQ in the full version.
\else
We define infinitely-often and auxiliary-input variants of IV-PoQ in \Cref{sec:def_IO_AI}
\fi
\fi

\subsection{Strong Soundness}
Unfortunately, we do not know if parallel or even sequential repetition amplifies the completeness-soundness gap for general (AI-/IO-)IV-PoQ. 
Here, we define a stronger notion of soundness which we call strong soundness. 
\ifnum\cameraready=1
In the full version, we show that sequential repetition amplifies the completeness-soundness gap if the base scheme satisfies strong soundness. 
\else
In \cref{sec:gap_amplification}, we show that sequential repetition amplifies the completeness-soundness gap if the base scheme satisfies strong soundness. 
\fi
In \cref{sec:construction}, we show that our (AI-/IO-)IV-PoQ satisfies strong soundness. Thus, gap amplification by sequential repetition works for our particular constructions of (AI-/IO-)IV-PoQ.  

Roughly, the $s$-strong-soundness requires that a PPT cheating prover can pass verification with probability at most $\approx s$ for {\it almost all fixed randomness}. The formal definition is given below.
\begin{definition}[Strong soundness for IV-PoQ]\label{def:strong_soundness}
We say that an IV-PoQ $(\cP,\cV=(\cV_1,\cV_2))$ satisfies $s$-strong-soundness if the following holds:
\paragraph{\bf $s$-strong-soundness:}
For any PPT malicious prover $\cP^*$ and any polynomial $p$, 
\begin{align}
\Pr_{r\gets \mathcal{R}}\left[\Pr[\top\gets\cV_2(I):I\gets\langle \cP^*_r(1^\secp),\cV_1(1^\secp)\rangle]\ge s(\secp)+\frac{1}{p(\secp)}\right]\le \frac{1}{p(\secp)}
\end{align}
for all sufficiently large $\secp$ 
where $\mathcal{R}$ is the randomness space for $\cP^*$ and $\cP^*_r$ is $\cP^*$ with the fixed randomness $r$. 
\end{definition}
It is defined similarly for (AI-/IO-)IV-PoQ. 
%\takashi{I commented out the formal definition because I believe that is clear.}
\if0
\begin{definition}[Strong soundness for AI-IV-PoQ]\label{def:AXstrong_soundness}
We say that an AI-IV-PoQ $(\cP,\cV=(\cV_1,\cV_2))$ satisfies $s$-strong-soundness if the following holds:
\paragraph{\bf $s$-strong-soundness:}
For any PPT malicious prover $\cP^*$  and polynomial $p$, there exists an infinite subset $\Lambda\subseteq \Sigma$ such that  
\begin{eqnarray*}
\Pr_{r\gets \mathcal{R}}\left[\Pr[\top\gets\cV_2(I):I\gets\langle \cP^*_r(\sigma),\cV_1(\sigma)\rangle]\ge s(|\sigma|)+\frac{1}{p(|\sigma|)}\right]\le \frac{1}{p(|\sigma|)}
\end{eqnarray*}
for all $\sigma\in \Lambda$
where $\mathcal{R}$ is the randomness space for $\cP^*$ and $\cP^*_r$ is $\cP^*$ with the fixed randomness $r$. 
\end{definition}
\fi

It is easy to see that $s$-strong-soundness implies $s$-soundness. 
\begin{lemma}\label{lem:strong_to_soundness}
For any $s$, $s$-strong-soundness implies $s$-soundness for (AI-/IO-)IV-PoQ. 
\end{lemma}
\begin{proof}
We focus on the case of IV-PoQ since the cases of (AI-/IO-)IV-PoQ are similar. 
If there is a PPT malicious prover $\cP^*$ that breaks $s$-soundness of IV-PoQ, then there exists a polynomial $p$ such that  
\begin{align}
\Pr[\top\gets\cV_2(I):I\gets\langle \cP^*(1^\secp),\cV_1(1^\secp)\rangle]\ge s(\secp)+\frac{3}{p(\secp)}
\end{align}
for infinitely many $\secp$. By a standard averaging argument, this implies 
\begin{align}
\Pr_{r\gets \mathcal{R}}\left[\Pr[\top\gets\cV_2(I):I\gets\langle \cP^*_r(1^\secp),\cV_1(1^\secp)\rangle]\ge s(\secp)+\frac{1}{p(\secp)}\right]\ge \frac{2}{p(\secp)}
\end{align}
for infinitely many $\secp$.
This contradicts $s$-strong-soundness. Thus, \cref{lem:strong_to_soundness} holds.
\end{proof}
We remark that the other direction does not seem to hold. 
For example, suppose that a PPT malicious prover $\cP^*$ passes the verification with probability $1$ for $0.99$-fraction of randomness and with probability $0$ for the rest of randomness. 
In this case, $\cP^*$ satisfies $0.99$-soundness. On the other hand, it breaks $s$-strong-soundness for any constant $s<1$.

\subsection{Gap Amplification}\label{sec:gap_amplification}

\begin{figure}[!t]
\begin{algorithm}[H]
\caption{$N$-sequential-repetition version $\Pi^{\Nseq}$ of $\Pi$}
\label{protocol:sequential_PoQ}
\begin{enumerate}
\item[\bf The first phase:] 
The QPT prover $\cP^{\Nseq}$ and PPT verifier $\cV_1^{\Nseq}$ 
run the first phase of $\Pi$ sequentially $N$ times (i.e., after they finish $i$-th execution, they start $(i+1)$-st execution) 
where $\cP^{\Nseq}$ plays the role of $\cP$ and $\cV_1^{\Nseq}$ plays the role of $\cV_1$. 
Let $I_i$ be the transcript of $i$-th execution of $\Pi$ for $i\in [N]$.
\item[\bf The second phase:]
The unbounded-time $\cV_2^{\Nseq}$ takes the transcript $\{I_i\}_{i\in [N]}$ of the first phase as input. 
For $i\in [N]$, 
it runs $\cV_2$ on $I_i$ and sets $X_i\seteq 1$ if it accepts and otherwise sets $X_i\seteq 0$. 
If $\frac{\sum_{i\in [N]}X_i}{N}\ge \frac{c(\secp)+s(\secp)}{2}$, 
it outputs $\top$ and otherwise outputs $\bot$. 
\end{enumerate}
\end{algorithm}
\end{figure}

We prove that sequential repetition amplifies the completeness-soundness gap if the base scheme satisfies strong soundness.
\begin{theorem}[Gap amplification theorem]\label{thm:gap_amplification_sequential}
Let $\Pi=(\cP,\cV=(\cV_1,\cV_2))$ be an (AI-/IO-)IV-PoQ that satisfies $c$-completeness and $s$-strong-soundness where $c(\secp)-s(\secp)\ge 1/\poly(\secp)$ and $c$ and $s$ are computable in polynomial-time. 
Let $\Pi^{\Nseq}=(\cP^{\Nseq},\cV^{\Nseq}=(\cV_1^{\Nseq},\cV_2^{\Nseq}))$ be its $N$-sequential-repetition version as described in \cref{protocol:sequential_PoQ}. 
If $N\geq \frac{\secp}{(c(\secp)-s(\secp))^{2}}$, then $\Pi^{\Nseq}$ satisfies $1$-completeness and $0$-soundness. 
\end{theorem}
\ifnum\cameraready=1
Its proof is given in the full version.
\else
Its proof is given in \cref{sec:gap_amplification}.
\fi

\begin{remark}
Note that the meaning of ``sequential repetition'' is slightly different from that for usual interactive arguments: we defer the second phases of each execution to the end of the protocol so that inefficient computations are only needed after completing the interaction. 
\end{remark}
\begin{remark}
If we assume $s$-soundness against \emph{non-uniform} PPT adversaries, we can easily prove a similar amplification theorem without introducing strong soundness.
However, for proving soundness against non-uniform PPT adversaries, we would need non-uniform hardness assumptions such as non-uniformly secure OWFs. 
Since our motivation is to demonstrate quantum advantage from the standard notion of \emph{uniformly} secure OWFs, we do not take the above approach.  
\end{remark}

%% file: IO_and_AI_PoQ.tex
\ifnum\llncs=1
\section{Definitions of Infinitely-Often and Auxiliary-Input IV-PoQ}\label{sec:def_IO_AI}
\fi
We define an infinitely-often version of IV-PoQ as follows.
\begin{definition}[Infinitely-often inefficient-verifier proofs of quantumness (IO-IV-PoQ)]
\label{def:IOPoQinefficient}
An infinitely-often inefficient-verifier proofs of quantumness (IO-IV-PoQ) is defined similarly to IV-PoQ (\cref{def:PoQinefficient}) except that we require the existence of an infinite set $\Lambda\subseteq \mathbb{N}$ such that
$c$-completeness and $s$-soundness hold for all $\secp\in \Lambda$ instead of for all $\secp\in \mathbb{N}$.  
%\takashi{I added the definition.}
\end{definition}
 
We also define an auxiliary-input variant of IV-PoQ as follows. It is defined similarly to IV-PoQ except that the prover and verifier take an auxiliary input instead of the security parameter and completeness should hold for all auxiliary inputs wheres soundness is replaced to auxiliary-input soundness, i.e., for any PPT cheating prover $\cP^*$, there exists an infinite set of auxiliary inputs under which soundness holds. 
\begin{definition}[Auxiliary-input inefficient-verifier proofs of quantumness (AI-IV-PoQ)]
\label{def:AXPoQinefficient}
An auxiliary-input inefficient-verifier proof of quantumness (AI-IV-PoQ)
is an
%interactive protocol $(\cV,\cP)$ between
%an algorithm $\cV$ (the verifier) and 
%a QPT algorithm $\cP$ (the prover) associated with an infinite subset $\Sigma\subseteq \bit^*$.
interactive protocol $(\cP,\cV)$ between a QPT algorithm $\cP$ (the prover)
and an algorithm $\cV=(\cV_1,\cV_2)$ (the verifier) where $\cV_1$ is PPT and $\cV_2$ is unbounded-time, associated with an infinite set $\Sigma\subseteq \bit^*$.
The protocol is divided into two phases.
In the first phase, 
$\cP$ and $\cV_1$ take an auxiliary input $\sigma \in \Sigma$ as input and interact with each other over a classical channel.
%$\cV$ is PPT, and interacts with $\cP$ over a classical channel.
%(We write $\cV$ in the first phase as $\cV_1$.)
%$\cV_1$ takes the security parameter $1^\secp$ as input and outputs nothing.
%$\cR$ takes $1^\secp$ as input and outputs nothing.
Let $I$ be the transcript, i.e., the sequence of all classical messages exchanged between $\cP$ and $\cV_1$.
In the second phase, $\cV_2$ takes $I$ as input and outputs $\top$ or $\bot$.
%$\cV$ becomes inefficient.
%(We write $\cV$ in the second phase as $\cV_2$.)
%$\cV_2$ runs on input $I$,
%and outputs $\top$ or $\bot$.
We require the following two properties for some functions 
$c$ and $s$ such that $c(|\sigma|)-s(|\sigma|)\ge1/\poly(|\sigma|)$.

\paragraph{\bf $c$-completeness:}
For any $\sigma\in\Sigma$, 
\begin{align}
\Pr[\top\gets\cV_2(I):I\gets\langle \cP(\sigma),\cV_1(\sigma)\rangle]\ge c(|\sigma|)-\negl(|\sigma|).
\end{align}

\paragraph{\bf $s$-soundness:}
For any PPT malicious prover $\cP^*$ and polynomial $p$, there exists an infinite set $\Lambda\subseteq\Sigma$ 
such that
\begin{align}
\Pr[\top\gets\cV_2(I):I\gets\langle\cP^*(\sigma), \cV_1(\sigma)\rangle]\le s(|\sigma|)+\frac{1}{p(|\sigma|)}
\end{align}
for all $\sigma\in\Lambda$. 
%\takashi{slightly modified the definition}
\end{definition}
\begin{remark}
We can set $\Sigma\seteq \bit^*$ for all our constructions of AI-IV-PoQ except for the one based on $\mathbf{SZK}\neq \mathbf{BPP}$. See also \cref{rem:Sigma_AI_com}.
\end{remark}
\begin{remark}
It is easy to see that IV-PoQ imply IO-IV-PoQ, and IO-IV-PoQ imply AI-IV-PoQ.
\end{remark}
Even though AI-IV-PoQ is weaker than IV-PoQ, we believe that it still demonstrates a meaningful notion of quantum advantage, because
it shows ``worst-case quantum advantage'' in the sense that no PPT algorithm can simulate the QPT honest prover on all auxiliary inputs $\sigma\in \Sigma$.   

%% file: coherent.tex
\section{Coherent Execution of Classical Bit Commitments}
\label{sec:coherent}
In this section, we explain our key concept, namely, executing classical bit commitments coherently.

Let $(\cS,\cR)$ be a classical (interactive) bit commitment scheme.
If we explicitly consider the randomness, the commit phase can be described
as in \cref{protocol:classical}.
\begin{figure}[!t]
\begin{algorithm}[H]
\caption{The commit phase of a classical bit commitment scheme}
\label{protocol:classical}
\begin{enumerate}
\item
The sender $\cS$ takes the comitted bit $b\in\bit$ and the security parameter $1^\secp$ as input.
The receiver $\cR$ takes $1^\secp$ as input.
    \item 
    $\cR$ samples a random seed $r\gets\bit^\ell$, where $\ell\coloneqq\poly(\secp)$.
    \item
    $\cS$ samples a random seed $x\gets\bit^\ell$. (Without loss of generality, we assume that
    $\cR$'s random seed and $\cS$'s random seed are of equal length.)
%    \item 
%    $\cS$ computes $\alpha_1\coloneqq f_1(b,x)$, and sends $\alpha_1$ to $\cR$.
%    Here, $f_1$ is the function that computes $\cS$'s first message. 
%     \item
%    $\cR$ computes $\beta_1\coloneqq g_1(r,\alpha_1)$, and sends $\beta_1$ to $\cS$.
%    Here, $g_1$ is the function that computes $\cR$'s first message.
    \item
    Let $L$ be the number of rounds.
    For $j=1$ to $L$, $\cS$ and $\cR$ repeat the following.
    \begin{enumerate}
        \item 
    $\cS$ computes $\alpha_j\coloneqq f_j(b,x,\alpha_1,\beta_1,\alpha_2,\beta_2,...,\alpha_{j-1},\beta_{j-1})$, and sends $\alpha_j$ to $\cR$.
    Here, $f_j$ is the function that computes $\cS$'s $j$th message. 
     \item
    $\cR$ computes $\beta_j\coloneqq g_j(r,\alpha_1,\beta_1,\alpha_2,\beta_2,...,\alpha_j)$, and sends $\beta_j$ to $\cS$.
    Here, $g_j$ is the function that computes $\cR$'s $j$th message.
    \end{enumerate}
\end{enumerate}
\end{algorithm}
\end{figure}

\begin{figure}[!t]
\begin{algorithm}[H]
\caption{Coherent execution of the commit phase of a classical bit commitment scheme}
\label{protocol:coherent}
\begin{enumerate}
\item
The sender $\cS$ takes $1^\secp$ as input.
The receiver $\cR$ takes $1^\secp$ as input.
    \item 
    $\cR$ samples a random seed $r\gets\bit^\ell$. 
    \item
    $\cS$ generates the state $\sum_{b\in\bit}\sum_{x\in\bit^\ell}\ket{b}\ket{x}$.
    \item
    Let $L$ be the number of rounds. For $j=1$ to $L$, $\cS$ and $\cR$ repeat the following.
    \begin{enumerate}
        \item 
         $\cS$ possesses the state
    \begin{equation}
    \sum_{b\in\bit}\sum_{x\in \cap_{i=0}^{j-1}X_{b}^i}\ket{b}\ket{x}, 
    \end{equation}
    where 
    \begin{align}
     X_{b}^0&\coloneqq\bit^\ell, 
    \end{align}
    and
    \begin{equation}
     X_{b}^i\coloneqq\{x\in\bit^\ell:f_i(b,x,\alpha_1,\alpha_2,...,\alpha_{i-1},\beta_{i-1})=\alpha_i\}   
    \end{equation}
    for $i\ge1$.
\item
      $\cS$ generates the state
    \begin{equation}
    \sum_{b\in\bit}\sum_{x\in \cap_{i=0}^{j-1}X_{b}^i}\ket{b}\ket{x}\ket{f_j(b,x,\alpha_1,\beta_1,...,\alpha_{j-1},\beta_{j-1})}, 
    \end{equation}
    and measures the third
    register in the computational basis to obtain the measurement result $\alpha_j$.
    The post-measurement state is 
     \begin{equation}
    \sum_{b\in\bit}\sum_{x\in \cap_{i=0}^j X_{b}^i}\ket{b}\ket{x}. 
    \end{equation}
    $\cS$ sends $\alpha_j$ to $\cR$.
    \item
    $\cR$ computes $\beta_j\coloneqq g_j(r,\alpha_1,\beta_1,...,\alpha_j)$, and sends $\beta_j$ to $\cS$.
    \end{enumerate}
\end{enumerate}
\end{algorithm}
\end{figure}

Now let us consider the coherent execution of \cref{protocol:classical}, which is shown in \cref{protocol:coherent}.
Let $t\coloneqq(\alpha_1,\beta_1,\alpha_2,\beta_2,...,\alpha_L,\beta_L)$ be the transcript
obtained in the execution of \cref{protocol:coherent}.
At the end of the execution of \cref{protocol:coherent},
$\cS$ possesses the state
\begin{equation}
\frac{1}{\sqrt{|X_{0,t}|+|X_{1,t}|}}\sum_{b\in\bit}\sum_{x\in X_{b,t}}\ket{b}\ket{x},
\end{equation}
where 
\begin{align}
X_{b,t}\coloneqq \bigcap_{j=0}^L X_b^j
=\Big\{x\in\bit^\ell:
\bigwedge_{j=1}^L
f_j(b,x,\alpha_1,\beta_1,...,\alpha_{j-1},\beta_{j-1})=\alpha_j
\Big\}.    
\end{align}
The probability $\Pr[t]$ that the transcript $t$ is obtained 
in the execution of \cref{protocol:coherent}
is
\begin{equation}
   \Pr[t]=\frac{|R_t|}{2^\ell}\frac{|X_{0,t}|+|X_{1,t}|}{2^{\ell+1}},
   \label{pr_coherent_t}
\end{equation}
where
\begin{equation}
R_t\coloneqq\Big\{r\in\bit^\ell:\bigwedge_{j=1}^L
g_j(r,\alpha_1,\beta_1,...,\alpha_j)=\beta_j\Big\}.
\end{equation}

In the remaining of this section, we show two lemmas, \cref{lem:equalS} and \cref{lem:kinji}, that will be used in the proofs of our main results.

The following \cref{lem:equalS} roughly claims that $|X_{0,t}|$ and $|X_{1,t}|$ are almost equal with overwhelming probability.
\begin{lemma}
\label{lem:equalS}
Let $0<\epsilon<1$ be a constant.
Define the set
\begin{equation}
T\coloneqq\{t:
(1-\epsilon)|X_{1,t}|<|X_{0,t}|<(1+\epsilon)|X_{1,t}|
\}.
\end{equation}
Then,
\begin{equation}
\sum_{t\in T}
\Pr[t]
\ge1-\negl(\secp).
\end{equation}
Here, $\Pr[t]$ is the
probability that the transcript $t$ is obtained 
in the execution of \cref{protocol:coherent} as is given in \cref{pr_coherent_t}.
\end{lemma}

\begin{proof}[Proof of \cref{lem:equalS}]
Intuitively, this follows from the statistical hiding property since whenever $t\notin T$,  
an unbounded adversary can guess the committed bit from the transcript $t$ with probability $1/2+\Omega(\epsilon)$.
Below, we provide a formal proof. 

Define 
\begin{align}
T^+&\coloneqq\{t:
(1+\epsilon)|X_{1,t}|\le |X_{0,t}|
\},\\
T^-&\coloneqq\{t:
|X_{0,t}|\le(1-\epsilon)|X_{1,t}|
\}.
\end{align}
In order to show the lemma, we want to show that
$
\sum_{t\in T^+\cup T^-}
\Pr[t]\le\negl(\secp).
$
To show this, assume that    
$
\sum_{t\in T^+\cup T^-}
\Pr[t]\ge\frac{1}{\poly(\secp)}
$
for infinitely many $\secp$.
Then the following computationally-unbounded malicious receiver $\cR^*$ can break the statistical hiding
of the classical bit commitment scheme in \cref{protocol:classical}.
\begin{enumerate}
\item
$\cR^*$ honestly executes the commit phase with $\cS$.
Let $t$ be the transcript obtained in the execution.
    \item 
    If $t\in T^+$, $\cR^*$ outputs 0.
    If $t\in T^-$, $\cR^*$ outputs 1.
    If $t\in T$, $\cR^*$ outputs $0$ with probability $1/2$ and outputs 1 with probability $1/2$.
\end{enumerate}
The probability that $\cR^*$ outputs 0 when $\cS$ commits $b\in\bit$ is
\begin{align}
\Pr[0\gets \cR^*|b]
=\sum_{t\in T^+} \Pr[t|b]
+\frac{1}{2}\sum_{t\in T} \Pr[t|b]
=\sum_{t\in T^+} \frac{|R_t|}{2^\ell}\frac{|X_{b,t}|}{2^\ell}
+\frac{1}{2}\sum_{t\in T} \frac{|R_t|}{2^\ell}\frac{|X_{b,t}|}{2^\ell}.
\end{align}
Therefore
\begin{align}
&\Pr[0\gets \cR^*|b=0]
-\Pr[0\gets \cR^*|b=1]\\
&=
\sum_{t\in T^+}\Pr[t|b=0]
+\frac{1}{2}\sum_{t\in T}\Pr[t|b=0]
-\sum_{t\in T^+}\Pr[t|b=1]
-\frac{1}{2}\sum_{t\in T}\Pr[t|b=1]\\
&=
\sum_{t\in T^+}\Pr[t|b=0]
+\frac{1}{2}\Big(1-\sum_{t\in T^+}\Pr[t|b=0]-\sum_{t\in T^-}\Pr[t|b=0]\Big)\\
&~~-\sum_{t\in T^+}\Pr[t|b=1]
-\frac{1}{2}\Big(1-\sum_{t\in T^+}\Pr[t|b=1]-\sum_{t\in T^-}\Pr[t|b=1]\Big)\\
&=
\frac{1}{2}\sum_{t\in T^+}(\Pr[t|b=0]-\Pr[t|b=1])
+\frac{1}{2}\sum_{t\in T^-}(\Pr[t|b=1]-\Pr[t|b=0])\\
&=
\frac{1}{2}\sum_{t\in T^+}\frac{|R_t|}{2^\ell}\frac{|X_{0,t}|-|X_{1,t}|}{2^\ell}
+\frac{1}{2}\sum_{t\in T^-}\frac{|R_t|}{2^\ell}\frac{|X_{1,t}|-|X_{0,t}|}{2^\ell}\\
&\ge
\frac{1}{2}\sum_{t\in T^+}\frac{|R_t|}{2^\ell}\frac{|X_{0,t}|-\frac{|X_{0,t}|}{1+\epsilon}}{2^\ell}
+\frac{1}{2}\sum_{t\in T^-}\frac{|R_t|}{2^\ell}\frac{|X_{1,t}|-(1-\epsilon)|X_{1,t}|}{2^\ell}\\
&=
\frac{1}{2}\Big(1-\frac{1}{1+\epsilon}\Big)\sum_{t\in T^+}\frac{|R_t|}{2^\ell}\frac{|X_{0,t}|}{2^\ell}
+\frac{\epsilon}{2}\sum_{t\in T^-}\frac{|R_t|}{2^\ell}\frac{|X_{1,t}|}{2^\ell}\\
&=
\frac{1}{2}\frac{\epsilon}{1+\epsilon}\sum_{t\in T^+}\frac{|R_t|}{2^\ell}\frac{2|X_{0,t}|}{2^{\ell+1}}
+\frac{\epsilon}{2}\sum_{t\in T^-}\frac{|R_t|}{2^\ell}\frac{2|X_{1,t}|}{2^{\ell+1}}\\
&\ge
\frac{1}{2}\frac{\epsilon}{1+\epsilon}\sum_{t\in T^+}\frac{|R_t|}{2^\ell}\frac{|X_{0,t}|+|X_{1,t}|}{2^{\ell+1}}
+\frac{\epsilon}{2}\sum_{t\in T^-}\frac{|R_t|}{2^\ell}\frac{|X_{0,t}|+|X_{1,t}|}{2^{\ell+1}}\\
&\ge
\frac{1}{2}\frac{\epsilon}{1+\epsilon}\sum_{t\in T^+}\frac{|R_t|}{2^\ell}\frac{|X_{0,t}|+|X_{1,t}|}{2^{\ell+1}}
+\frac{\epsilon}{2(1+\epsilon)}\sum_{t\in T^-}\frac{|R_t|}{2^\ell}\frac{|X_{0,t}|+|X_{1,t}|}{2^{\ell+1}}\\
&=
\frac{1}{2}\frac{\epsilon}{1+\epsilon}\sum_{t\in T^+\cup T^-}\frac{|R_t|}{2^\ell}\frac{|X_{0,t}|+|X_{1,t}|}{2^{\ell+1}}\\
&=
\frac{1}{2}\frac{\epsilon}{1+\epsilon}
\sum_{t\in T^+\cup T^-}\Pr[t]\\
&\ge
\frac{1}{2}\frac{\epsilon}{1+\epsilon}
\frac{1}{\poly(\secp)}\\
&=
\frac{1}{\poly(\secp)}
\end{align}
for infinitely many $\secp$, which breaks the statistical hiding.
\end{proof}

The following \cref{lem:kinji}
roughly claims that 
whenever $t\in T$
a good approximation $k$ of $2|X_{0,t}|$ and $2|X_{1,t}|$ up to a small constant multiplicative error 
can be chosen with probability $1/m=1/\poly(\secp)$.
\begin{lemma}
\label{lem:kinji}
Let
$0<\epsilon<1$ be a constant.
Let
\begin{equation}
T\coloneqq\{t:
(1-\epsilon)|X_{1,t}|<|X_{0,t}|<(1+\epsilon) |X_{1,t}|
\}.
\end{equation}
Let $m$ be an integer such that $(1+\epsilon)^m\ge 2^{\ell+1}$.
For any $t\in T$,
there exists an integer 
$j\in\{0,1,2,...,m-1\}$
such that
%$k\in\{1,\lfloor 1+\epsilon\rfloor,\lfloor (1+\epsilon)^2\rfloor,\lfloor(1+\epsilon)^3\rfloor
%,...,\lfloor(1+\epsilon)^{m-1}\rfloor\}$ such that
\begin{equation}
k\le 2|X_{0,t}|\le (1+\epsilon)k
\label{k1}
\end{equation}
and
\begin{equation}
\frac{k}{1+\epsilon}\le 2|X_{1,t}|\le \frac{1+\epsilon}{1-\epsilon}k.
\label{k2}
\end{equation}
Here, $k\coloneqq \lceil(1+\epsilon)^j\rceil$.
\end{lemma}

\begin{proof}[Proof of \cref{lem:kinji}]
Let $t\in T$.
Because $0\le (1-\epsilon)|X_{1,t}|<|X_{0,t}|$, 
we have $|X_{0,t}|\ge1$.
Because $2|X_{0,t}|\le2^{\ell+1}$, there exists an integer $j\in\{0,1,2,...,m-1\}$ such that
$(1+\epsilon)^j< 2|X_{0,t}|\le (1+\epsilon)^{j+1}$.
Let us take $k=\lceil(1+\epsilon)^{j}\rceil$.
Because $2|X_{0,t}|$ is an integer,
$k\le 2|X_{0,t}|$.
Moreover, we have
\begin{align}
2|X_{0,t}|\le (1+\epsilon)^{j+1}
=(1+\epsilon)\times(1+\epsilon)^j
\le(1+\epsilon)\times\lceil(1+\epsilon)^j\rceil
=(1+\epsilon)\times k.
\end{align}
In summary, we have
$k\le 2|X_{0,t}|\le (1+\epsilon)k$.
We also have    
$2|X_{1,t}|<\frac{2|X_{0,t}|}{1-\epsilon}\le \frac{1+\epsilon}{1-\epsilon}k$,
and
$2|X_{1,t}|>\frac{2|X_{0,t}|}{1+\epsilon}\ge \frac{k}{1+\epsilon}$.
\end{proof}

%% file: construction.tex
\section{Construction of IV-PoQ}
\label{sec:construction}
In this section, we prove \cref{thm:PoQ_from_Com}. That is, 
we construct IV-PoQ from statistically-hiding and computationally-binding classical bit commitments. 
%Then \cref{thm:PoQ_from_OWFs,thm:AXPoQ_from_OWFs,thm:PoQ_from_SZK,thm:AXPoQ_from_SZK} follow by instantiating (auxiliary-input) statistically-hiding and computationally-binding bit commitments 
%with the schemes of \cref{thm:commitment_from_OWFs,thm:const_commitment_from_dCRH,thm:AXcommitment_from_AXOWFs,thm:const_AXcommitment_from_AICRH,thm:const_AXcommitment_from_SZK}. 
%In the following, we focus on the plain (i.e., non-auxiliary-input) version because
%The proof of the auxiliary-input version is almost identical. 
%and show its completeness and soundness.
%We also show that our inefficient verifier needs
%only the power of $\mathbf{P}^{\mathbf{NP}}$.

Let $(\cS,\cR$) be a statistically-hiding and computationally-binding
classical (interactive) bit commitment scheme.
Note that this can be constructed from one-way functions 
\ifnum\cameraready=1
~\cite{SIAM:HNORV09}.
\else
(\cref{thm:commitment_from_OWFs})~\cite{SIAM:HNORV09}.
\fi
The first phase where the PPT verifier $\cV_1$ and the QPT prover $\cP$ interact
is given in \cref{protocol:PoQ1}.
The second phase where the inefficient verifier $\cV_2$ runs 
is given in \cref{protocol:PoQ2}.

\ifnum\llncs=0
\input{algorithms.tex} %In the full version, algorithms 4 and 5 are placed here
\fi

We can show the completeness and soundness of our IV-PoQ
as follows.

\begin{theorem}[Completeness]
\label{thm:completeness}
Our IV-PoQ satisfies $(\frac{7}{8}+\frac{1}{\poly(\secp)})$-completeness.  
%For the honest QPT prover $\cP$, $\cV_2$ accepts with probability at least
%$\frac{7}{8}+\frac{1}{\poly(\secp)}$.
\end{theorem}

\begin{theorem}[Soundness]
\label{thm:soundness}
Our IV-PoQ satisfies $\frac{7}{8}$-strong-soundness, which in particular implies $\frac{7}{8}$-soundness.
%For any PPT malicious prover,
%$\cV_2$ accepts with probability at most $\frac{7}{8}+\negl(\secp)$. 
\end{theorem}

\ifnum\llncs=1
\input{algorithms.tex} %In the LNCS version, algorithms 4 and 5 are placed here
\fi

The above theorems only gives an inverse-polynomial completeness-soundness gap, but we can amplify the gap to $1$ by sequential repetition by \cref{thm:gap_amplification_sequential}. 
\cref{thm:completeness} is shown in \cref{sec:completeness}.    
\cref{thm:soundness} is shown in \cref{sec:soundness}.    
By combining
\cref{thm:completeness} and
\cref{thm:soundness}, we obtain
\cref{thm:PoQ_from_Com}.

\if0
The completeness-soundness gap is $1/\poly(\secp)$ because
\begin{eqnarray}
&&(1-\negl(\secp))
\frac{0.01}{m}
\Big(\frac{1}{2}+\frac{1}{2}\cos^2\frac{\pi}{8}\Big)
+
(1-\negl(\secp))
\Big(1-\frac{0.01}{m}\Big)\frac{7}{8}
-\Big(\frac{7}{8}+\negl(\secp)\Big)\\
&=&
(1-\negl(\secp))
\frac{0.01}{m}
\Big(\frac{1}{2}+\frac{1}{2}\cos^2\frac{\pi}{8}-\frac{7}{8}\Big)
-\negl(\secp)\\
&=&
(1-\negl(\secp))
\frac{0.01\times0.025}{m}
-\negl(\secp)\\
&\ge&\frac{1}{\poly(\secp)}.
\end{eqnarray}
\fi

%% file: algorithms.tex
%\ifnum\llncs=0
%\begin{figure}[!t]
%\else
\begin{figure}[H]
%\fi
\begin{algorithm}[H]
\caption{The first phase by $\cV_1$ and $\cP$}
\label{protocol:PoQ1}
\begin{enumerate}
\item \label{step:coherent_commit}
The PPT verifier $\cV_1$ and the QPT prover $\cP$ coherently execute the commit phase of the classical bit commitment scheme $(\cS,\cR)$.
(See \cref{protocol:coherent}.)
$\cV_1$ plays the role of the receiver $\cR$.
$\cP$ plays the role of the sender $\cS$. 
    Let $t$ be the transcript obtained in the execution.
    \item 
    $\cP$ has the state 
    $\sum_{b\in\bit}\sum_{x\in X_{b,t}}
    \ket{b}\ket{x}$.
    \item\label{step:choose_hash}
    Let $0<\epsilon<1$ be a small constant. (We set $\epsilon\seteq \frac{1}{100}$ for clarity.)
    Let $m$ be an integer such that $(1+\epsilon)^m\ge 2^{\ell+1}$. (Such an $m$ can be computed in $\poly(\secp)$ time. 
    In fact, we have only to take the minimum integer $m$ such that $m\ge\frac{\ell+1}{\log_2 (1+\epsilon)}$.)
    $\cV_1$ chooses 
    $j\gets\{0,1,2,...,m-1\}$. Define $k\coloneqq \lceil (1+\epsilon)^j\rceil$.
%    (If $t\in T$, $k$ satisfies \cref{k1} and \cref{k2} of \cref{lem:kinji}.)
    Let $\mathcal{H}\coloneqq\{h:\cX\to\cY\}$
    be a pairwise-independent hash family with
    $\cX\coloneqq\bit^\ell$ and $\cY\coloneqq[k]$. 
    $\cV_1$ chooses $h_0,h_1\gets\mathcal{H}$,
    and sends $(h_0,h_1)$ to $\cP$.
    \item \label{step:send_y}
    $\cP$ changes its state into 
    $\sum_{b\in\bit}\sum_{x\in X_{b,t}}
    \ket{b}\ket{x}\ket{h_b(x)}$,
    and measures the third register in the computational basis to obtain the result $y\in[k]$.
    $\cP$ sends $y$ to $\cV_1$.
    The post-measurement state is
    \begin{equation}
    \sum_{b\in\bit}\sum_{x\in X_{b,t}\cap h_b^{-1}(y)}
    \ket{b}\ket{x}.
    \label{thestate}
    \end{equation}
    (If there is only a single $x_b$ such that $x_b\in X_{b,t}\cap h_b^{-1}(y)$ for each $b\in\bit$,
    \cref{thestate} is $\ket{0}\ket{x_0}+\ket{1}\ket{x_1}$. We will show later that it occurs with a non-negligible probability.)
%    \item
%    Assume 
%    that $|S'(0,r,t,h_0,y)|=|S'(1,r,t,h_1,y)|=1$.
%    (In \cref{thm:Sis1}, we show that it occurs with probability at least $\frac{1-\epsilon}{1+\epsilon}(0.2)^2$.)
%    Let $x_b$ be the unique element of $S'(b,r,t,h_b,y)$ for each $b\in\bit$.
%    Then, $\cP$'s state is
%    \begin{equation}
%    \frac{1}{\sqrt{2}}(\ket{0}\ket{x_0}+\ket{1}\ket{x_1}).
%    \end{equation}
    \item \label{step:send_v1_and_xi}
    From now on, $\cV_1$ and $\cP$ run the protocol of \cite{NatPhys:KMCVY22}.
    $\cV_1$ chooses $v_1\leftarrow \bit$. 
    $\cV_1$ chooses $\xi\leftarrow \bit^\ell$. 
    $\cV_1$ sends $v_1$ and $\xi$ to $\cP$. 
    \item \label{item:v1}
    \begin{itemize}
    \item\label{step:v_1_0_or_1}
    If $v_1=0$: $\cP$ measures all qubits of the state of \cref{thestate} in the computational basis,
    and sends the measurement result $(b',x')\in\bit\times\bit^\ell$ to $\cV_1$. 
    $\cV_1$ halts.
    \item
    If $v_1=1$: $\cP$ changes the state of \cref{thestate} into
    \begin{equation}
    \sum_{b\in\bit}\sum_{x\in X_{b,t}\cap h_b^{-1}(y)}\ket{b\oplus (\xi\cdot x)}\ket{x},
    \label{thestate2}
    \end{equation}
    measures its second register in the Hadamard basis to obtain the measurement result $d\in\{0,1\}^\ell$,
    and sends $d$ to $\cV_1$.
    The post-measurement state is
     \begin{equation}
    \sum_{b\in\bit}\sum_{x\in X_{b,t}\cap h_b^{-1}(y)}(-1)^{d\cdot x}\ket{b\oplus (\xi\cdot x)}.
    \label{thestate3}
    \end{equation}
    (If there is only a single $x_b$ such that $x_b\in X_{b,t}\cap h_b^{-1}(y)$ for each $b\in\bit$,
    \cref{thestate2} is $\ket{\xi\cdot x_0}\ket{x_0}+\ket{1\oplus(\xi\cdot x_1)}\ket{x_1}$,
    and \cref{thestate3} is
    $|\xi\cdot x_0\rangle+(-1)^{d\cdot(x_0\oplus x_1)}|1\oplus (\xi\cdot x_1)\rangle$.)
    \end{itemize}
    \item \label{step:send_v2}
    $\cV_1$ chooses $v_2\leftarrow\bit$. $\cV_1$ sends $v_2$ to $\cP$.
    \item \label{item:v2}
    If $v_2=0$, $\cP$ measures \cref{thestate3} in the basis 
$\left\{\cos\frac{\pi}{8}|0\rangle+\sin\frac{\pi}{8}|1\rangle,
\sin\frac{\pi}{8}|0\rangle-\cos\frac{\pi}{8}|1\rangle\right\}$.
    If $v_2=1$, $\cP$ measures \cref{thestate3} in the basis 
$\left\{\cos\frac{\pi}{8}|0\rangle-\sin\frac{\pi}{8}|1\rangle,
\sin\frac{\pi}{8}|0\rangle+\cos\frac{\pi}{8}|1\rangle\right\}$.
    Let $\eta\in\bit$ be the measurement result. (For the measurement in the basis $\{|\phi\rangle,|\phi^\perp\rangle\}$, the result 0 corresponds to $|\phi\rangle$ and the result 1 corresponds to $|\phi^\perp\rangle$.)
    $\cP$ sends $\eta$ to $\cV_1$.
\end{enumerate}
\end{algorithm}
\end{figure}

\begin{figure}[!t]
\begin{algorithm}[H]
\caption{The second phase by $\cV_2$}
\label{protocol:PoQ2}
\begin{enumerate}
\item
$\cV_2$ takes
$(t,h_0,h_1,y,v_1=0,\xi,b',x')$
or $(t,h_0,h_1,y,v_1=1,\xi,d,v_2,\eta)$ as input.
%   \item
%   If $k\neq k^*$, $\cV_2$ outputs $\top$ with probability $7/8$ and outputs $\bot$ with probability $1/8$. Then $\cV_2$ halts.
   \item\label{is_y_good}
  If $|X_{0,t}\cap h_0^{-1}(y)|=|X_{1,t}\cap h_1^{-1}(y)|=1$ is not satisfied, 
  $\cV_2$ outputs $\top$ with probability $7/8$ and outputs $\bot$ with probability $1/8$. Then $\cV_2$ halts.
  If $|X_{0,t}\cap h_0^{-1}(y)|=|X_{1,t}\cap h_1^{-1}(y)|=1$ is satisfied, 
  $\cV_2$ computes $(x_0,x_1)$ and goes to the next step.
  Here, $x_b$ is the single element of $X_{b,t}\cap h_b^{-1}(y)$ for each $b\in\bit$.
  \item 
   If $v_1=0$ and $x'=x_{b'}$, $\cV_2$ outputs $\top$ and halts.
    Otherwise, $\cV_2$ outputs $\bot$ and halts.
     \item
     If $v_1=1$,
    $\cV_2$ outputs $\top$ if 
    \begin{align}
    (\xi\cdot x_0\neq \xi\cdot x_1) \wedge
    (\eta=\xi\cdot x_0),
    \end{align}
    or
    \begin{align}
    (\xi\cdot x_0 = \xi\cdot x_1) \wedge
    (\eta=v_2\oplus d\cdot(x_0\oplus x_1)).
    \end{align}
    Otherwise, $\cV_2$ outputs $\bot$.    
\end{enumerate}
\end{algorithm}
\end{figure}

%% file: completeness.tex
\subsection{Completeness}
\label{sec:completeness}

In this subsection, we show $(\frac{7}{8}+\frac{1}{\poly(\secp)})$-completeness.

\begin{proof}[Proof of \cref{thm:completeness}]
Let $p_{\mathsf{good}}$ be the probability that
\cref{thestate} 
in the \cref{step:send_y} of \cref{protocol:PoQ1}
is in the form of $\ket{0}\ket{x_0}+\ket{1}\ket{x_1}$.
Then
$
p_{\mathsf{good}}\ge (1-\negl(\secp))\frac{0.1}{m}
$
because of the following reasons.
\begin{itemize}
    \item 
In Step \ref{step:coherent_commit} of \cref{protocol:PoQ1}, the probability that the transcript $t$ such that $t\in T$
is obtained is at least $1-\negl(\secp)$ from \cref{lem:equalS}
where $T$ is defined in \cref{lem:equalS}.
\item
Given $t\in T$, in Step \ref{step:choose_hash} of \cref{protocol:PoQ1}, the probability that $\cV_1$ chooses $j$ such that
$k$ satisfies \cref{k1} and \cref{k2} of \cref{lem:kinji} is $\frac{1}{m}$.
\item
Given that $k$ satisfies \cref{k1} and \cref{k2} of \cref{lem:kinji},
in Step \ref{step:send_y} of \cref{protocol:PoQ1}, the probability that $y$ such that
$|X_{0,t}\cap h_0^{-1}(y)|=|X_{1,t}\cap h_1^{-1}(y)|=1$
is obtained is at least 0.1 from \cref{lem:Sis1} shown below.
\end{itemize}
Moreover, 
if 
\cref{thestate} 
in Step \ref{step:send_y} of \cref{protocol:PoQ1}
is in the form of $\ket{0}\ket{x_0}+\ket{1}\ket{x_1}$,
the probability that $\cV_2$ outputs $\top$ in \cref{protocol:PoQ2} is 
$\frac{1}{2}+\frac{1}{2}\cos^2\frac{\pi}{8}\ge 0.9$
as is shown 
\ifnum\cameraready=1
in the full version.
\else
in \cref{app:completeness}. 
\fi

Therefore,
the probability that $\cV_2$ outputs $\top$ is
\begin{align}
\Big(\frac{1}{2}+\frac{1}{2}\cos^2\frac{\pi}{8}\Big)
p_{\mathsf{good}}
+
\frac{7}{8}
(1-p_{\mathsf{good}})
=\frac{7}{8}
+\Big(\frac{1}{2}+\frac{1}{2}\cos^2\frac{\pi}{8}-\frac{7}{8}\Big)
p_{\mathsf{good}}
\ge\frac{7}{8}+\frac{1}{\poly(\secp)},
\end{align}
which shows the completeness.
\end{proof}

\begin{lemma}
\label{lem:Sis1}
Assume that $k$ satisfies \cref{k1} and \cref{k2} of \cref{lem:kinji}.
In Step \ref{step:send_y} of \cref{protocol:PoQ1}, the probability that $y$ such that
$|X_{0,t}\cap h_0^{-1}(y)|=|X_{1,t}\cap h_1^{-1}(y)|=1$
is obtained is at least 0.1. 
\end{lemma}

\begin{proof}[Proof of \cref{lem:Sis1}]
\if0
From \cref{lem:hash},
for any $b\in\bit$ and $y\in[k]$,
\begin{eqnarray}
    \Pr_{h_b\gets\mathcal{H}}[|X_{b,t}\cap h_b^{-1}(y)|\ge1]&\ge&
    \frac{|X_{b,t}|}{k}-
\frac{|X_{b,t}|^2}{2 k^2}\\
&\ge&\frac{1}{2(1+\epsilon)}-\frac{(1+\epsilon)^2}{8(1-\epsilon)^2}.
\end{eqnarray}
\fi

\if0
Let us show it.
First,
\begin{eqnarray}
\sum_{j=1}^{2^\secp}j \Pr_{h_b\gets\mathcal{H}}[|S(b,r,t)\cap h_b^{-1}(y)|=j]    
&=&\frac{|S(b,r,t)|}{k}.
\end{eqnarray}
Second,
\begin{eqnarray}
&&\sum_{j=1}^{2^\secp}(j-1) \Pr_{h_b\gets\mathcal{H}}[|S(b,r,t)\cap h_b^{-1}(y)|=j]\\
&\le&
\sum_{j=2}^{2^\secp}{j\choose 2} \Pr_{h_b\gets\mathcal{H}}[|S(b,r,t)\cap h_b^{-1}(y)|=j]\\
&=&\frac{1}{k^2}{|S(b,r,t)|\choose2}\\
&=&\frac{|S(b,r,t)|(|S(b,r,t)|-1)}{2 k^2}\\
&\le&\frac{|S(b,r,t)|^2}{2 k^2}\\
&\le&\frac{1}{8}\left(\frac{1+\epsilon}{1-\epsilon}\right)^2.
\end{eqnarray}
By extracting both sides of the above two equations,
\begin{eqnarray}
\sum_{j=1}^{2^\secp}    
\Pr_{h_b\gets\mathcal{H}}[|S(b,r,t)\cap h_b^{-1}(y)|=j]    
&\ge&
\frac{|S(b,r,t)|}{k}-\frac{1}{8}\left(\frac{1+\epsilon}{1-\epsilon}\right)^2\\
&\ge&
\frac{1}{2(1+\epsilon)}-\frac{1}{8}\left(\frac{1+\epsilon}{1-\epsilon}\right)^2.
\end{eqnarray}
\fi

By using \cref{lem:hash2} with $S=X_{b,t}$, $h=h_b$, and $\cY=[k]$, we have, for any $b\in\bit$, $t$, and $y\in\cY$,
\begin{align}
   \Pr_{h_b\gets\mathcal{H}}[|X_{b,t}\cap h_b^{-1}(y)|=1]
   \ge
   \frac{|X_{b,t}|}{k}
   -\frac{|X_{b,t}|^2}{k^2}
   \ge
   \frac{1}{2(1+\epsilon)}-\frac{(1+\epsilon)^2}{4(1-\epsilon)^2}.
   \label{comp3}
\end{align}
Here, in the last inequality, we have used \cref{lem:kinji}.

In Step \ref{step:send_y} of \cref{protocol:PoQ1}, the probability that $y$ is obtained is
$
\frac{|X_{0,t}\cap h_0^{-1}(y)|+|X_{1,t}\cap h_1^{-1}(y)|}{|X_{0,t}|+|X_{1,t}|}.
$
Let us define 
\begin{equation}
G_{b,t,h_b}\coloneqq\{y\in[k]:|X_{b,t}\cap h_b^{-1}(y)|=1\}.    
\end{equation}
Then, the probability that we obtain $y$ such that
$|X_{0,t}\cap h_0^{-1}(y)|=|X_{1,t}\cap h_1^{-1}(y)|=1$
is
\begin{align}
&\underset{h_0,h_1\gets\mathcal{H}}{\mathbb{E}}\left[
\sum_{y\in G_{0,t,h_0}\cap G_{1,t,h_1}}
\frac{2}{|X_{0,t}|+|X_{1,t}|}\right]\\
&=
\underset{h_0,h_1\gets\mathcal{H}}{\mathbb{E}}\left[
\frac{2
|G_{0,t,h_0}\cap G_{1,t,h_1}|
}{|X_{0,t}|+|X_{1,t}|}\right]\\
&\ge
\underset{h_0,h_1\gets\mathcal{H}}{\mathbb{E}}\left[
\frac{2
|G_{0,t,h_0}\cap G_{1,t,h_1}|
}{\frac{(1+\epsilon)k}{1-\epsilon}}\right]\label{comp1}\\
&=
\frac{2(1-\epsilon)}{1+\epsilon}
\underset{h_0,h_1\gets\mathcal{H}}{\mathbb{E}}\left[
\frac{
|G_{0,t,h_0}\cap G_{1,t,h_1}|
}{k}\right]\\
&=
\frac{2(1-\epsilon)}{1+\epsilon}
\frac{1}{|\mathcal{H}|^2}
\sum_{h_0,h_1\in \mathcal{H}}
\frac{1}{k}\sum_{y\in [k]}
\delta_{y\in G_{0,t,h_0}}
\delta_{y\in G_{1,t,h_1}}\\
&=
\frac{2(1-\epsilon)}{1+\epsilon}
\frac{1}{k}\sum_{y\in [k]}
\Big(\frac{1}{|\mathcal{H}|}
\sum_{h_0\in \mathcal{H}}
\delta_{y\in G_{0,t,h_0}}\Big)
\Big(\frac{1}{|\mathcal{H}|}
\sum_{h_1\in \mathcal{H}}
\delta_{y\in G_{1,t,h_1}}\Big)\\
&=
\frac{2(1-\epsilon)}{1+\epsilon}
\frac{1}{k}\sum_{y\in [k]}
\Big(\Pr_{h_0\gets\mathcal{H}}[y\in G_{0,t,h_0}]\Big)
\Big(\Pr_{h_1\gets\mathcal{H}}[y\in G_{1,t,h_1}]\Big)\\
&\ge
\frac{2(1-\epsilon)}{1+\epsilon}
\frac{1}{k}\sum_{y\in[k]}
\Big[
 \frac{1}{2(1+\epsilon)}-\frac{(1+\epsilon)^2}{4(1-\epsilon)^2}
\Big]^2\label{comp2}\\
&=
\frac{2(1-\epsilon)}{1+\epsilon}
\Big[
 \frac{1}{2(1+\epsilon)}-\frac{(1+\epsilon)^2}{4(1-\epsilon)^2}
\Big]^2\\
&>0.1.
\end{align}
Here, $\delta_\alpha$ is 1 if the statement $\alpha$ is true, and is 0 if not.
In \cref{comp1}, we have used \cref{lem:kinji},
and in \cref{comp2}, we have used \cref{comp3}.
In the last inequality, we have taken $\epsilon=\frac{1}{100}$.

\end{proof}

%% file: soundness.tex
\subsection{Soundness}
\label{sec:soundness}

%For any PPT malicious prover, at the end of the coherent execution of the commitment scheme,the prover cannot know $(s_0,s_1)$ such that $s_0\in S(0,r,t)$ and $s_1\in S(1,r,t)$. This is shown from the computational binding. Then, by almost the same argument as that in \cite{NatPhys:KMCVY22},  we have
%\begin{eqnarray*} 
%\Pr[\cV\to\top]\leq \frac{7}{8}+\negl(\secp).
%\end{eqnarray*}

In this subsection, we show $\frac{7}{8}$-strong-soundness.
\begin{proof}[Proof of \cref{thm:soundness}]
Our goal is to prove that for any PPT malicious prover $\cP^*$ and any polynomial $p$, 
\begin{equation}\label{eq:goal}
\Pr_{r\gets \mathcal{R}}\left[\Pr[\top\gets\cV_2(I):I\gets\langle \cP^*_r(1^\secp),\cV_1(1^\secp)\rangle]\ge  \frac{7}{8}+\frac{1}{p(\secp)}\right]\le \frac{1}{p(\secp)}
\end{equation}
for sufficiently large $\secp$ 
where $\mathcal{R}$ is the randomness space for $\cP^*$ and $\cP^*_r$ is $\cP^*$ with the fixed randomness $r$. 

%Our goal is to prove that 
%for any PPT malicious prover,
%\begin{equation} \label{eq:goal}
%\Pr[\top\gets\cV_2]\leq \frac{7}{8}+\negl(\secp).
%\end{equation}
%We assume that a malicious prover is deterministic for simplicity. The case of a randomized malicious prover can be proven similarly by doing the same analysis for each fixed randomness and then taking average over the randomness. 
%without loss of generality.   
Toward contradiction, suppose that there are a 
PPT prover $\cP^*$ and a polynomial $p$ such that 
\begin{align}\label{eq:V_attack_success_average}
\Pr_{r\gets \mathcal{R}}\left[\Pr[\top\gets\cV_2(I):I\gets\langle \cP^*_r(1^\secp),\cV_1(1^\secp)\rangle]\ge  \frac{7}{8}+\frac{1}{p(\secp)}\right]> \frac{1}{p(\secp)}
\end{align}
for infinitely many $\secp$. 
Then we prove the following lemma.
\begin{lemma}\label{lem:reduction_B}
There is an oracle-aided PPT algorithm $\B$ that breaks the computational binding property of the commitment scheme if it is given black-box access to $\cP^*_r$ such that
\begin{align} \label{eq:V_attack_success}
\Pr[\top\gets\cV_2(I):I\gets\langle \cP^*_r(1^\secp),\cV_1(1^\secp)\rangle]\ge  \frac{7}{8}+\frac{1}{p(\secp)}
\end{align}
for infinitely many $\secp$. 
\end{lemma}
By combining \cref{eq:V_attack_success_average,lem:reduction_B},  
$\B^{\cP^*_r}$ for random $r\gets \mathcal{R}$ breaks the computational binding property, which is a contradiction.
Thus, we only have to prove \cref{lem:reduction_B} for completing the proof of \cref{thm:soundness}.

\begin{proof}[Proof of \cref{lem:reduction_B}.]
The proof is very similar to that of \cite{ITCS:MorYam23}, which in turn is based on \cite{NatPhys:KMCVY22}. 
Nonetheless, there is a difference that we have to deal with the case where $|X_{0,t}\cap h_0^{-1}(y)|=|X_{1,t}\cap h_1^{-1}(y)|=1$ is not satisfied. Thus, we provide the full proof even though we sometimes repeat the same arguments as those in \cite{ITCS:MorYam23} where some sentences are taken verbatim from there with notational adaptation.

We fix $r$ and an infinite set $\Gamma\subseteq \mathbb{N}$
such that \cref{eq:V_attack_success} holds for all $\secp \in \Gamma$.  
In the following, we simply write 
$\cP^*$ to mean $\cP^*_r$ and 
$\Pr[\top\gets\cV_2]$ to mean $\Pr[\top\gets\cV_2(I):I\gets\langle \cP^*_r(1^\secp),\cV_1(1^\secp)\rangle]$. 
We also often omit to say ``for all $\secp\in \Gamma$'', but whenever we refer to some inequality where $\secp$ appears, we always mean it holds for all $\secp \in \Gamma$.

Define  
\begin{align}\label{eq:def_good}
\good \seteq \left\{(t,h_0,h_1,y):\Pr[\top\gets\cV_2~|~(t,h_0,h_1,y)]\geq  \frac{7}{8}+\frac{1}{2p(\secp)}\right\},
%\good\seteq \{(t,(h_0,h_1),y):|\{x\in X_{0,t}:h_0(x)=y\}|=|\{x\in X_{1,t}:h_1(x)=y\}|=1\} 
\end{align}
where $\Pr[\top\gets\cV_2~|~(t,h_0,h_1,y)]$ denotes $\cV_2$'s acceptance probability conditioned on a fixed $(t,h_0,h_1,y)$,
and define
$
p_{\good}:=\Pr[(t,h_0,h_1,y)\in \good].
$
Note that we have $|X_{0,t}\cap h_0^{-1}(y)|=|X_{1,t}\cap h_1^{-1}(y)|=1$ for all $(t,h_0,h_1,y)\in \good$ since otherwise $\Pr[\top\gets\cV_2~|~(t,h_0,h_1,y)]=  \frac{7}{8}$. 
Then we have
\begin{align}
\Pr[\top \gets \cV_2]&=\Pr[\top \gets \cV_2 \land (t,h_0,h_1,y)\in \good]+\Pr[\top \gets \cV_2 \land (t,h_0,h_1,y)\notin \good]\\
&\le p_\good +(1-p_\good) \cdot \left(\frac{7}{8}+\frac{1}{2p(\secp)}\right). \label{eq:V_acceptance_prob_rewrite}
\end{align}
By \cref{eq:V_attack_success,eq:V_acceptance_prob_rewrite},
we have
\begin{align}
p_\good \geq \frac{1}{2p(\secp)}. \label{eq:good_prob_bound}
\end{align}
% and 
%\begin{align}
%\Pr[\top \gets \cV_2 | (t,(h_0,h_1),y)\in \good] \geq \frac{7}{8}+\frac{1}{p(\secp)}. \label{eq:cond_prob_bound}
%\end{align}
We fix $(t,h_0,h_1,y)\in \good$ until \cref{eq:succ_E}.

For $b\in \bit$, let $x_b\in \bit^\ell$ be the unique element in $X_{b,t}\cap h_b^{-1}(y)$. 
 Note that it is well-defined since we assume $(t,h_0,h_1,y)\in \good$, which implies $|X_{0,t}\cap h_0^{-1}(y)|=|X_{1,t}\cap h_1^{-1}(y)|=1$.

 We define the following probabilities all of which are conditioned on the fixed value of $(t,h_0,h_1,y)$:
\begin{description}
\item[$p_0$:] The probability that $\cV_2$ returns $\top$ conditioned on $v_1=0$.
\item[$p_1$:] The probability that $\cV_2$ returns $\top$ conditioned on $v_1=1$.
\item[$p_{1,0}$:] The probability that $\cV_2$ returns $\top$ conditioned on $v_1=1$ and $v_2=0$.
\item[$p_{1,1}$:] The probability that $\cV_2$ returns $\top$ conditioned on $v_1=1$ and $v_2=1$.
\end{description}
Clearly, we have 
\begin{align} \label{eq:p0p1}
\Pr[\top\gets\cV_2|(t,h_0,h_1,y)]= \frac{p_0+p_1}{2}
\end{align}
and 
\begin{align}\label{eq:p10p11}
p_{1}= \frac{p_{1,0}+p_{1,1}}{2}.
\end{align}
%We remark that we have either $p_0=0$ or $p_0=1$ since there is no randomness for the case of $v_1=0$ once we fix the transcript in Step~\ref{item:RSP} and $\A$'s randomness both of which are included in $(t,(h_0,h_1),y)$.  
%By Inequality \ref{eq:contradiction}, we must have 
%\begin{align}\label{eq:p0_lowerbound}
%    p_0 = 1
%\end{align}
%since otherwise $p_0=0$ and then Equation \ref{eq:p0p1} implies $\Pr[\cV\to\top|(\mathsf{ST}_\A,\{x_0,x_1\})]\le \frac{1}{2}$, which contradicts Inequality \ref{eq:contradiction}. 
By 
$(t,h_0,h_1,y)\in \good$, 
\cref{eq:def_good,eq:p0p1}, and a trivial inequality $p_0,p_1\le 1$, we have 
\begin{align}\label{eq:p0_lowerbound}
    p_0 \geq \frac{3}{4}+\frac{1}{p(\secp)}
\end{align}  
and 
\begin{align}\label{eq:p1_lowerbound}
    p_1 \geq \frac{3}{4}+\frac{1}{p(\secp)}.
\end{align}

Let $\A$ be a classical deterministic polynomial-time algorithm that works as follows:
\begin{enumerate}
    \item $\A$ takes $(t,h_0,h_1,y)$ and $\xi\in \bit^\ell$ as input.
     \item $\A$ runs Step \ref{item:v1} of $\cP^*$ where 
     the transcript of Step \ref{step:coherent_commit}-\ref{step:send_y} is set to be $(t,h_0,h_1,y)$ and the transcript of Step \ref{step:send_v1_and_xi} is set to be $(v_1=1,\xi)$.  Let $d\in \bit^\ell$ be the message sent from $\cP^*$ to $\cV_1$. Note that $\cP^*$'s message is determined by the previous transcript since $\cP^*$ is deterministic. (Recall that $\cP^*$ is a shorthand of $\cP^*_r$ for a fixed randomness $r$.)   
    \item $\A$ runs Step~\ref{item:v2} of $\cP^*$ where the transcript of Step \ref{step:coherent_commit}-\ref{step:send_y} is set to be $(t,h_0,h_1,y)$, the transcript of Step \ref{step:send_v1_and_xi} is set to be $(v_1=1,\xi)$, the transcript of Step \ref{item:v1} is set to be $d$, and the transcript of Step \ref{step:send_v2} is set to be $v_2=0$. Let $\eta_{1,0}$ be the message sent from $\cP^*$ to $\cV_1$. 
    \item $\A$ runs Step~\ref{item:v2} of $\cP^*$ where the transcript of Step \ref{step:coherent_commit}-\ref{step:send_y} is set to be $(t,h_0,h_1,y)$, the transcript of Step \ref{step:send_v1_and_xi} is set to be $(v_1=1,\xi)$, the transcript of Step \ref{item:v1} is set to be $d$, and the transcript of Step \ref{step:send_v2} is set to be $v_2=1$. Let $\eta_{1,1}$ be the message sent from $\cP^*$ to $\cV_1$.  
    \item $\cA$ outputs  $\eta_{1,0}\oplus \eta_{1,1}\oplus 1$. 
\end{enumerate} 
By the union bound, the probability that both $(d,\eta_{1,0})$ and $(d,\eta_{1,1})$ pass the verification is at least 
\begin{align}
    1-(1-p_{1,0})-(1-p_{1,1})=-1+2p_1\geq \frac{1}{2}+\frac{1}{p(\secp)},
\end{align}
where the equation follows from \cref{eq:p10p11} and
the inequality follows from \cref{eq:p1_lowerbound}. 
When this occurs, for each $v_2\in \bit$, we have
\begin{align}
    (\xi\cdot x_0 \neq \xi\cdot x_1) \wedge
    (\eta_{1,v_2}=\xi\cdot x_0),~~
    or~~
    (\xi\cdot x_0= \xi\cdot x_1) \wedge
    (\eta_{1,v_2}=v_2\oplus d\cdot(x_0\oplus x_1)).
    \end{align}
(Remark that the same $d$ is used for both cases of $v_2=0$ and $v_2=1$.) 
In particular, 
if $\xi\cdot x_0 \neq \xi\cdot x_1$ then $\eta_{1,0}=\eta_{1,1}$, and 
if $\xi\cdot x_0 = \xi\cdot x_1$ 
then $\eta_{1,0}= \eta_{1,1}\oplus 1$. 
This implies that
$
    \eta_{1,0}\oplus \eta_{1,1}\oplus 1=\xi\cdot(x_0\oplus x_1).
$
Therefore, we have 
\begin{align}
    \Pr_{\xi\gets \bit^\ell}[\cA((t,h_0,h_1,y),\xi)=\xi\cdot(x_0\oplus x_1)]\geq \frac{1}{2}+\frac{1}{p(\secp)}.
\end{align}
Thus, by the Goldreich-Levin theorem~\cite{GL89}, there is a PPT algorithm $\mathcal{E}$ such that 
\begin{align} \label{eq:succ_E}
    \Pr[\mathcal{E}(t,h_0,h_1,y)=x_0\oplus x_1]\geq\frac{1}{p'(\secp)}
\end{align}
for some polynomial $p'$. (Remark that what we have shown so far is that the above holds for any fixed $(t,h_0,h_1,y)\in \good$.)
%for $\frac{1}{2p(\secp)}$-fraction of $(t,(h_0,h_1),y)$, 
%$(t,(h_0,h_1),y)\in \good$ and the above  hold simultaneously.)

Then, we construct a PPT algorithm $\mathcal{B}$ that breaks the computational binding property of the classical bit commitment scheme as follows:
\begin{enumerate}
    \item $\mathcal{B}$ interacts with the receiver $\cR$ in the same way as $\cP^*$ does in Step \ref{step:coherent_commit} of \cref{protocol:PoQ1},
   and let $t$ be the transcript obtained from the execution. 
    \item $\mathcal{B}$ chooses hash functions $h_0$ and $h_1$ as in Step \ref{step:choose_hash} of \cref{protocol:PoQ1}, and send them to
$\cP^*$.
\item $\cP^*$ returns $y$ as a message of Step \ref{step:send_y} in \cref{protocol:PoQ1}. 
At this point, $(x_0,x_1)$ is implicitly determined if $(t,h_0,h_1,y)\in \good$. 
    \item $\mathcal{B}$ 
    sends $v_1=0$ and $\xi \gets \bit^\ell$ to $\cP^*$ as a message of Step \ref{step:send_v1_and_xi} in \cref{protocol:PoQ1}.
  \item   $\cP^*$ returns $(b',x')$ as a message of the first case of Step \ref{step:v_1_0_or_1} in \cref{protocol:PoQ1}. 
    \item $\mathcal{B}$ runs $\mathcal{E}(t,h_0,h_1,y)$ and let $z$ be the output. 
    \item $\mathcal{B}$ 
    sets $x'_0\seteq x'$ and $x'_1\seteq x'\oplus z$ if $b'=0$, and $x'_0\seteq x'\oplus z$ and $x'_1\seteq x'$ otherwise. For each $b\in \bit$, $\mathcal{B}$ generate a decommitment $\mathsf{dcom}_b$ corresponding to the sender's randomness $x'_b$ and transcript $t$.   
    %the honest committer of the commitment scheme with the randomness $x'_b$ and transcript $t$ and lets $\mathsf{dcom}_b$ be the corresponding decommitment. 
     $\mathcal{B}$ outputs $(0,\mathsf{decom}_0)$ and $(1,\mathsf{decom}_1)$. 
    \end{enumerate}
Recall that we have shown that for any $(t,h_0,h_1,y)\in \good$, \cref{eq:p0_lowerbound,eq:succ_E} hold. Thus, for any $(t,h_0,h_1,y)\in \good$, 
we have 
$
    \Pr[x'=x_{b'}|(t,h_0,h_1,y)]\geq  \frac{3}{4}+\frac{1}{p(\secp)}
$
and 
$
    \Pr[z=x_0\oplus x_1 |(t,h_0,h_1,y)]\geq \frac{1}{p'(\secp)}.
$
Moreover, the two events $x'= x_{b'}$ and $z=x_0\oplus x_1$ are independent once we fix $(t,h_0,h_1,y)$. 
Therefore, for any $(t,h_0,h_1,y)\in \good$, we have 
\begin{align}
    \Pr[
    x'=x_{b'} \wedge
    z=x_0\oplus x_1 |(t,h_0,h_1,y)]\geq \frac{3}{4p'(\secp)}.
\end{align}
Combined with \cref{eq:good_prob_bound}, we have 
\begin{align}\label{eq:C_breaks_binding}
    \Pr[(t,h_0,h_1,y)\in \good\land x'_0=x_0\land x'_1=x_1]\geq \frac{3}{8p(\secp)p'(\secp)}.
\end{align}
By the definition of $x_b$, 
we have $x_b\in X_{b,t}$ for $b\in \bit$. 
Thus, by the perfect correctness of the commitment scheme,  
$\mathsf{decom}_b$ derived from 
$(x_b,t)$ is a valid decommitment.  
Thus,  \cref{eq:C_breaks_binding} implies that $\mathcal{B}$ outputs valid decommitments for both messages $0$ and $1$ with probability at least $\frac{3}{8p(\secp)p'(\secp)}$ (for all $\secp\in \Gamma$).  
This completes the proof of \cref{lem:reduction_B}.
%This contradicts the binding property of the commitment protocol. Therefore, \cref{eq:goal} holds and the proof of soundness is completed. 
\end{proof}
This completes the proof of \cref{thm:soundness}. 
\end{proof}

\ifnum\llncs=0
\input{V2}
\else
\subsection{Computational Power of the Inefficient Verifier}
We observe that $\cV_2$ can be a classical deterministic polynomial-time algorithm querying an $\NP$ oracle. 
\ifnum\cameraready=1
See the full version for details. 
\else
See \Cref{sec:powerofV2} for details. 
\fi
\fi

%% file: V2.tex
\ifnum\llncs=0
\subsection{Computational Power of the Inefficient Verifier}
\label{sec:powerofV2}
\else
\section{Computational Power of the Inefficient Verifier}
\label{sec:powerofV2}
\fi

In this subsection, we show that $\cV_2$ can be a classical deterministic polynomial-time algorithm querying an $\NP$ oracle.
(The inefficient verifier in our construction actually uses randomness when $|X_{0,t}\cap h_0^{-1}(y)|=|X_{1,t}\cap h_1^{-1}(y)|=1$ is not satisfied, 
but the inefficient verifier can be deterministic if we let the first phase verifier append the randomness to the transcript.)

The inefficient parts of $\cV_2$ are
verifying $|X_{0,t}\cap h_0^{-1}(y)|=|X_{1,t}\cap h_1^{-1}(y)|=1$ and
finding $(x_0,x_1)$, where $x_b$ is the single element of $X_{b,t}\cap h_b^{-1}(y)$.
We show that these two tasks can be done in classical deterministic polynomial-time
querying an $\mathbf{NP}$ oracle.
Note that the membership of $x\in X_{b,t}\cap h_b^{-1}(y)$ can be decided in a classical deterministic polynomial time.
Therefore, the decision problem
\begin{itemize}
    \item[Yes:]
    There exists $x\in\bit^\ell$ such that $x\in X_{b,t}\cap h_b^{-1}(y)$.
    \item[No:] 
    For any $x\in\bit^\ell$, $x\notin X_{b,t}\cap h_b^{-1}(y)$.
\end{itemize}
is in $\mathbf{NP}$.

First, $\cV_2$ queries the above decision problem to the $\mathbf{NP}$ oracle for each $b\in\bit$.
If the answer is no for a $b\in\bit$, it means that $|X_{b,t}\cap h_b^{-1}(y)|=0$. In that case,
$\cV_2$ concludes that $|X_{0,t}\cap h_0^{-1}(y)|=|X_{1,t}\cap h_1^{-1}(y)|=1$ is not satisfied.
If the answer is yes for both $b\in\bit$,
$|X_{0,t}\cap h_0^{-1}(y)|\ge 1$ and
$|X_{1,t}\cap h_1^{-1}(y)|\ge1$ are guaranteed.

Then, $\cV_2$ finds an element $x_b\in X_{b,t}\cap h_b^{-1}(y)$ for each $b\in\bit$.
Finding such an element is just an $\mathbf{NP}$ search problem, which can be solved in
classical deterministic polynomial-time by querying the $\mathbf{NP}$ oracle polynomially many times.

$\cV_2$ finally queries the following decision problem to the $\mathbf{NP}$ oracle for each $b\in\bit$.
\begin{itemize}
    \item[Yes:] 
    There exists $x\in\bit^\ell$ such that $x\in (X_{b,t}\cap h_b^{-1}(y))\setminus\{x_b\}$.
    \item[No:] 
    For any $x\in\bit^\ell$, $x\notin (X_{b,t}\cap h_b^{-1}(y))\setminus\{x_b\}$.
\end{itemize}
If the answer is yes for a $b\in\bit$, it means that $|X_{b,t}\cap h_b^{-1}(y)|\ge2$. 
In that case, $\cV_2$ concludes that $|X_{0,t}\cap h_0^{-1}(y)|=|X_{1,t}\cap h_1^{-1}(y)|=1$ is not satisfied.
If the answer is no for both $b\in\bit$,
$\cV_2$ concludes that $|X_{0,t}\cap h_0^{-1}(y)|=|X_{1,t}\cap h_1^{-1}(y)|=1$ is satisfied.

%% file: two_round.tex
\section{Implausibility of Two-Round AI-IV-PoQ}
In this section, we prove \cref{thm:impossible_reduction,thm:oracle_separation_two-round_OWF}.

\subsection{Impossibility of Classical Reduction}
In this subsection, we formally state \cref{thm:impossible_reduction} and prove it. 
%For formally stating \cref{thm:impossible_reduction}, 

First, we define  game-based assumptions. The definition is identical to \emph{falsifiable assumptions} in \cite{STOC:GenWic11} except for the important difference that the challenger can be unbounded-time. 
\begin{definition}[Game-based assumptions]\label{def:game_based_assumption}
A game-based assumption consists of a possibly unbounded-time interactive machine $\mathcal{C}$ (the challenger) and a constant $t\in[0,1)$. On the security parameter $1^\secp$, the challenger $\mathcal{C}(1^\secp)$ interacts with a classical or quantum machine $\A$ (the adversary) over a classical channel and finally outputs a bit $b$. We denote this execution by $b\gets \langle\A(1^\secp),\mathcal{C}(1^\secp) \rangle$.

We say that a game-based assumption $(\mathcal{C},t)$ holds against classical (resp. quantum) adversaries if for any PPT (resp. QPT) adversary $\A$, $|\Pr[1\gets \langle\A(1^\secp),\mathcal{C}(1^\secp) \rangle]-t|\le \negl(\secp).$
\end{definition}
\begin{remark}[Examples]
As explained in \cite{STOC:GenWic11}, the above definition captures a very wide class of assumptions used in cryptography even if we restrict the challenger to be PPT. They include (but not limited to) general assumptions such as security of OWFs, public key encryption, digital signatures, oblivious transfers etc. as well as concrete assumptions such as the hardness of factoring, discrete-logarithm, LWE etc. 
In addition, since we allow the challenger to be unbounded-time, it also captures some non-falsifiable assumptions such as hardness of   
indistinguishability obfuscation \cite{JACM:BGIRSVY12,SICOMP:GGHRSW16} or succinct arguments \cite{SICOMP:Micali00} etc. 
Examples of assumptions that are not captured by the above definition include so called knowledge assumptions \cite{C:Damgaard91,ICALP:CanDak08,TCC:CanDak09,ITCS:BCCT12} and zero-knowledge proofs with non-black-box zero-knowledge \cite{C:HadTan98,FOCS:Barak01}. 
\end{remark}

We clarify meanings of several terms used in the statement of our theorem. 
\begin{definition}[Classical oracle-access to a cheating prover]
Let $\Pi=(\cP,\cV=(\cV_1,\cV_2))$ be a two-round AI-IV-PoQ. 
We say that a (possibly unbounded-time stateless) randomized classical machine $\cP^*$ breaks $s$-soundness of $\Pi$ if there is a polynomial $\poly$ such that 
$\Pr[\top\gets\cV_2(I):I\gets\langle \cP^*(\sigma),\cV_1(\sigma)\rangle]\ge s(|\sigma|)+1/\poly(|\sigma|)$ for all but finitely many $\sigma\in \Sigma$. 
We say that an oracle-aided classical machine $\mathcal{R}$ is given oracle access to $\cP^*$ if it can query an auxiliary input $\sigma$ and the first-round message $m_1$ of $\Pi$ and the oracle returns the second-round message $m_2$ generated by $\cP^*$ with a fresh randomness $r$ in each invocation.  
\end{definition}
\begin{remark}
Since we consider two-round protocols, we can assume that $\cP^*$ is stateless without loss of generality. 
\end{remark}

Then we state the formal version of \cref{thm:impossible_reduction}. 
\begin{theorem}\label{thm:impossible_reduction_formal}
Let $\Pi$ be a two-round AI-IV-PoQ that satisfies $c$-completeness and $s$-soundness where $c(\secp)-s(\secp)\ge 1/\poly(\secp)$ 
and let $(\mathcal{C},t)$ be a game-based assumption. 
Suppose that there is an oracle-aided PPT machine $\mathcal{R}$ (the reduction algorithm) such that for any (possibly unbounded-time stateless) randomized classical machine $\cP^*$ that breaks $s$-soundness of $\Pi$, 
$|\Pr[1\gets \langle\mathcal{R}^{\cP^*}(1^\secp),\mathcal{C}(1^\secp) \rangle]-t|$ is non-negligible. 
Then the game-based assumption $(\mathcal{C},t)$ does not hold against quantum adversaries.  
\end{theorem}
%\begin{remark}
%Restricting $\cP^*$  to be deterministic only makes the above theorem stronger since any reduction that works for all randomized adversaries in particular works for all deterministic adversaries. 
%\end{remark}
\begin{proof}
Let $\widetilde{\cP}^*$ be an unbounded-time randomized classical machine that simulates the honest QPT prover $\cP$. Then $\widetilde{\cP}^*$ breaks $s$-soundness of $\Pi$ because of $c$-soundness and $c(\secp)-s(\secp)\ge 1/\poly(\secp)$. 
Thus, $|\Pr[1\gets \langle\mathcal{R}^{\widetilde{\cP}^*}(1^\secp),\mathcal{C}(1^\secp) \rangle]-t|$ is non-negligible. Since  $\widetilde{\cP}^*$ is simulatable by a QPT machine, $\mathcal{R}^{\widetilde{\cP}^*}$ is simulatable by a QPT machine. Thus, 
$(\mathcal{C},t)$ does not hold against quantum adversaries.
%Since $\Pi$ is two-round, we can regard any deterministic cheating prover $\cP^*$ as a stateless oracle that takes an auxiliary input $\sigma$ and the first round message $m_1$ from the verifier as input and outputs the second round message $m_2$.  
%Let $\ell(\secp)$ be the length of the first-round message when the security parameter is $\secp$. 
%Let $\mathcal{D}$ be a distribution over randomized classical machines with the randomness space $\bit^{\secp}$ resulting from the following process:
%For each $\secp\in \mathbb{N}$, $m_1\in \bit^{\ell(\secp)}$, and $r\in \bit^\secp$, we independently sample $m_2(\secp,m_1,r)$ according to the distribution   
%the distribution of $\cP^*(1^\secp,m_1;r)$ is independent and identical to the distribution of $\cP(1^\secp,m_1)$ where $\cP$ is the honest QPT prover. 
%aaa
%We consider the following trivial unbounded-time prover $\widetilde{\cP}^*$.  
\end{proof}
\begin{remark}\label{rem:why_not_quantum}
One might think that a similar proof extends to the case of quantum reductions. However, we believe that this is non-trivial. 
For example, suppose that a quantum reduction algorithm $\mathcal{R}$ queries  
a uniform superposition $\sum_{m_1}\ket{m_1}$ to the oracle where we omit an auxiliary input for simplicity. If the oracle is a classical randomized cheating prover $\cP^*$, then it should return a state of the form  $\sum_{m_1}\ket{m_1}\ket{\cP^*(m_1;r)}$ for a randomly chosen $r$ where $\cP^*(m_1;r)$ is the second message $m_2$ sent by $\cP^*$ given the first-round message $m_1$ and randomness $r$. 
On the other hand, if we try to simulate the oracle by using the honest QPT prover $\cP$, then the resulting state is of the form $\sum_{m_1,m_2}\ket{m_1}\ket{m_2}\ket{garbage_{m_1,m_2}}$. Due to potential entanglement between the first two registers and the third register, this does not correctly simulate the situation with a classical prover. 
\end{remark}

\subsection{Oracle Separation}
In this subsection, we formally state \cref{thm:oracle_separation_two-round_OWF} and prove it.

First, we define cryptographic primitives. 
The following definition is taken verbatim from \cite{TCC:ReiTreVad04} except for the difference that we consider quantum security. We remark that we restrict primitives themselves to be classical and only allow the adversary (the machine $M$) to be quantum.
\begin{definition}[Cryptographic primitives; quantumly-secure version of {\cite[Definition 2.1]{TCC:ReiTreVad04}}]
A primitive $\primitive$ is a pair $(F_{\primitive}, R_{\primitive})$, where $F_{\primitive}$ is a set of functions $f:\bit^*\rightarrow \bit^*$, and $R_{\primitive}$ is a relation over pairs $(f,M)$ of a function $f\in F_{\primitive}$ and an interactive quantum machine $M$. 
The set $F_{\primitive}$  is required to contain at least one function which is computable by a PPT machine.

A function $f:\bit^*\rightarrow \bit^*$ implements $\primitive$ or is an implementation of $\primitive$ if $f\in F_{\primitive}$. An efficient implementation of $\primitive$ is an implementation of $\primitive$ which is computable by a PPT machine. A machine $M$ $\primitive$-breaks $f\in F_{\primitive}$ if $(f, M)\in R_{\primitive}$. A quantumly-secure implementation of $\primitive$ is an implementation of $\primitive$ such that no QPT machine $\primitive$-breaks $f$.  %The primitive $\primitives$ exists if an efficient and 
\end{definition}
It was pointed out in \cite{AC:BaeBrzFis13} that the above definition was too general and there are subtle logical gaps and counter examples in their claims.
In particular, \cite{TCC:ReiTreVad04} implicitly assumes that 
if two machines $M$ and $M'$ behave identically, then $(f, M)\in R_{\primitive}$ and $(f, M')\in R_{\primitive}$ are equivalent. 
We formalize this property following \cite{AC:BaeBrzFis13} 
where the definitions are taken verbatim from there except for
adaptation to the quantumly-secure setting. 
\begin{definition}[Output distribution~{\cite[Definition B.1 in the ePrint version]{AC:BaeBrzFis13}}]
An interactive (oracle-aided) quantum  Turing machine $M$
together with its oracle defines an output distribution, namely, each fixed finite sequence of inputs
fed to $M$ induces a distribution on the output sequences by considering all random choices of $M$ and 
its oracle. The output distribution of $M$ is defined to be the set of these distributions, indexed by
the finite sequences of input values.
\end{definition} 
\begin{definition}[Semantical cryptographic primitive {\cite[Definition B.2 in the ePrint version]{AC:BaeBrzFis13}}]
A cryptographic primitive $\primitive = (F_\primitive, R_\primitive)$ is called semantical, if for
all $f\in F_\primitive$ and all interactive  (oracle-aided) quantum Turing machines $M$ and $M'$
(including their oracles), it
holds: If $M$ induces the same output distribution as $M'$, then
$(f,M)\in R_\primitive$ if and only if $(f,M')\in R_\primitive$.
\end{definition}
\begin{remark}[Examples]
As explained in \cite{TCC:ReiTreVad04,AC:BaeBrzFis13}, most cryptographic primitives considered in the literature are captured by semantical cryptographic primitives. They include (but not limited to) OWFs, public key encryption, digital signatures, oblivious transfers 
indistinguishability obfuscation etc. 
On the other hand, we note that it does not capture concrete assumptions such as the hardness of factoring, discrete-logarithm, LWE etc. unlike  game-based assumptions defined in \cref{def:game_based_assumption}.
Similarly to game-based assumptions, semantical cryptographic primitives  do not capture knowledge-type assumptions or zero-knowledge proofs (with non-black-box zero-knowledge). %though they are captured by (non-semantical) cryptographic primitives. 
\end{remark}

Next, we define secure implementation relative to oracles following \cite{TCC:ReiTreVad04}.
\begin{definition}[Secure implementation relative to oracles {\cite[Definition 2.2]{TCC:ReiTreVad04}}]
A quantumly-secure implementation of primitive $\primitive$ exists relative to an oracle $O$ 
if there exists an implementation of $f$ of $\primitive$ which is computable by a PPT oracle machine with access to $O$ and such that no QPT oracle machine with access to $O$
$\primitive$-breaks $f$.
\end{definition}
\begin{remark}[Example]
A quantumly-secure implementation of 
OWFs and collision-resistant hash functions exists relative to a random oracle~\cite{SICOMP:BBBV97,Zhandry15}. 
\cite{AC:HosYam20} implicitly proves that quantumly-secure implementation of trapdoor-permutations exist relative to a classical oracle.
We believe that we can prove similar statements for most cryptographic primitives by appropriately defining oracles. 
\end{remark}

Now, we are ready to state the formal version of \cref{thm:oracle_separation_two-round_OWF}.
\begin{theorem}\label{thm:separation_formal}
Suppose that a semantical cryptographic primitive $\primitive=(F_{\primitive}, R_{\primitive} )$ has  a quantumly-secure implementation relative to a classical oracle.
Then there is a randomized classical oracle relative to which two-round AI-IV-PoQ do not exist but a quantumly-secure implementation of $\primitive$ exists. 
\end{theorem}
\begin{remark}
If we assume that a quantumly-secure implementation of a semantical cryptographic primitive $\primitive$ exists in the unrelativized world, then the assumption of the theorem is trivially satisfied relative to a trivial oracle that does nothing. Thus, the above theorem can be understood as a negative result on constructing two-round AI-IV-PoQ from any primitive whose quantumly-secure implementation is believed to exist. 
\end{remark}
\begin{proof}
Let $f$ be a quantumly-secure implementation of $\primitive$ relative to a classical oracle $O$. 
Let $Q^O$ be a randomized oracle that takes a description of an $n$-qubit input quantum circuit $C^O$  with $O$-gates and a classical string $x\in \bit^n$ as input and returns a classical string according to the distribution of $C^O(x)$.   
We prove that two-round AI-IV-PoQ do not exist but $f$ is 
a quantumly-secure implementation relative to oracles $(O,Q^O)$. 

Let $\Pi=(\cP,\cV=(\cV_1,\cV_2))$ be a two-round AI-IV-PoQ that satisfies $s$-soundness 
relative to $(O,Q^O)$. 
Since $Q^O$ is simulatable in QPT with oracle access to $O$, given the auxiliary input $\sigma$, one can generate a description of quantum circuit $C^O$  with $O$-gates that simulates $\cP^{O,Q^O}(\sigma)$ in a classical polynomial time. Then let us consider a classical cheating prover $\cP^*$ relative to the oracles $(O,Q^O)$ that works as follows: Receiving the  auxiliary input $\sigma$ and the first-round message $m_1$ from the external verifier, $\cP^*$ generates the above quantum circuit $C^O$, queries $(C^O,m_1)$ to the oracle $Q^O$ to receive the response $m_2$, and send $m_2$ to the external verifier. 
Clearly, $\cP^*$ passes the verification with the same probability as $\cP$ does. Therefore, $\Pi$ cannot satisfy $c$-soundness for any $c<s$. 
This means that there is no two-round AI-IV-PoQ relative to $(O,Q^O)$.

On the other hand, if $f$ is not a quantumly-secure implementation of $\primitive$ relative to $(O,Q^O)$, then there is a QPT oracle-aided machine $M$ such that $(f,M^{O,Q^O})\in R_\primitive$. Again, since  $Q^O$ is simulatable in QPT with oracle access to $O$, there is a QPT oracle-aided machine $\widetilde{M}$ such that $M^{O,Q^O}$ and $\widetilde{M}^{O}$ induce the same output distributions. Since $\primitive$ is semantical, we have $(f,\widetilde{M}^{O})\in R_\primitive$. This contradicts the assumption that $f$ is a quantumly-secure implementation of $\primitive$ relative to $O$. Therefore, $f$ is a quantumly-secure implementation of $\primitive$ relative to $(O,Q^O)$.

\end{proof}

%% file: QOWFs.tex
\section{Variants of PoQ from QE-OWFs}
\label{sec:QEOWFs}
\begin{definition}[QE-OWFs]
A function $f:\bit^*\to\bit^*$ is a (classically-secure) quantum-evaluation OWF (QE-OWF) if the following two properties
are satisfied.
\begin{itemize}
    \item 
There exists a QPT algorithm $\mathsf{QEval}$ such that
$\Pr[f(x)\gets \mathsf{QEval}(x)]\ge1-2^{-|x|}$ for all $x\in\bit^*$.\footnote{Actually the threshold can be any value larger than 1/2, because the amplification is possible.}
\item
For any PPT adversary $\cA$,
there exists a negligible function $\negl$ such that for any $\secp$,
\begin{equation}
\Pr[f(x')=f(x):x'\gets \cA(1^\secp,f(x)),x\gets\bit^\secp]\le \negl(\secp).
\end{equation}
\end{itemize}
\end{definition}

\begin{remark}
It is usually useless to consider OWFs whose evaluation algorithm is QPT but the security is against PPT adversaries.
However, for our applications, classical security is enough.
We therefore consider the classically-secure QE-OWFs, because it only makes our result stronger.
\end{remark}

\if0
\begin{remark}
Instead of the above definition of the one-wayness,
we could define the one-wayness such that
$\Pr[\cC\to\top]\le\negl(\secp)$
for any PPT adversary $\cA$ in the following security game:
\begin{enumerate}
    \item 
    The challenger $\cC$ samples $x\gets\bit^\secp$.
    \item
    $\cC$ runs $y\gets \mathsf{QEval}(x)$, and sends $r$ to the adversary $\cA$.
    \item
    $\cA$ sends $x'$ to $\cC$.
    \item
    $\cC$ runs $y'\gets\mathsf{QEval}(x)$. If $y=y'$, $\cC$ outputs $\top$. Otherwise, it outputs $\bot$.
\end{enumerate}
It is clear that the one-wayness in these two definitions are equivalent. 
\end{remark}
\fi

Before explaining our construction of variants of PoQ from QE-OWFs,
we point out that QE-OWFs seems to be weaker than classically-secure and classical-evaluation OWFs.
Let $g$ be a classically-secure and classical-evaluation OWF.
Let $L$ be any language in $\mathbf{BQP}$.
From them, we construct the function $f$ as follows:
$f(x,y)\coloneqq L(x)\|g(y)$,
where $L(x)=1$ if $x\in L$ and $L(x)=0$ if $x\notin L$.
Then we have the following lemma.

\begin{lemma}
$f$ is QE-OWFs. Moreover, if $\mathbf{BQP}\neq\mathbf{BPP}$, $f$ cannot be evaluated
in classical polynomial-time.
\end{lemma}

\begin{proof}
First, it is clear that there exists a QPT algorithm $\mathsf{QEval}$ such that
for any $x,y$
\begin{equation}
\Pr[f(x,y)\gets\mathsf{QEval}(x,y)]\ge 1-2^{-|x\|y|}.    
\end{equation}

Second, let us show the one-wayness of $f$.
Assume that it is not one-way. Then, there exists a PPT adversary $\cA$ and a polynomial $p$ such that
\begin{equation}
\frac{1}{2^n}\sum_{x\in\bit^n}
\frac{1}{2^m}\sum_{y\in \bit^m}
\Pr[L(x)=L(x')\wedge g(y)=g(y'):(x',y')\gets \cA(L(x)\|g(y))]\ge \frac{1}{p}.
\end{equation}
From this $\cA$, we can construct a PPT adversary $\cB$ that breaks the one-wayness of $g$ as follows.
\begin{enumerate}
   \item
   On input $g(y)$, sample $x\gets\bit^n$ and $b\gets\bit$.
   \item
   Run $(x',y')\gets \cA(b\|g(y))$.
   \item
   Output $y'$.
\end{enumerate}
The probability that $\cB$ breaks the one-wayness of $g$ is
\begin{align}
&\frac{1}{2^m}\sum_{y\in \bit^m}
\frac{1}{2^n}\sum_{x\in\bit^n}
\frac{1}{2}\sum_{b\in\bit}
\sum_{x',y'}\Pr[(x',y')\gets \cA(b\|g(y))]
\delta_{g(y),g(y')}\\
&\ge
\frac{1}{2}
\frac{1}{2^n}\sum_{x\in\bit^n}
\frac{1}{2^m}\sum_{y\in \bit^m}
\sum_{x',y'}\Pr[(x',y')\gets \cA(L(x)\|g(y))]
\delta_{g(y),g(y')}\\
&\ge
\frac{1}{2}
\frac{1}{2^n}\sum_{x\in\bit^n}
\frac{1}{2^m}\sum_{y\in \bit^m}
\sum_{x',y'}\Pr[(x',y')\gets \cA(L(x)\|g(y))]
\delta_{g(y),g(y')}\delta_{L(x),L(x')}\\
&\ge
\frac{1}{2p},
\end{align}
which is non-negligible.

Finally, it is clear that if there exists a PPT algorithm that computes $f(x,y)$
for any $x,y$ with probability at least $1-2^{-|x\|y|}$, then the algorithm can solve $L$, which contradicts
$\mathbf{BQP}\neq\mathbf{BPP}$.
\end{proof}

Now we show the main result of this section.
\begin{theorem}\label{PoQ_from_QE-OWFs}
If (classically-secure) QE-OWFs exist, then QV-PoQ exist or (classically-secure and classical-evaluation) infinitely-often OWFs exist.
\end{theorem}

\begin{proof}
Let $f$ be a classically-secure QE-OWF. 
From the $f$, we construct a QV-PoQ $(\cP,\cV)$ as follows.
\begin{enumerate}
    \item 
    The verifier $\cV$ chooses $x\gets\bit^\secp$, and sends it to the prover $\cP$.
    \item
    $\cP$ runs $y\gets\mathsf{QEval}(x)$, and sends $y$ to $\cV$.
    \item
    $\cV$ runs $y'\gets\mathsf{QEval}(x)$.
    If $y=y'$, $\cV$ outputs $\top$. Otherwise, $\cV$ outputs $\bot$.
\end{enumerate}
The $1$-completeness is shown as follows.
The probability that $\cV$ accepts with the honest prover is
\begin{align}
\frac{1}{2^\secp}\sum_{x}\sum_y\Pr[y\gets\mathsf{QEval}(x)]^2    
&\ge \frac{1}{2^\secp}\sum_{x}\Pr[f(x)\gets\mathsf{QEval}(x)]^2\\
&\ge \frac{1}{2^\secp}\sum_{x}(1-2^{-\secp})^2\\
&\ge 1-\negl(\secp).
\end{align}

If the soundness is also satisfied, then we have a QV-PoQ.

Assume that the soundness is not satisfied. Then there exists a PPT algorithm $P^*$ such that for any polynomial $\poly$
such that
\begin{align}
\frac{1}{2^\secp}\sum_{x}\sum_y\Pr[y\gets P^*(x)]\Pr[y\gets \mathsf{QEval}(x)]\ge1-\frac{1}{\poly(\secp)}   
\end{align}
for infinitely many $\secp$.
Then we have
for any polynomial $\poly$
\begin{align}
1-\frac{1}{\poly(\secp)}
&\le
\frac{1}{2^\secp}\sum_{x}\sum_y\Pr[y\gets P^*(x)]\Pr[y\gets \mathsf{QEval}(x)]\\
&=\frac{1}{2^\secp}\sum_{x}\Pr[f(x)\gets P^*(x)]\Pr[f(x)\gets \mathsf{QEval}(x)]\\
&+\frac{1}{2^\secp}\sum_{x}\sum_{y\neq f(x)}\Pr[y\gets P^*(x)]\Pr[y\gets \mathsf{QEval}(x)]\\
&\le\frac{1}{2^\secp}\sum_{x}\Pr[f(x)\gets P^*(x)]+\frac{1}{2^\secp}\sum_{x}\sum_{y\neq f(x)}\Pr[y\gets \mathsf{QEval}(x)]\\
&\le\frac{1}{2^\secp}\sum_{x}\Pr[f(x)\gets P^*(x)]+\frac{1}{2^\secp}\sum_{x}2^{-\secp}
\end{align}
for infinitely many $\secp$,
which gives that for any polynomial $\poly$
\begin{equation}
   \frac{1}{2^\secp}\sum_{x}  \Pr[f(x)\gets P^*(x)]
   \ge 1-\frac{1}{\poly(\secp)}
\end{equation}
for infinitely many $\secp$.
If we write the random seed for $P^*$ explicitly, it means that
for any polynomial $\poly$
\begin{equation}
   \frac{1}{2^{\secp+p(\secp)}}\sum_{x\in \bit^\secp}\sum_{r\in \bit^{p(\secp)}}  \delta_{f(x),P^*(x;r)}
   \ge 1-\frac{1}{\poly(\secp)}
   \label{QOWFs_assumption}
\end{equation}
for infinitely many $\secp$,
where $p(\secp)$ is the length of the seed,
and $\delta_{\alpha,\beta}=1$ if $\alpha=\beta$ and it is 0 otherwise.
Define the set
\begin{equation}
G\coloneqq\{(x,r)\in\bit^\secp\times\bit^p:f(x)= P^*(x;r)\}.    
\end{equation}
Then, from \cref{QOWFs_assumption}, we have for any polynomial $\poly$
\begin{equation}
\frac{2^{\secp+p}-|G|}{2^{\secp+p}} \le\frac{1}{\poly(\secp)}
\label{notGisnegl}
\end{equation}
for infinitely many $\secp$.

Define the function $g:(x,r)\to P^*(x;r)$.
We show that it is a classically-secure and classical-evaluation infinitely-often distributionally OWF.
(For the definition of distributionally OWFs, see \cref{def:dOWFs}.)
It is enough because distributionally OWFs imply OWFs (\cref{lem:dOWFs}).
To show it, assume that it is not. Then,
for any polynomial $\poly$
there exists a PPT algorithm $\cA$ 
such that 
\begin{align}
   \Big\|
   \frac{1}{2^{\secp+p}}\sum_{x,r}(x,r)\otimes g(x,r)
   -\frac{1}{2^{\secp+p}}\sum_{x,r} \cA(g(x,r))\otimes g(x,r)\Big\|_1\le\frac{1}{\poly(\secp)}
   \label{gisnotOW}
\end{align}
for infinitely many $\secp$.
Here, we have used quantum notations although everything is classical, because
it is simpler.
Moreover, for the notational simplicity, we omit bras and kets:
$(x,r)$ means $|(x,r)\rangle\langle (x,r)|$, 
$g(x,r)$ means $|g(x,r)\rangle\langle g(x,r)|$, 
and $\cA(g(x,r))$ is the (diagonal) density matrix that represents the classical output distribution of
the algorithm $\cA$ on input $g(x,r)$.

From the algorithm $\cA$, we construct a PPT adversary $\cB$ that
breaks the distributional one-wayness of $f$ as follows:
\begin{enumerate}
    \item 
    On input $f(x)$, sample $r\gets\bit^p$.
    \item
    Run $(x',r')\gets\cA(f(x))$.
    \item
    Output $x'$.
\end{enumerate}
Then for any polynomial $\poly$
\begin{align}
   &\Big\|
   \frac{1}{2^\secp}\sum_{x}x\otimes f(x)
   -\frac{1}{2^\secp}\sum_{x}\cB(f(x))\otimes f(x)\Big\|_1\\
   &=
   \Big\|
   \frac{1}{2^\secp}\sum_{x}x\otimes f(x)
   -\frac{1}{2^\secp}\sum_{x} \frac{1}{2^p}\sum_{r}\mbox{Tr}_R[\cA(f(x))]\otimes f(x)\Big\|_1\label{TrR}\\
    &=
   \Big\|
   \frac{1}{2^\secp}\sum_{x}\sum_{r}\frac{1}{2^p}\mbox{Tr}(r)x\otimes f(x)
   -\frac{1}{2^\secp}\sum_{x} \frac{1}{2^p}\sum_{r}\mbox{Tr}_R[\cA(f(x))]\otimes f(x)\Big\|_1\\
     &\le
   \Big\|
   \frac{1}{2^\secp}\sum_{x}\sum_{r}\frac{1}{2^p}x\otimes r\otimes f(x)
   -\frac{1}{2^\secp}\sum_{x} \frac{1}{2^p}\sum_{r}\cA(f(x))\otimes f(x)\Big\|_1\\
      &=
   \Big\|
   \frac{1}{2^{\secp+p}}\sum_{x,r}[x\otimes r\otimes f(x)
   -\cA(f(x))\otimes f(x)]\Big\|_1\\
       &=
   \Big\|
   \frac{1}{2^{\secp+p}}\sum_{(x,r)\in G}[x\otimes r\otimes f(x)
   -\cA(f(x))\otimes f(x)]\\
   &+\frac{1}{2^{\secp+p}}\sum_{(x,r)\notin G}[x\otimes r\otimes f(x)
   -\cA(f(x))\otimes f(x)]\\
   &+\frac{1}{2^{\secp+p}}\sum_{(x,r)\notin G}[x\otimes r\otimes P^*(x;r)
   -\cA(P^*(x;r))\otimes P^*(x;r)]\\
   &-\frac{1}{2^{\secp+p}}\sum_{(x,r)\notin G}[x\otimes r\otimes P^*(x;r)
   -\cA(P^*(x;r))\otimes P^*(x;r)]
   \Big\|_1\\
        &=
   \Big\|
   \frac{1}{2^{\secp+p}}\sum_{(x,r)\in G}[x\otimes r\otimes P^*(x;r)
   -\cA(P^*(x;r))\otimes P^*(x;r)]\\
   &+\frac{1}{2^{\secp+p}}\sum_{(x,r)\notin G}[x\otimes r\otimes f(x)
   -\cA(f(x))\otimes f(x)]\\
   &+\frac{1}{2^{\secp+p}}\sum_{(x,r)\notin G}[x\otimes r\otimes P^*(x;r)
   -\cA(P^*(x;r))\otimes P^*(x;r)]\\
   &-\frac{1}{2^{\secp+p}}\sum_{(x,r)\notin G}[x\otimes r\otimes P^*(x;r)
   -\cA(P^*(x;r))\otimes P^*(x;r)]
   \Big\|_1\\
         &=
   \Big\|
   \frac{1}{2^{\secp+p}}\sum_{x,r}[x\otimes r\otimes P^*(x;r)
   -\cA(P^*(x;r))\otimes P^*(x;r)]\\
   &+\frac{1}{2^{\secp+p}}\sum_{(x,r)\notin G}[x\otimes r\otimes f(x)
   -\cA(f(x))\otimes f(x)]\\
   &-\frac{1}{2^{\secp+p}}\sum_{(x,r)\notin G}[x\otimes r\otimes P^*(x;r)
   -\cA(P^*(x;r))\otimes P^*(x;r)]
   \Big\|_1\\
          &\le
   \Big\|
   \frac{1}{2^{\secp+p}}\sum_{x,r}[x\otimes r\otimes P^*(x;r)
   -\cA(P^*(x;r))\otimes P^*(x;r)]\Big\|_1\\
   &+\frac{1}{2^{\secp+p}}\sum_{(x,r)\notin G}\Big\|x\otimes r\otimes f(x)
   -\cA(f(x))\otimes f(x)\Big\|_1\\
   &+\frac{1}{2^{\secp+p}}\sum_{(x,r)\notin G}\Big\|x\otimes r\otimes P^*(x;r)
   -\cA(P^*(x;r))\otimes P^*(x;r)\Big\|_1\\
   &\le\frac{1}{\poly(\secp)}
\end{align}
for infinitely $\secp$,
which means that $f$ is not distributional one-way. 
It contradicts the assumption that $f$ is one-way, because one-wayness implies distributionally one-wayness.
In \cref{TrR}, $R$ is the register of the state $\cA(f(x))$ that contains ``the output $r$'' of the algorithm $\cA$.
The last inequality comes from \cref{gisnotOW} and \cref{notGisnegl}.
\end{proof}

\begin{remark}\label{rem:uniform_non-uniform}
It is crucial in the proof of \Cref{PoQ_from_QE-OWFs} that the soundness of QV-PoQ only considers \emph{uniform} PPT adversaries so that we can use the adversary to construct infinitely-often OWFs. If the soundness of QV-PoQ also considers \emph{non-uniform} PPT adversaries, then we should allow the evaluation algorithm of infinitely-often OWFs to be \emph{non-uniform}. 
\end{remark}

%% file: PPneqBPP.tex
%\takashi{Is there a better name for this section?}
\section{Necessity of Assumptions for (AI-/IO-)IV-PoQ}
\label{sec:necessity_assumption}
In this appendix, we prove 
that
the existence of AI-IV-PoQ implies $\mathbf{PP}\neq \mathbf{BPP}$.
(Remember that IV-PoQ implies IO-IV-PoQ, and that IO-IV-PoQ implies AI-IV-PoQ.)

Assume that there is an AI-IV-PoQ.
From it, we can construct another AI-IV-PoQ $(\cP,\cV=(\cV_1,\cV_2))$ where each prover's message is a single bit.
Without loss of generality, we can assume that the first phase of the AI-IV-PoQ runs as follows:
For $j=1,2,...,N$, the prover $\cP$ and the verifier $\cV_1$ repeat the following.
\begin{enumerate}
    \item 
    $\cP$ possesses a state $\ket{\psi_{j-1}}$. It applies
    a unitary $U_j(\alpha_1,\beta_1,...,\alpha_{j-1},\beta_{j-1})$ 
    that depends on the transcript $(\alpha_1,\beta_1,...,\alpha_{j-1},\beta_{j-1})$ 
    obtained so far on the state, and measures
    a qubit in the computational basis
    to obtain the measurement result $\alpha_j\in\bit$.
    $\cP$ sends $\alpha_j\in\bit$ to $\cV_1$.
    The post-measurement state is $\ket{\psi_j}$.
    \item
   $\cV_1$ does some classical computing and sends a bit string $\beta_j$ to $\cP$. 
\end{enumerate}
We claim that
a PPT
prover $\cP^*$ that queries the $\verb|#|\mathbf{P}$ oracle
can exactly sample from $\cP$'s $j$th output probability distribution $\Pr[\alpha_j|\alpha_{j-1},...,\alpha_1]$ for each $j\in[N]$.
This is shown as follows.
\begin{enumerate}
    \item 
    $\cP^*$ computes $\Pr[\alpha_1]=\| (|\alpha_1\rangle\langle \alpha_1|\otimes I)U_1|0...0\rangle\|^2$ for
    each $\alpha_1\in\bit$,
    and samples $\alpha_1$ according to $\Pr[\alpha_1]$.
    Let $\alpha_1^*\in\bit$ be the result of the sampling.
    $\cP^*$ sends $\alpha_1^*$ to $\cV_1$.
    \item
    $\cV_1$ sends $\beta_1^*$ to $\cP^*$.
    \item
    $\cP^*$ and $\cV_1$ repeat the following for $j=2,3,...,N$.
    \begin{enumerate}
        \item 
     $\cP^*$ computes $\Pr[\alpha_j|\alpha_{j-1}^*,...,\alpha_1^*]
     =\frac{\Pr[\alpha_j,\alpha_{j-1}^*,...,\alpha_1^*]}{\Pr[\alpha_{j-1}^*,...,\alpha_1^*]}$ for
    each $\alpha_j\in\bit$,
    and samples $\alpha_j$ according to $\Pr[\alpha_j|\alpha_{j-1}^*,...,\alpha_1^*]$.
    Here, 
    \begin{align}
&    \Pr[\alpha_j,\alpha_{j-1}^*,...,\alpha_1^*]\nonumber\\
    &=\|
    (|\alpha_1^*\rangle\langle \alpha_1^*|\otimes...\otimes|\alpha_{j-1}^*\rangle\langle\alpha_{j-1}^*|\otimes|\alpha_j\rangle\langle\alpha_j|\otimes I)\\
    &\times U_j(\alpha_1^*,\beta_1^*,...,\alpha_{j-1}^*,\beta_{j-1}^*)
    ...U_2(\alpha_1^*,\beta_1^*)U_1
    |0...0\rangle\|^2.
    \end{align}
    Let $\alpha_j^*\in\bit$ be the result of the sampling.
    $\cP^*$ sends $\alpha_j^*$ to $\cV_1$.
    \item
    $\cV_1$ sends $\beta_j^*$ to $\cP^*$.
    \end{enumerate}
\end{enumerate}
    It is known that for any QPT algorithm that outputs $z\in\bit^\ell$, the probability $\Pr[z]$ that it outputs $z$ can be computed
    in classical deterministic polynomial-time by querying the \verb|#|$\mathbf{P}$ oracle~\cite{ForRog99}.
    Therefore, $\cP^*$ can compute $\Pr[\alpha_1]$ and $\Pr[\alpha_j|\alpha_{j-1}^*,...,\alpha_1^*]$ for any $j=2,3,...,N$.
    It is known that a $\verb|#|\mathbf{P}$ function can be computed in classical deterministic polynomial-time
    by querying the $\mathbf{PP}$ oracle. Therefore, if $\mathbf{PP}=\mathbf{BPP}$,
    $\cP^*$ is enough to be a PPT algorithm.

\if0
    \begin{lemma}\label{lem:ForRog}
    Let us consider a QPT algorithm that outputs $(x,y)\in \bit^n\times\bit^m$.
    For any $y\in\bit^m$, $\Pr[x|y]$ can be sampled in a classical deterministic polynomial-time
    querying the \verb|#|$\mathbf{P}$ oracle.
    \end{lemma}
    \begin{proof}
    It is known that for any QPT algorithm that outputs $z\in\bit^\ell$, the probability $\Pr[z]$ that it outputs $z$ can be computed
    in classical deterministic polynomial-time by querying the \verb|#|$\mathbf{P}$ oracle~\cite{ForRog99}.
    The sampling of $\Pr[x|y]$ can be done as follows.
    \begin{enumerate}
       \item
       Compute $\Pr[x_1|y]=\frac{\Pr[x_1,y]}{\Pr[y]}$ for each $x_1\in\bit$,
       and sample $x_1$ from it.
       Let $x_1^*$ be the result.
        \item
       Compute $\Pr[x_2|x_1^*,y]=\frac{\Pr[x_2,x_1^*,y]}{\Pr[x_1^*,y]}$ for each $x_2\in\bit$,
       and sample $x_2$ from it.
       Let $x_2^*$ be the result.
        \item
       Compute $\Pr[x_3|x_2^*,x_1^*,y]=\frac{\Pr[x_3,x_2^*,x_1^*,y]}{\Pr[x_1^*,x_2^*,y]}$ for each $x_3\in\bit$,
       and sample $x_3$ from it.
       Let $x_3^*$ be the result.
       \item
       Repeat until all bits of $x$ are obtained.
    \end{enumerate}
    \end{proof}
    \fi
 

%% file: omitted_preliminaries.tex
\section{Omitted Preliminaries}\label{sec:omitted_pre}

\ifnum\llncs=1
\input{definitions}
\fi

\input{appendix_hashlemma}

\subsection{Auxiliary-Input Collision-Resistance and $\mathbf{PWPP}\nsubseteq \mathbf{FBPP}$}\label{sec:AXCRH}
We prepare several definitions.
\begin{definition}[Auxiliary-input collision-resistant hash functions]
An auxiliary-input collision-resistant hash function is a polynomial-time computable
auxiliary-input function ensemble $\mathcal{H}\coloneqq \{h_\sigma:\bit^{p(|\sigma|)}\to \bit^{q(|\sigma|)}\}_{\sigma\in\bit^*}$
such that 
$q(|\sigma|)<p(|\sigma|)$ and
for every
uniform PPT adversary $\cA$ and polynomial $\poly$, there exists an infinite set $\Lambda\subseteq\bit^*$ such that,
\begin{equation}
\Pr[h_\sigma(x')=h_\sigma(x): (x,x')\gets\cA(\sigma)]
\le\frac{1}{\poly(|\sigma|)}
\end{equation}
for all $\sigma\in \Lambda$.
\end{definition}

A search problem is specified by a relation $R\subseteq \bit^*\times \bit^*$ where one is given $x\in \bit^*$ and asked to find $y\in \bit^*$ such that $(x,y)\in R$. 
We say that a search problem $R$ is many-one reducible to another search problem $S$ %, denoted by $R\le_{m} S$, 
if there are polynomial-time computable functions $f$ and $g$ such that $(f(x),y)\in S$ implies $(x,g(x,y))\in R$.

A search problem $\text{Weak Pigeon}$ is defined as follows: 
$(C,(x,x'))\in \text{Weak Pigeon}$ if and only if $C$ represents a circuit from $\bit^m$ to $\bit^n$ for $m>n$, $x\neq x'$, and $C(x)=C(x')$.

\begin{definition}[\cite{Jerabek16}]
$\mathbf{PWPP}$ is the class of all search problems that are many-one reducible to $\text{Weak Pigeon}$.
\end{definition}
\begin{definition}[\cite{STOC:Aaronson10}]
$\mathbf{FBPP}$ is the class of all search problems $R$ for which there exists a PPT algorithm $\A$ such that for any $x\in \bit^*$,
\begin{align}
\Pr[(x,y)\in R: y \gets \A(x)]= 1-o(1). 
\end{align}
\end{definition}

The following theorem almost directly follows from the definitions. 
\begin{theorem}
There exist auxiliary-input collision-resistant hash functions if and only if $\mathbf{PWPP}\nsubseteq \mathbf{FBPP}$. 
\end{theorem}
\begin{proof}
We first show the ``if'' direction. Toward contradiction, we assume that there is no auxiliary-input collision-resistant hash function. 
We consider a polynomial-time computable auxiliary-input function ensemble $\mathcal{H}=\{h_\sigma\}$ defined as follows: if $\sigma$ represents a circuit $C:\bit^n \rightarrow \bit^m$ such that $n>m$, then $h_\sigma$ takes $x\in \bit^n$ and outputs $C(x)\in \bit^m$. Otherwise, $h_{\sigma}$ is defined arbitrarily, (say, $h_{\sigma}$ is a constant function that always outputs $0$). By our assumption, this is not an auxiliary-input collision-resistant hash function, so there is a PPT algorithm $\A$ and a polynomial $\poly$ such that 
\begin{equation}
\Pr[h_\sigma(x')=h_\sigma(x): (x,x')\gets\cA(\sigma)]
\ge1/\poly(|\sigma|)
\end{equation}
for all but finitely many $\sigma$. 
For those $\sigma$, we can amplify the success probability of $\A$ to $1-o(1)$ by polynomially many times repetition. This gives a PPT algorithm that solves Weak Pigeon with probability $1-o(1)$, which contradicts $\mathbf{PWPP}\nsubseteq \mathbf{FBPP}$.
Therefore, the above $\mathcal{H}$ must be an auxiliary-input collision-resistant hash function. 

Next, we prove the ``only if'' direction. Toward contradiction, we assume $\mathbf{PWPP}\subseteq \mathbf{FBPP}$. This means that there is a PPT algorithm that solves Weak Pigeon with probability $1-o(1)$ on all instances. 
Let $\mathcal{H}=\{h_\sigma\}$ be an arbitrary polynomial-time computable auxiliary-input function ensemble. Then we can use the algorithm that solves Weak Pigeon to find a collision of $h_\sigma$ for all auxiliary inputs $\sigma$ with probability $1-o(1)$. This means that $\mathcal{H}=\{h_\sigma\}$ is not an auxiliary-input collision-resistant hash function. This contradicts the existence of auxiliary-input collision-resistant hash functions. Thus, we have $\mathbf{PWPP}\nsubseteq \mathbf{FBPP}$.
\end{proof}

\subsection{Auxiliary-Input Commitments from $\mathbf{SZK}\nsubseteq \mathbf{BPP}$}\label{sec:commitment_from_worst_case_SZK}
%\takashi{I revised this subsection to consider promise problems rather than languages. This is because, as far as I know, SZK-completeness is only known for promise problems.}
We prove \cref{thm:const_AXcommitment_from_SZK}, i.e., we construct constant-round auxiliary-input statistically-hiding and computationally-binding bit commitments assuming $\mathbf{SZK}\nsubseteq \mathbf{BPP}$. 

First, we recall the definition of instance-dependent commitments.
\begin{definition}[Instance-dependent commitment]
Instance-dependent commitments for a promise problem $(L_{\mathsf{yes}},L_{\mathsf{no}})$ is a family of commitment schemes $\{\Pi_\sigma\}_{\sigma\in \bit^*}$ such that 
\begin{itemize}
\item $\Pi_\sigma$ is statistically hiding if $\sigma\in L_{\mathsf{yes}}$,
\item $\Pi_\sigma$ is statistically binding if $\sigma\in L_{\mathsf{no}}$.
\end{itemize} 
\end{definition}
\begin{theorem}[\cite{TCC:OngVad08}]
\label{thm:instance_dependent_szk}
For any promise problem $(L_{\mathsf{yes}},L_{\mathsf{no}})\in \mathbf{SZK}$, there is a constant-round instance-dependent commitments for $(L_{\mathsf{yes}},L_{\mathsf{no}})$. 
\end{theorem}
Instance-dependent commitments for $\mathbf{SZK}$ directly give auxiliary-input statistically-hiding and computationally-binding commitments if $\mathbf{SZK}\nsubseteq  \mathbf{BPP}$. 
\begin{proof}[Proof of \cref{thm:const_AXcommitment_from_SZK}]
Fix a $\mathbf{SZK}$-complete promise problem $(L_{\mathsf{yes}},L_{\mathsf{no}})$. (For example, we can take the statistical difference problem \cite{JACM:SahVad03}.) 
We regard instance-dependent commitments for $(L_{\mathsf{yes}},L_{\mathsf{no}})$ as auxiliary-input commitments where $\Sigma\seteq L_{\mathsf{yes}}$. Then correctness and statistical hiding as   auxiliary-input commitments directly follows from those of instance-dependent commitments noting that they are statistically hiding if $\sigma \in L_{\mathsf{yes}}$.
Suppose that this scheme does not satisfy computational binding. Then there is a PPT malicious committer $\cS^*$ that finds decommitments to both $0$ and $1$ with an inverse-polynomial probability for all but finitely many $\sigma\in L_{\mathsf{yes}}$. 
We can augment it to construct $\widetilde{\cS}^*$ that works for all  $\sigma\in L_{\mathsf{yes}}$ by letting it find decommitments by brute-force for $\sigma$ for which $\cS^*$ fails to find decommitments. Since there are only finitely many such  $\sigma$, $\widetilde{\cS}^*$ still runs in PPT.  
On the other hand, if $\sigma\in L_{\mathsf{no}}$, $\widetilde{\cS}^*$ cannot find decommitments to $0$ and $1$ except for a negligible probability. Therefore, we can use $\widetilde{\cS}^*$ to distinguish elements of $L_{\mathsf{yes}}$ and $L_{\mathsf{no}}$ with an inverse-polynomial distinguishing gap. Since $(L_{\mathsf{yes}},L_{\mathsf{no}})$ is $\mathbf{SZK}$-complete, this means $\mathbf{SZK}\subseteq \mathbf{BPP}$. This is a contradiction, and thus the above commitment scheme satisfies computational binding. 
\end{proof}

%% file: appendix_hashlemma.tex
\subsection{Poof of \cref{lem:hash2}}
\label{sec:proof_hashing}
Here, we prove \cref{lem:hash2}. As a preparation, we prove the following lemma. 
\begin{lemma}
\label{lem:hash}
Let $\mathcal{H}\coloneqq\{h:\cX\to\cY\}$
be a pairwise-independent hash family such that $|\cX|\ge2$.
Let $S\subseteq \cX$ be a subset of $\cX$. 
For any $y\in\cY$,
\begin{equation}
\Pr_{h\gets\mathcal{H}}[|S\cap h^{-1}(y)|\ge1]
\ge
\frac{|S|}{|\cY|}-
\frac{|S|^2}{2 |\cY|^2}.
\label{ge}
\end{equation}
\end{lemma}

\begin{proof}[Proof of \cref{lem:hash}]
First, if $|S|=0$, \cref{ge} trivially holds.
Second, let us consider the case when $|S|=1$.
In that case,
\begin{align}
\Pr_{h\gets\mathcal{H}}[|S\cap h^{-1}(y)|\ge1]
&=
\frac{1}{|\cY|}\\
&\ge\frac{1}{|\cY|}-
\frac{1}{2 |\cY|^2}\\
&=\frac{|S|}{|\cY|}-
\frac{|S|^2}{2 |\cY|^2},
\end{align}
and therefore \cref{ge} is satisfied.
Here, the first equality comes from the fact that
the probability that the unique element of $S$ is mapped to $y$ is $1/|\cY|$.

Finally, let us consider the case when $|S|\ge2$.
The following argument is based on \cite{stackexchange}.
First, for each $y\in\cY$,
\begin{align}
\sum_{j=1}^{|S|}j \Pr_{h\gets\mathcal{H}}[|S\cap h^{-1}(y)|=j]
&=\underset{h\gets\mathcal{H}}{\mathbb{E}}[|S\cap h^{-1}(y)|]\\
&=\underset{h\gets\mathcal{H}}{\mathbb{E}}[|\{x\in S: h(x)=y\}|]\\
&=\sum_{x\in S}\Pr_{h\gets\mathcal{H}}[h(x)=y]\\
&=\frac{|S|}{|\cY|}.\label{ext1}
\end{align}
Second,
for each $y\in\cY$,
\begin{align}
\sum_{j=1}^{|S|}(j-1) \Pr_{h\gets\mathcal{H}}[|S\cap h^{-1}(y)|=j]
&\le
\sum_{j=1}^{|S|}{j\choose 2} \Pr_{h\gets\mathcal{H}}[|S\cap h^{-1}(y)|=j]\\
&=\underset{h\gets\mathcal{H}}{\mathbb{E}}\left[{|S\cap h^{-1}(y)|\choose 2}\right]\\
&=\underset{h\gets\mathcal{H}}{\mathbb{E}}[|\{\{x,x'\}\subseteq S:x\neq x', h(x)=h(x')=y\}|]\\
&=\sum_{\{x,x'\}\subseteq S, x\neq x'}\Pr_{h\gets\mathcal{H}}[h(x)=h(x')=y]\\
&=\frac{1}{|\cY|^2}{|S|\choose2}\\
&=\frac{|S|(|S|-1)}{2 |\cY|^2}\\
&\le\frac{|S|^2}{2 |\cY|^2}.\label{ext2}
\end{align}
(Note that ${n\choose m}=0$ for any $n<m$.)
By extracting both sides of \cref{ext1} and \cref{ext2},
we have
\begin{eqnarray}
\sum_{j=1}^{|S|}    
\Pr_{h\gets\mathcal{H}}[|S\cap h^{-1}(y)|=j]    
&\ge&
\frac{|S|}{|\cY|}-
\frac{|S|^2}{2 |\cY|^2},
\end{eqnarray}
which shows \cref{lem:hash}.
\end{proof}

Then we prove \cref{lem:hash2}. For the reader's convenience, we restate \cref{lem:hash2} below. 
\begin{lemma}[Restatement of \cref{lem:hash2}]\label{lem:hash2again}
Let $\mathcal{H}\coloneqq\{h:\cX\to\cY\}$
be a pairwise-independent hash family such that $|\cX|\ge2$.
Let $S\subseteq\cX$ be a subset of $\cX$.
For any $y\in\cY$,
\begin{equation}
\Pr_{h\gets\mathcal{H}}[|S\cap h^{-1}(y)|=1]\ge
\frac{|S|}{|\cY|}
-\frac{|S|^2}{|\cY|^2}.
\label{ge2}
\end{equation}
\end{lemma}

\begin{proof}[Proof of \cref{lem:hash2again}]
For any $y\in \cY$,
\begin{align}
\frac{|S|}{|\cY|}
&=
\sum_{j=1}^{|S|}j \Pr_{h\gets\mathcal{H}}[|S\cap h^{-1}(y)|=j]\\
&\ge
\Pr_{h\gets\mathcal{H}}[|S\cap h^{-1}(y)|= 1]
+2\Pr_{h\gets\mathcal{H}}[|S\cap h^{-1}(y)|\ge2]\\
&=
2\Pr_{h\gets\mathcal{H}}[|S\cap h^{-1}(y)|\ge1]
-\Pr_{h\gets\mathcal{H}}[|S\cap h^{-1}(y)|=1]\\
&\ge
\frac{2|S|}{|\cY|}-
\frac{|S|^2}{|\cY|^2}
-\Pr_{h\gets\mathcal{H}}[|S\cap h^{-1}(y)|=1].
\end{align}
Here, 
the first equality is from \cref{ext1},
and in the last inequality we have used \cref{lem:hash}.
Therefore,
\begin{align}
\Pr_{h\gets\mathcal{H}}[|S\cap h^{-1}(y)|=1]
&\ge
\frac{2|S|}{|\cY|}-
\frac{|S|^2}{|\cY|^2}
-\frac{|S|}{|\cY|}\\
&=
\frac{|S|}{|\cY|}
-\frac{|S|^2}{|\cY|^2},
\end{align}
which shows \cref{lem:hash2again} (which is equivalent to \cref{lem:hash2}).
\end{proof}

%% file: app_completeness.tex
\section{Omitted Proofs for the Completeness}
\label{app:completeness}
In this appendix, we show that if $\cP$ starts  
Step \ref{step:send_v1_and_xi}
of \cref{protocol:PoQ1}
with the state $\ket{0}\ket{x_0}+\ket{1}\ket{x_1}$,
the probability that $\cV_2$ outputs $\top$ is
$\frac{1}{2}+\frac{1}{2}\cos^2\frac{\pi}{8}$.
The proof is the same as that in \cite{NatPhys:KMCVY22}, but for the convenience of readers,
we provide it here.

First, the probability that $\cV_1$ chooses $v_1=0$ is 1/2, and in that case,
the honest $\cP$ sends correct $x_0$ or $x_1$, and therefore
$\cV_2$ outputs $\top$ with probability 1.

Second, the probability that $\cV_1$ chooses $v_1=1$ is 1/2.
In that case, \cref{thestate3} is $\ket{\xi\cdot x_0}+(-1)^{d\cdot(x_0\oplus x_1)}\ket{1\oplus (\xi\cdot x_1)}$,
which is one of the BB84 states $\{\ket{0},\ket{1},\ket{+},\ket{-}\}$.
By the straightforward calculations, we can show that
\begin{align}
&\Pr[\eta=0|\ket{0},v_2=0]\\
&\Pr[\eta=1|\ket{1},v_2=0]\\
&\Pr[\eta=0|\ket{+},v_2=0]\\
&\Pr[\eta=1|\ket{-},v_2=0]\\
&\Pr[\eta=0|\ket{0},v_2=1]\\
&\Pr[\eta=1|\ket{1},v_2=1]\\
&\Pr[\eta=1|\ket{+},v_2=1]\\
&\Pr[\eta=0|\ket{-},v_2=1]
\end{align}
are all equal $\cos^2\frac{\pi}{8}$.
Then, we can confirm that the probability that
 \begin{align}
    (\xi\cdot x_0\neq \xi\cdot x_1) \wedge
    (\eta=\xi\cdot x_0),
    \end{align}
    or
    \begin{align}
    (\xi\cdot x_0 = \xi\cdot x_1) \wedge
    (\eta=v_2\oplus d\cdot(x_0\oplus x_1))
    \end{align}
    occurs is $\cos^2\frac{\pi}{8}$.
    In fact, if $\xi\cdot x_0\neq \xi\cdot x_1$, $\cP$'s state is $\ket{0}$ or $\ket{1}$.
    In that case, the probability that $\eta=\xi\cdot x_0$ is 
    $\cos^2\frac{\pi}{8}$ for any $v_2\in\bit$.
    If $\xi\cdot x_0= \xi\cdot x_1$, $\cP$'s state is $\ket{+}$ or $\ket{-}$.
    In that case, the probability that 
    $\eta=v_2\oplus d\cdot(x_0\oplus x_1)$
    is $\cos^2\frac{\pi}{8}$.
    

%% file: app_dOWFs.tex
\section{Distributionally OWFs}
\label{sec:dOWFs}

\begin{definition}[Distributionally OWFs~\cite{FOCS:ImpLub89}]
\label{def:dOWFs}
A polynomial-time-computable function $f:\bit^*\to\bit^*$ is distributionally one-way
if there exists a polynomial $p$ such that for any PPT
algorithm $\cA$ and all sufficiently large $\secp$, the statistical difference between 
$\{(x,f(x))\}_{x\gets\bit^\secp}$ and $\{(\cA(1^\secp,f(x)),f(x))\}_{x\gets\bit^\secp}$ is greater than $1/p(\secp)$.
\end{definition}

The following result is known.
\begin{lemma}[\cite{FOCS:ImpLub89}]
\label{lem:dOWFs}
If distributionally OWFs exist then OWFs exist.
\end{lemma}

%% file: gap_amplification.tex
\section{Proof of \cref{thm:gap_amplification_sequential}}
\label{sec:gap_amplification}

\begin{proof}[Proof of \cref{thm:gap_amplification_sequential}]
We focus on the case of IV-PoQ since it is almost identical for (AI-/IO-)IV-PoQ. 
First, $1$-completeness of $\Pi^{\Nseq}$ immediately follows from Hoeffding's inequality. (Recall that $1$-completeness means that the honest prover's acceptance probability is at least $1-\negl(\secp)$.)  
In the following, we prove that ${\Pi}^{\Nseq}$ satisfies $0$-soundness.
Suppose that it does not satisfy $0$-soundness. 
Then, there is a PPT malicious prover ${\cP}^{\Nseq^*}$ against $\Pi^{\Nseq}$ and a polynomial $p$ such that 
\begin{align}
\Pr[\top\gets\cV_2(I_1,\ldots,I_N):(I_1,\ldots,I_N)\gets\langle {\cP^{\Nseq}}^*(1^\secp),\cV_1^{\Nseq}(1^\secp)\rangle]\ge \frac{1}{p(\secp)}
\end{align}
for infinitely many $\secp$.  
We define random variables $X_1, \ldots, X_N$ as in \cref{protocol:sequential_PoQ}, i.e., $X_i=1$ if 
$\cV_2$ accepts the $i$-th transcript $I_i$ and otherwise $X_i=0$.  
Using this notation, the above inequality can be rewritten as 
\begin{align}\label{eq:sum_X}
\Pr\left[\frac{\sum_{i\in [N]}X_i}{N}\ge \frac{c(\secp)+s(\secp)}{2}\right]\ge \frac{1}{p(\secp)}
\end{align}
for infinitely many $\secp$.
For $i\in [N]$, let $X'_i$ be an independent random variable over $\bit$ such that $\Pr[X'_i=1]=s(\secp)+\frac{1}{2Np(\secp)}$. 
Noting that $\frac{c(\secp)+s(\secp)}{2}-(s(\secp)+\frac{1}{2Np(\secp)})\ge \frac{c(\secp)-s(\secp)}{4}$ for sufficiently large $\secp$,\footnote{ 
This can be seen as follows:
$\frac{c(\secp)+s(\secp)}{2}-(s(\secp)+\frac{1}{2Np(\secp)})=\frac{c(\secp)-s(\secp)}{2}-\frac{1}{2Np(\secp)}\ge \frac{c(\secp)-s(\secp)}{2}-\frac{1}{2N} 
\ge \frac{c(\secp)-s(\secp)}{2}-\frac{(c(\secp)-s(\secp))^2}{2\secp} 
\ge \frac{c(\secp)-s(\secp)}{2}-\frac{c(\secp)-s(\secp)}{4}
=\frac{c(\secp)-s(\secp)}{4}
$
for sufficiently large $\secp$ where the first inequality follows from the assumption $p(\secp)\ge 1$, 
the second inequality follows from $N\geq \frac{\secp}{(c(\secp)-s(\secp))^{2}}$, 
the third inequality follows from $2\secp\ge 4$ and $0\le c(\secp)-s(\secp)\le 1$ for sufficiently large $\secp$.
} 
by Hoeffding's inequality, we have 
\begin{align}\label{eq:sum_X_prime}
\Pr\left[\frac{\sum_{i\in [N]}X'_i}{N}\ge \frac{c(\secp)+s(\secp)}{2}\right]\le  \negl(\secp).
\end{align}
Moreover, we prove below that for any $k\in [N]$, we have 
\begin{align}
&\Pr\left[\frac{\sum_{i=1}^{k}X_i+\sum_{i=k+1}^{N}X'_i}{N}\ge \frac{c(\secp)+s(\secp)}{2}\right]\nonumber\\
&-
\Pr\left[\frac{\sum_{i=1}^{k-1}X_i+\sum_{i=k}^{N}X'_i}{N}\ge \frac{c(\secp)+s(\secp)}{2}\right]
\le
\frac{1}{2Np(\secp)}
\label{eq:sum_X_hybrid}
\end{align}
for sufficiciently large $\secp$. 
By a standard hybrid argument, \cref{eq:sum_X,eq:sum_X_hybrid} imply 
\begin{align}
\Pr\left[\frac{\sum_{i\in [N]}X'_i }{N}\ge \frac{c(\secp)+s(\secp)}{2}\right]\ge \frac{1}{2p(\secp)}.
\end{align}
for infinitely many $\secp$.
This contradicts \cref{eq:sum_X_prime}. 
Thus, we only have to prove \cref{eq:sum_X_hybrid} holds for all $k\in [N]$ and sufficiently large $\secp$. 

\paragraph{\bf Proof of \cref{eq:sum_X_hybrid}.}
Let $\cP^*_k$ be a malicious prover against $\Pi$ that works as follows:
$\cP^*_k$  first simulates the interaction between ${\cP}^{\Nseq^*}$ and $\cV_1^{\Nseq}$ for the first $k-1$ executions of $\Pi$ where it also simulates $\cV_1^{\Nseq}$ by itself in this phase. 
Then $\cP^*_k$ starts interaction with the external verifier $\cV_1$ of $\Pi$ where it works similarly to ${\cP}^{\Nseq^*}$ in the $k$th execution of $\Pi$.  
Note that the randomness of $\cP^*_k$ consists of the randomness $r_P$ for ${\cP}^{\Nseq^*}$ and randomness $r_V^{k-1}$ for $\cV_1^{\Nseq}$ for the first $k-1$ executions of $\Pi$. 
%Note that $r$ and $r'$ perfectly determines the transcripts  $I_1,\ldots,I_{k-1}$ 
%of the first ${k-1}$ executions of $\Pi$.
Therefore, by applying $s$-strong-soundness of $\Pi$ for the above malicious prover $\cP^*_k$, 
%Since we assume that $\Pi$ satisfies $s$-strong-soundness, 
there is a set of $\left(1-\frac{1}{2Np(\secp)}\right)$-fraction of $(r_P,r_V^{k-1})$,
%of randomness $r$ of ${\cP^*}^{\Nseq}$ and transcripts $I_1,\ldots,I_{k-1}$ 
%of the first ${k-1}$ executions of $\Pi$, 
 which we denote by $\mathcal{G}_{k-1}$, such that for all $(r_P,r_V^{k-1})\in \mathcal{G}_{k-1}$, 
\begin{align}
\Pr[X_k=1|r_P,r_V^{k-1}]\le s(\secp)+ \frac{1}{2Np(\secp)} \label{eq:prob_Xk_is_1}
\end{align}
for sufficiently large $\secp$,
where $\Pr[X_k=1|r_P,r_V^{k-1}]$ means the conditional probability that $X_k=1$ occurs conditioned on the fixed values of $(r_P,r_V^{k-1})$. 
On the other hand, because $X'_k$ is independent of $(r_P,r_V^{k-1})$, 
\begin{align}
\Pr[X'_k=1|r_P,r_V^{k-1}]=\Pr[X'_k=1]= s(\secp)+ \frac{1}{2Np(\secp)} \label{eq:prob_Xk_prime_is_1}
\end{align}
for any fixed $(r_P,r_V^{k-1})$.

For notational simplicity, we denote the events that 
\begin{equation}
\frac{\sum_{i=1}^{k}X_i+\sum_{i=k+1}^{N}X'_i}{N}\ge \frac{c(\secp)+s(\secp)}{2} 
\end{equation}
and 
\begin{equation}
\frac{\sum_{i=1}^{k-1}X_i+\sum_{i=k}^{N}X'_i}{N}\ge \frac{c(\secp)+s(\secp)}{2} 
\end{equation}
by $E_k$ and $E_{k-1}$, respectively. 
Then for any $(r_P,r_V^{k-1})\in \mathcal{G}_{k-1}$, 
\begin{align}
&\Pr\left[E_{k} \middle| (r_P,r_V^{k-1})\right] \label{eq_Sk}\\
=&\Pr\left[\frac{X_{k}}{N}\ge \frac{c(\secp)+s(\secp)}{2}-\frac{\sum_{i=1}^{k-1}X_i+\sum_{i=k+1}^{N}X'_i}{N} \middle| (r_P,r_V^{k-1})\right]\\
\le&\Pr\left[\frac{X'_{k}}{N}\ge \frac{c(\secp)+s(\secp)}{2}-\frac{\sum_{i=1}^{k-1}X_i+\sum_{i=k+1}^{N}X'_i}{N} \middle| (r_P,r_V^{k-1})\right] \label{eq:larger_prob_is_zero} \\
=&\Pr\left[E_{k-1} \middle| (r_P,r_V^{k-1})\right] \label{eq_Sk-1}
\end{align}
for sufficiently large $\secp$ where 
\cref{eq:larger_prob_is_zero} follows from \cref{eq:prob_Xk_is_1,eq:prob_Xk_prime_is_1} 
and the observations that 
$X_1,\ldots,X_{k-1}$ are determined by $(r_P,r_V^{k-1})$ and $X'_{k+1},\ldots,X'_{N}$ are independent of 
$X_k$ or $X'_k$.

Then, for sufficiently large $\secp$, 
we have 
\begin{align}
&
\Pr\left[E_{k}\right]-\Pr\left[E_{k-1}\right]\\
=&\left(\Pr\left[E_{k}~\land~ (r_P,r_V^{k-1})\in \mathcal{G}_{k-1}\right]-\Pr\left[E_{k-1}~\land~ (r_P,r_V^{k-1})\in \mathcal{G}_{k-1}\right]\right)\\
&+\left(\Pr\left[E_{k}~\land~ (r_P,r_V^{k-1})\notin \mathcal{G}_{k-1}\right]-\Pr\left[E_{k-1}~\land~ (r_P,r_V^{k-1})\notin \mathcal{G}_{k-1}\right]\right)\\
\le&
\Pr\left[(r_P,r_V^{k-1})\in \mathcal{G}_{k-1}\right]\cdot
\left(\Pr\left[E_{k} \middle| (r_P,r_V^{k-1})\in \mathcal{G}_{k-1}\right]-\Pr\left[E_{k-1} \middle| (r_P,r_V^{k-1})\in \mathcal{G}_{k-1}\right]\right)\\
&+\Pr\left[(r_P,r_V^{k-1})\notin \mathcal{G}_{k-1}\right]\\
\le& \frac{1}{2Np(\secp)}, \label{eq:prob_good_k-1}
\end{align}
where 
\cref{eq:prob_good_k-1} follows from Equations (\ref{eq_Sk})-(\ref{eq_Sk-1}) and 
the fact that $\mathcal{G}_{k-1}$ consists of $\left(1-\frac{1}{2Np(\secp)}\right)$-fraction of $(r_P,r_V^{k-1})$. 
This implies \cref{eq:sum_X_hybrid} and completes the proof of \cref{thm:gap_amplification_sequential}.
\end{proof}